\documentclass[a4paper,11pt,reqno]{amsart}

\usepackage{hyperref}
\usepackage{amsmath,amssymb}
\usepackage{mathtools}
\usepackage{mathrsfs}

\newtheorem{thm}{Theorem}[section]
\newtheorem{lemma}[thm]{Lemma}
\newtheorem{prop}[thm]{Proposition}
\newtheorem{cor}[thm]{Corollary}
\newtheorem{rem}[thm]{Remark}

\numberwithin{equation}{section}

\def\bn{\textbf{n}}
\def\rist{\upharpoonright}
\def\Ga{\Gamma}
\def\Dom{\mbox{Dom}}
\def\Gc{\breve{G}}
\def\Gt{M}
\def\fr{\varnothing}
\def\tv{|\!|\!|}
\def\lap{\Delta}
\def\k{\textbf{k}}
\def\x{\textbf{x}}
\def\y{\textbf{y}}

\def\C{\mathbb{C}}
\def\R{\mathbb{R}}
\def\N{\mathbb{N}}
\def\Z{\mathbb{Z}}
\def\al{\alpha}

\def\ga{\gamma}
\def\eps{\varepsilon}
\def\Fou{\mathfrak{F}}
\def\lam{\lambda}
\def\si{\sigma}
\def\te{\theta}
\def\Om{\Omega}
\def\de{\partial}
\def\kp{\textbf{k}_\parallel}
\def\xp{\textbf{x}_\parallel}
\def\yp{\textbf{y}_\parallel}
\def\kp{\textbf{k}_\parallel}
\def\geqs{\geqslant}
\def\leqs{\leqslant}
\def\lec{\leqs_c}

\def\fp{f_{\parallel}}
\def\l{\left}
\def\r{\right}
\def\la{\langle}
\def\ra{\rangle}
\def\ov{\overline}
\def\DD{\mathcal{D}}
\def\Ma{N}

\def\dd{\displaystyle}
\def\cemb{\hookrightarrow}

\def\q{\textbf{q}}
\def\Bou{\mathfrak{B}}
\def\SV{\mathfrak{S}}
\def\Com{\SV_{\infty}}

\def\sF{\supp F}

\def\supp{\mbox{supp}}
\def\ran{\mbox{ran}}

\def\AA{A}

\def\Rf{R_{\varnothing}}

\DeclareMathOperator*{\slim}{s-lim}

\def\Afree{A_{\fr}}

\def\Hu{\AA^{(1)}}
\def\Hp{\AA^{(\parallel)}}
\def\Ru{R^{(1)}}
\def\Rp{R^{(\parallel)}}
\def\siu{{\si_{1\!}}}
\def\sip{{\si_{\parallel\!}}}
\def\su{{s_{1\!}}}
\def\sp{{s_{\parallel\!}}}
\def\tu{{t_{1\!}}}
\def\tp{{t_{\parallel\!}}}
\def\ru{{r_{1\!}}}
\def\rp{{r_{\parallel\!}}}

\def\lapp{\lap_{\parallel}}
\def\sgn{\mbox{sgn}}
\def\Te{\Theta}
\def\loc{loc}
\def\uu{u_1}
\def\up{u_{\parallel}}

\def\Foup{\mathfrak{F}_{\parallel}}

\def\HH{\mathcal{H}}
\def\Iu{I_{\su}}
\def\Ip{I_{\sp}}
\def\Iup{I_{\su,\sp}}

\def\SI{\mathcal{I}_{\te,\su}}
\def\SJ{\mathcal{J}_{\te,\su}}
\def\SIm{\SI^{(<)}}
\def\SIp{\SI^{(>)}}
\def\fu{\varphi^{(1)}}
\def\fp{\varphi^{(\parallel)}}
\def\ff{\varphi}
\def\KK{K}
\def\h{\textbf{h}}
\def\hp{\textbf{h}_\parallel}
\def\etu{\eta_1}
\def\etp{\eta_{\parallel}}
\def\uno{\mathsf 1}
\def\zero{\mathsf 0}

\def\dag{{*}}

\renewcommand{\epsilon}{\varepsilon}
\renewcommand{\leq}{\leqslant}
\renewcommand{\geq}{\geqslant}

\renewcommand{\Im}{\operatorname{Im}\,}
\renewcommand{\Re}{\operatorname{Re}\,}

\title[]{Scattering from local deformations of a semitransparent plane}
\author{Claudio Cacciapuoti}
\address{DiSAT, Sezione di Matematica, Universit\`a dell'Insubria, via Valleggio 11, I-22100
Como, Italy}
\email{claudio.cacciapuoti@unisubria.it}

\author{Davide Fermi}
\address{Dipartimento di Matematica, Universit\`a di Milano, Via Cesare Saldini 50, I-20133 Milano, Italy }
\email{davide.fermi@unimi.it}

\author{Andrea Posilicano}
\address{DiSAT, Sezione di Matematica, Universit\`a dell'Insubria, via Valleggio 11, I-22100
Como, Italy}
\email{andrea.posilicano@unisubria.it}

\date{}

\begin{document}

\begin{abstract}
%% Text of abstract
We study scattering for the couple $(A_{F},A_{0})$ of Schr\"odinger operators in $L^2(\R^3)$ formally defined as 
$A_0 = -\Delta + \alpha\, \delta_{\pi_0}$ and $A_F = -\Delta + \alpha\, \delta_{\pi_F}$, $\alpha >0$, where $\delta_{\pi_F}$ is the Dirac $\delta$-distribution supported on the deformed plane given by the graph of the compactly supported, Lipschitz continuous function $F:\R^{2}\to\R$ and $\pi_{0}$ is the undeformed plane corresponding to the choice $F\equiv 0$. We provide a Limiting Absorption Principle, show asymptotic completeness of the wave operators and give a representation formula for the corresponding Scattering Matrix $S_{F}(\lambda)$. Moreover we show that, as $F\to 0$, 
$\|S_{F}(\lambda)-\uno\|^{2}_{\Bou(L^{2}({\mathbb S}^{2}))}={\mathcal O}\!\left(\int_{\R^{2}}d\x|F(\x)|^{\gamma}\right)$, $0<\gamma<1$. 

We correct a minor mistake in the computation of the scattering matrix, occurring in the published version of this paper (see J. Math. Anal. Appl. {\bf 473}(1) (2019), pp. 215--257).  The mistake was  in Section \ref{s:sm}, and affected the statement of Corollary \ref{coroll}, specifically, Eq. \eqref{SM_A}. Regrettably the formula for $S_F$ in the Corrigendum J. Math. Anal. Appl. {\bf 482}(1) (2020), 123554, still contains a misprint, the correct expression is the one given here. 
\end{abstract}

\maketitle

\begin{footnotesize}
 \emph{Keywords:}  Scattering theory, point interactions supported by unbounded hypersurfaces, Kre\v{\i}n's resolvent  formulae. 
 
%	47B25;	%	Operator theory;Symmetric and selfadjoint operators (unbounded)
%	47F05.	%	Operator theory; Partial differential operators (should also be assigned at least one other classification number in section 47)
 \end{footnotesize}

\vspace{1cm}

%% main text
\section{Introduction\label{s:intro}}

In this paper we are concerned with self-adjoint Schr\"odinger operators in $L^2(\R^3)$ formally defined as 
\begin{equation}\label{AF}
A_F = -\Delta + \alpha\, \delta_{\pi_F}\,, \qquad \alpha >0\,,
\end{equation}
where $\delta_{\pi_F}$ is the Dirac $\delta$-distribution supported on the surface $\pi_F := \{\x \equiv(x^{1},\xp) \in \R^3 \,|\, x^1 = F(\xp)\}$, where  $F:\R^2\to \R$ is such that:\\ 

$F\in C^{0,1}_{0}(\R^2)$, i.e., $F$ has compact support and is Lipschitz continuous. \\

For the  rigorous definitions of these self-adjoint operators we refer to Section \ref{s:res} below; here we only notice that the functions in their self-adjointness domain have to satisfy semitransparent boundary conditions  of the kind $[\de_\bn f]_{\pi_F} = \al\,f\!\rist_{\pi_F}$, where $[\de_\bn f]_{\pi_F}$  denotes the jump of the normal derivative across the surface $\pi_{F}$.
\par The spectral properties of these operators, mainly presence and estimates on the number of their bound states below the essential spectrum (whenever $\alpha<0$), have been studied in many papers (see, e.g., \cite{CK11, EK03} and references therein for the 2D case, where the support of the perturbation is a deformed line;  \cite{EKL18} for the 3D case; and  the very recent survey \cite{EXNER}). Much less is known regarding their scattering theory: the 2D case  has been studied in the paper \cite{EK05}, where existence and completeness of the wave operators have been provided (see also \cite{EK2018} for similar results in 3D in a setting in which the singular potential is supported on a curve) and the scattering matrix was studied (whenever $\alpha<0$) for the negative part $(-\alpha^{2}/4,0)$ of the spectrum;  in this case the scattering problem is essentially one-dimensional in the sense that it is described by a $2\times 2$ matrix of reflection and transmission amplitudes. Here we are instead concerned with the more involved 3D case and we provide a comprehensive scattering analysis for the case of repulsive interaction $\alpha>0$, from a Limiting Absorption Principle, through existence and asymptotic completeness of the wave operators, to a representation formula for the scattering matrix on the half line $(0,+\infty)$ minus the (possibly empty) discrete subset of embedded eigenvalues. Since the singular potential $\alpha\,\delta_{\pi_{F}}$ does not vanish at infinity, one cannot expect existence of the wave operators for the scattering couple $( A_{F}, \Afree)$ (here and below $\Afree$ denotes the free Schr\"odinger operator) and so, as for the 2D case in \cite{EK05}, we consider here the scattering couple $(A_{F}, A_{0})$, where $A_{0}$ formally corresponds to $A_0 = -\Delta + \alpha\, \delta_{\pi_0}$ and $\delta_{\pi_0}$ is the Dirac $\delta$-distribution supported on the plane $\pi_0 := \{\x \equiv(x^{1},\xp)\in \R^3 \,|\, x^1 = 0\}$, i.e. $A_{0}$ corresponds to $A_{F}$ with $F\equiv 0$. Whereas we follow the same strategy as in \cite{MPS2} and \cite{MP}, due to $A_{0}\not=A_{\fr}$ and to the unboundedness of the obstacles, here, with respect to the results provided in \cite[Subsection 6.4]{MPS2} and \cite[Subsection 5.4]{MP}, more work is needed and the proofs are, for the most, different and more elaborate.\par  
After supplying some preliminary material in Section \ref{s:prelim}, we introduce in Section \ref{s:res} the rigorous definition of the self-adjoint operators $A_{F}$ by providing their resolvent through a Kre\u{\i}n's type formula expressed in terms of the free resolvent (i.e. the resolvent of $\Afree$) (see Theorem \ref{thm:RF}). Let us remark here that the operator $A_{F}$ could be equivalently defined by quadratic form methods (see, e.g., \cite{BEL14, BEKS,ER16} and references therein); however, in order to study the Limiting Absorption Principle (LAP for short) and the Scattering Matrix, one needs a convenient resolvent formula. One more important remark about our use of resolvent formulae is the following: proceeding as in \cite{EK05}, one could try to provide a resolvent formula of $A_{F}$ expressed directly in terms of $A_{0}$; however this would imply the use of trace (evaluation) operators in the operator domain of $A_{0}$, and these are less well-behaved than in the Sobolev space $H^{2}(\R^{3})$, the self-adjointness domain of $A_{\fr}$ (in particular is not clear what should be the correct trace space). Moreover, such an approach can lead, even in the case of a smooth deformation $F$, to singular perturbations of $A_{0}$ supported on not Lipschitz sets. Therefore we prefer to work with the difference of the two Kre\u{\i}n's formulae, one for $A_{F}$ and the other one for $A_{0}$, both expressed in terms of the free resolvent. This suffices, in Section \ref{s:LAP}, for the proof of LAP for $A_{F}$ (see Theorem \ref{ThmLAPF}), obtained  after providing LAP for the operator $A_{0}$ (see Proposition \ref{PropLAP0}) and by building on some abstract results by Renger (see \cite{Ren} and \cite{Ren2}).   \par 
Then, in Section \ref{s:altres}, by a careful analysis of the difference of the two resolvent for $A_{F}$ and $A_{0}$, we obtain a Kre\u{\i}n's type formula for the resolvent of $A_{F}$ which contains only the resolvent of $A_{0}$ and trace maps on the (compact) support of the deformation $F$ (see Theorem \ref{TS}); 
such  resolvent formula resembles the one used in \cite{EK05}.\par
LAP proved in Section \ref{s:LAP} and the formula for the resolvents difference from Section \ref{s:altres}  are the starting point of our analysis of the scattering theory. In Section \ref{s:wo}, using the same approach as in  Section \ref{s:altres},  we prove existence and asymptotic completeness of the wave operators associated to the scattering couple $(A_{F},A_{0})$ (see Theorem \ref{CorWd22}).  In Section \ref{s:sm}, following the same kind of reasoning as in \cite[Sec. 4]{MP} (see also \cite[Rem. 5.7]{MPS}), using this latter resolvent formula, Birman-Yafaev stationary scattering theory and the Kato-Birman invariance principle, we obtain a representation formula for the Scattering Matrix $S_{F}(\lambda)$ of the scattering couple $(A_{F}, A_{0})$ for any energy $\lambda\in(0,+\infty)\backslash\sigma_p^+(A_F)$, where $\sigma_p^+(A_F)$ is the (possibly empty) discrete set of embedded eigenvalues (see Theorem \ref{thmR} and Corollary \ref{coroll}). Finally, see Theorem \ref{teo7}, using such a representation formula, we provide an estimate, in operator norm, on the difference $S_{F}(\lambda)-\uno$; in particular, as the deformation $F\to 0$, one gets 
$$
\|S_{F}(\lambda)-\uno\|^{2}_{\Bou(L^{2}({\mathbb S}^{2}))}={\mathcal O}\!\left(\int_{\R^{2}}d\xp\;|F(\xp)|^{\gamma}\right)\,,\qquad 0<\gamma<1\,.
$$  

The paper is concluded with a technical appendix  containing the proof of the LAP for the self-adjoint operator corresponding to the Laplacian with a $\delta$-interaction in one dimension.

%%%%%%
%SECTION
%%%%%%
\section{Notation and preliminaries \label{s:prelim}}
\noindent
$\bullet$ We denote by $\lap$ the distributional Laplacian on $\DD'(\R^3)$. \\
$\bullet$ We denote by $ \Afree  $ the \emph{free (positive) Laplacian} on $\R^3$; this is the self-adjoint operator
\begin{equation}\label{Lapfree}
 \Afree  : H^2(\R^3) \subset L^2(\R^3) \to L^2(\R^3)\,,\qquad \Afree f:=-\Delta f\;,
\end{equation}
with purely absolutely continuous spectrum $\si( \Afree ) =\si_{ac}( \Afree ) = [0,+\infty)$. The corresponding \emph{free resolvent} operator is 
\begin{equation}\label{Rfree}
\Rf(z) := ( \Afree  - z)^{-1} : L^2(\R^3) \to H^2(\R^3) \qquad z \in \C \backslash [0,+\infty)\,.
\end{equation}
\vskip5pt\noindent
$\bullet$
For any pair $X,Y$ of Banach spaces, we denote the Banach space of bounded linear operators from $X$ to $Y$ with $\Bou(X,Y)$ and set $\Bou(X) = \Bou(X,X)$; the two-sided ideal of compact operators is $\Com(X,Y)$ ($\Com(X) = \Com(X,X)$).
\vskip5pt\noindent
$\bullet$
For any complex  number $z\in\C\backslash[0,+\infty) $ we define its  square root with the branch cut such that $\Im\! \sqrt z >0$. 
\vskip5pt\noindent
$\bullet$
In order to avoid the appearance of cumbersome expressions we shall use the following short-hand notation, for any pair of real numbers $B_1,B_2 \in \R$:
\[
B_1 \lec B_2 \quad \Leftrightarrow \quad B_1 \leqs \mbox{const.}\;B_2\,, \quad\! \mbox{for some finite constant $\mbox{const.} > 0$}\;. 
\]

\subsection{Test functions and distributions\label{ss:testdist}}
For $k \in \{1,2,3,...\}$, we consider the vector space $\DD(\R^k)$ of smooth, compactly supported functions $f:\R^k\to\C$ and equip this space with the well known inductive limit topology. The topological dual space $\DD'(\R^k)$ is the space of Schwartz distributions on $\R^k$. Unless otherwise stated, all derivatives considered in the sequel are to be understood in the sense of distributions.

\subsection{Sobolev spaces\label{subSob}}
Let us now introduce Sobolev spaces on $\R^k$  of $L^2$ type, of both integer and fractional order.	

The Sobolev space of positive integer order $n \in \N$ is the Hilbert space
\begin{equation*}
H^n(\R^k) := \{ f \in \DD'(\R^k)~|~\de^{\al}\! f \in L^2(\R^k) ~~ \mbox{for all $\al \in \N^k$, $|\al| \leqs n$}\}~,
\end{equation*}
endowed with the standard inner product
\begin{equation*}
\la f | g \ra_{H^n(\R^k)} := \sum_{\al \in \N^k,\,|\al| \leqs n}\! \la \de^{\al}\! f| \de^{\al}\! g \ra_{L^2(\R^k)} %\label{innerint}
\end{equation*}
which induces the norm $\|f\|_{H^n(\R^k)} :=\sqrt{\la f|f\ra_{H^n(\R^k)}\!}$\,. Of course, $H^0(\R^k) = L^2(\R^k)$. 

For any $r \in [0,+\infty)\backslash \N$, let us denote with $[r]$ its integer part and put $\rho := r - [r]$. The Sobolev space of fractional order $r$ is
\begin{equation*}\begin{aligned}
H^r(\R^k) := &\l\{\,f\in\mathcal{D}'(\R^k) ~\Big|~ \de^{\al}\! f \in L^2(\R^k) ~~
\mbox{for $\al\in\N^k$, $|\al|\leqs [r]$}\,,~~\mbox{and}\phantom{\int_{\R^k}} \r. \\
& \l.\ \int_{\R^k\times\R^k}\hspace{-0.2cm}d\x\,d\y\;
{|\de^{\al}\!f(\x) - \de^{\al}\!f(\y)|^2 \over |\x-\y|^{k + 2\rho}}\;<+\infty
~~ \mbox{for $\al \in \N^k$, $|\al| = [r]$} \,\r\}\,;
\end{aligned}\end{equation*}
this is also a complex Hilbert space with the inner product
\begin{equation}\begin{aligned}
&\la f|g \ra_{H^{r}(\R^k)} :=  \la f|g \ra_{H^{[r]}(\R^k)} \\
+&
\sum_{\al\in\N^k, |\al| = [r]} \int_{\R^k\times\R^k}\hspace{-0.2cm}
d\x\,d\y\;{\ov{(\de^{\al}\!f(\x) - \de^{\al}\!f(\y))}\,(\de^{\al}\!g(\x)-\de^{\al}\!g(\y))
\over |\x - \y|^{k + 2 \rho}} ~, \label{innerfra}  
\end{aligned}
\end{equation}
inducing the norm $\|f\|_{H^r(\R^k)} := \sqrt{\la f | f \ra_{H^r(\R^k)}\!}$\;.

As well known, for any $r \in [0,+\infty)$, the space $H^r(\R^k)$ coincides with the Banach space obtained by taking the closure of $\DD(\R^k)$ with respect to the norm $\|\cdot\|_{H^r(\R^k)}$. We denote with $H^{-r}(\R^k)$ its topological dual; so, each $f \in H^{-r}(\R^k)$ is a continuous linear form on $H^r(\R^k)$ and, by restriction to $\DD(\R^k)$, it can be identified with a distribution on $\R^k$. 

In some of the forthcoming proofs, we will also use the well-known equivalent norm on the Sobolev spaces $H^r(\R^k)$ ($r \in \R$) defined by standard functional calculus in terms of real powers of the operator $(\mathbf{1}-\lap_k)$ ($\lap_k$ is the free Laplacian on $\R^k$) as follows:
\begin{equation}
\tv f\tv_{H^r(\R^k)} := \|(\mathbf{1}-\lap_k)^{r/2}f\|_{L^2(\R^k)}\;. \label{Equinorm}
\end{equation}
Of course, using the distributional Fourier transform $\Fou : \DD'(\R^k) \to \DD'(\R^k)$, normalized so as to be unitary with respect to the inner product of $L^2(\R^k)$, the above norm can be re-expressed as
\begin{equation}
\tv f\tv_{H^r(\R^k)} = \|(1+ |\k|^2)^{r/2}\,\Fou f\|_{L^2(\R^k)}\;. \label{EquinormFou}
\end{equation}

It is a well-known fact that the norm \eqref{Equinorm} is equivalent to the norm introduced previously on $H^r(\R^k)$ for any $r \in \R$. More explicitly, this means that for any $f \in H^r(\R^k)$ there holds true the following chain of inequalities
\begin{equation}
\tv f\tv_{H^{r}(\R^k)}\; \lec \; \|f\|_{H^{r}(\R^k)}\; \lec\; \tv f\tv_{H^{r}(\R^k)}\;. \label{EquinormIne}
\end{equation}

\subsection{Sobolev spaces on domains}
Whenever $\Omega\subseteq \R^{k}$ is an open set with a Lipschitz boundary, we define $$H^{s}(\Omega):=\{f\in \mathcal{D}'(\Omega) ~|~ \exists \tilde f\in H^{s}(\R^{k})\ \text{such that}\ \tilde f\!\rist\!\Omega=f\}
$$ 
and 
$$
H^{s}_{\overline\Omega}(\R^k):=\{f\in H^{s}(\R^{k}) ~|~ \text{supp}f\subseteq\overline\Omega\}\,.
$$
\subsection{Basic results on Sobolev spaces}

\begin{rem}\label{r:embedding}
In what follows we shall often use without further notice the continuous Sobolev embedding $H^{s}(\R^k) \cemb  H^{r}(\R^k)$ for $s\geq r$; in particular to infer that given an operator  $O : H^{r}(\R^{k}) \to H^{r'}(\R^{k'})$ then:  
\begin{itemize}
\item[$i)$] If $O\in \Bou(H^{r}(\R^{k}),H^{r'}(\R^{k'}))$ then $O \in \Bou(H^{s}(\R^{k}),H^{s'}(\R^{k'}))$ for all $r\leq s$, and $s' \leq r'$; 
\item[$ii)$] If $O\in \SV_p(H^{r}(\R^{k}),H^{r'}(\R^{k'}))$ then $O \in \SV_p(H^{s}(\R^{k}),H^{s'}(\R^{k'}))$ for all $r\leq s$,  $s' \leq r'$, and $1\leq p \leq +\infty$.
\end{itemize}
A similar remark holds true when dealing with the embedding between weighted Sobolev spaces, see Section \ref{ss:wss}.  
\end{rem}

\begin{lemma}\label{LemTenPro} Let $h,k \in \{1,2,3,...\}$ and consider the vector spaces $\DD(\R^h) \otimes \DD(\R^k)$ $($where $\otimes$ indicates the algebraic tensor product$)$, $\DD(\R^h \!\times \R^k)$ and $H^r(\R^h \!\times \R^k)$. Then, the following inclusions hold true for all $r \in \R$:
\begin{equation*}
\DD(\R^h) \otimes \DD(\R^k) \subset \DD(\R^h \!\times \R^k) \subset H^r(\R^h \!\times \R^k)\;.
\end{equation*}
Moreover, the sets on the l.h.s. of the above inclusions are dense subspaces of the topological vector spaces on the r.h.s., assuming the latter are endowed with their natural topologies.
\end{lemma}
\proof The thesis is a just restatement of known results. For the proof of the first inclusion, see, e.g., \cite[p. 74, Prop. 6.1]{Blan} or \cite[p. 409, Th. 39.2]{Tre}; for the proof of the second inclusion, see \cite[p. 182, Th. 13.2]{Blan} or \cite[p. 107, Prop. 13.1]{Tre2}.
\endproof

\begin{lemma}\label{LemFact} Let $h,k \in \{1,2,3,...\}$ and $r \geqs 0$. Then, there hold the following continuous embeddings:
\begin{equation}\label{EstFT1}
H^r(\R^{h+k}) \cemb H^{r_h}(\R^h) \otimes H^{r_k}(\R^k) \qquad \mbox{for all\, $r_k + r_h \leqs r$}\,;
\end{equation}
\begin{equation}\label{EstFT2}
H^{r_h}(\R^h) \otimes H^{r_k}(\R^k) \cemb H^{\min(r_h,r_k)}(\R^{h+k}) \qquad \mbox{for all\, $r_k,r_h \geqs 0$}\,.
\end{equation}
\end{lemma}

\proof Consider the norms introduced in Eq. \eqref{Equinorm}. To prove both statements \eqref{EstFT1} and \eqref{EstFT2} it suffices to recall that these norms are equivalent to the standard norms on $H^{r}(\R^k)$ and to notice that they fulfil the elementary inequalities reported hereafter for any factorized function of the form $f = u_h \otimes u_k \in \DD(\R^h) \otimes \DD(\R^k)$.

On the one hand, noting that $(1 + |\q_h|^2 + |\q_k|^2)^r \geqs (1 + |\q_h|^2 + |\q_k|^2)^{r_h} (1 + |\q_h|^2 + |\q_k|^2)^{r_k} \geqs (1 + |\q_h|^2)^{r_h} (1 + |\q_k|^2)^{r_k}$ for $r \geqs 0$ and $r_h + r_k \leqs r$, we have
\[\begin{aligned}
&\tv f\tv_{H^r(\R^{h+k})}^2 
=  \int_{\R^h \times \R^k}d\q_h d\q_k\;(1+|\q_h|^2+|\q_k|^2)^r\, |(\Fou u_h)(\q_h)|^2\, |(\Fou u_k)(\q_k)|^2 \\
\geqs & \int_{\R^h \times \R^k} d\q_h\,d\q_k\;(1+|\q_h|^2)^{r_h} (1+|\q_k|^2)^{r_k}\, |(\Fou u_h)(\q_h)|^2\, |(\Fou u_k)(\q_k)|^2 \\
= & \;\tv u_h \tv_{H^{r_h}(\R^h)}^2\,\tv u_k\tv_{H^{r_k}(\R^k)}^2 = \tv u_h \!\otimes\! u_k\tv_{H^{r_h}(\R^h) \otimes H^{r_k}(\R^k)}^2\,.
\end{aligned}
\]

On the other hand, for any $r,r_h,r_k \geqs 0$ one can use the trivial estimates $(1 + |\q_h|^2 + |\q_k|^2) \leqs (1 + |\q_h|^2) (1 + |\q_k|^2)$ and $\min(r_h,r_k) \leqs r_h$, $\min(r_h,r_k) \leqs r_k$ to obtain the following chain of inequalities:
\[
\begin{aligned}
&\tv f\tv_{H^{\min(r_h,r_k)}(\R^{h+k})}^2 \\ 
=& \int_{\R^h \times \R^k}d\q_h\,d\q_k\;(1+|\q_h|^2+|\q_k|^2)^{\min(r_h,r_k)}\, |(\Fou u_h)(\q_h)|^2\, |(\Fou u_k)(\q_k)|^2  \\
\leqs  & \int_{\R^h \times \R^k} d\q_h\,d\q_k\;(1+|\q_h|^2)^{r_h} (1+|\q_k|^2)^{r_k}\, |(\Fou u_h)(\q_h)|^2\, |(\Fou u_k)(\q_k)|^2  \\ 
= & \;\tv u_h\tv_{H^{r_h}(\R^h)}^2\,\tv u_k\tv_{H^{r_k}(\R^k)}^2 = \tv u_h \!\otimes\! u_k\tv_{H^{r_h}(\R^h) \otimes H^{r_k}(\R^k)}^2\,.
\end{aligned}
\]

The above bounds suffice to infer the claims stated in Eq.s \eqref{EstFT1} and \eqref{EstFT2}, by standard density arguments (see also the previous Lemma \ref{LemTenPro}).
\endproof

\subsection{Sobolev spaces on the boundary and trace operators\label{ss:trace}}
Let us consider the space $\R^3$ and indicate with boldface letters $\x = (x^1,x^2,x^3)$ a set of Cartesian coordinates on it. We consider the plane
\begin{equation*}
\pi_0 := \{\x \in \R^3 \,|\, x^1 = 0\}\;,
\end{equation*}
and refer to it as the ``\emph{flat}'' or the ``\emph{non-deformed}'' plane. We write $\xp = (x^2,x^3)$ for the coordinates induced by the set of coordinates $\x$ on $\pi_0$, which allow to naturally identify the latter with $\R^{2}$.

For  $F\in C_0^{0,1}(\R^2)$, as in our assumptions,  we write $\sF$ for the support of $F$ and consider the surface
\begin{equation*}
\pi_F := \{\x \in \R^3 \,|\, x^1 = F(\xp)\}\;,
\end{equation*}
which is referred to as the ``deformed'' plane. Needless to say, the two planes $\pi_0$ and $\pi_F$ coincide when $F = 0$. Besides, let us remark that, similarly to $\pi_0$, the plane $\pi_F$ can also be identified with $\R^{2}$ considering the change of coordinates 
\begin{equation*}
(x^1,\xp) \mapsto (y^1,\yp) := (x^1-F(\xp),\xp)\;. \label{change}
\end{equation*}
Correspondingly, one has an isomorphism $I_F$ of $H^r(\R^3)$ into itself for any order $r \in \R$ such that $|r| \leqs 1$ (see \cite[Sec. 1.3.3]{Gris} or \cite[Ch. 3]{McL}), defined by $I_F f(x^1,\xp) = f(x^1+F(\xp),\xp)$. 

To proceed, let us consider the map $\tau_0: \DD(\R^3) \to \DD(\R^{2})$ defined by
\begin{equation*}
(\tau_0 f)(\xp) := f(0,\xp)\;,
\end{equation*}
i.e. the evaluation of smooth functions on the plane $\pi_0$. As well known \cite{LioMag}, this map can be uniquely extended to a surjective, continuous operator $\tau_0 : H^{r+1/2}(\R^3) \to H^{r}(\R^{2})$ for all $r > 0$. 

We shall also use the map $\tau_F : \DD(\R^3) \to \DD(\R^{2})$ defined by
\begin{equation*}
(\tau_F f)(\xp) := f(F(\xp),\xp)\;,
\end{equation*}
i.e. the evaluation of smooth functions on the surface $\pi_F$.
\begin{rem}\label{r:taubonded}
Noticing that $(\tau_F f)(\xp) $ $= (\tau_0 I_Ff )(\xp) $, we infer that this map  can be uniquely extended to a surjective, continuous operator 
\begin{equation*}\label{taubounded}
\tau_F \in \Bou(H^{r+1/2}(\R^3),H^{r}(\R^{2})) \qquad \mbox{for all $r \in (0,1/2]$}\;;
\end{equation*}
see also, e.g., \cite[Th. 3.37]{McL}. 
\end{rem}
The operators $\tau_0,\tau_F$ are commonly referred to as \emph{traces} on the planes $\pi_0,\pi_F$.
Considering then the $C^{0,1}$ open domains 
$$
\Omega_{F}^{\pm}:=\{\x \in \R^3 \,|\, x^1> \pm F(\xp)\}\,,
$$
the trace operator $\tau_{F}$ can be extended to the larger spaces 
$$
H^{r}(\R^3 \backslash \pi_{F}):=H^{r}(\Omega_{F}^{-})\oplus H^{r}(\Omega_{F}^{+}) \qquad \mbox{for $r \in (0,1)$}
$$
by setting 
\begin{equation*}\label{lat-trace}
\tau_{F}:H^{r+1/2}(\Omega^{-}_{F})\oplus H^{r+1/2}(\Omega^{+}_{F}) \to H^{r}(\R^{2})\,,
\quad
\tau_{F}(f_{-}\oplus f_{+}):=\frac12\,(\tau_{F}^{-}f_{-}+\tau_{F}^{+}f_{+})\,,\end{equation*}\noindent
where 
$$
\tau^{\pm}_{F} \in \Bou(H^{r+1/2}(\Omega^{\pm}_{F}), H^{r}(\R^{2})) \qquad \mbox{for $r \in (0,1)$}
$$
are the lateral traces defined as the unique bounded extensions of the evaluation maps 
$$
\tau^{\pm}_{F}f(\xp):=\lim _{\epsilon\downarrow 0}f(F(\xp)\pm\epsilon,\xp)\,,\qquad  f\in H^{r+1/2}(\Omega^{\pm}_{F})\cap C(\overline{\Omega^{\pm}_{F}})\,.
$$
In terms of the above lateral traces, for later convenience we introduce the notations
\begin{equation}
f\!\rist_{\pi_F^\pm} \;:= \tau^{\pm}_F f\,, \qquad f\!\rist_{\pi_F} \;:= \tau_F f\,, \qquad [f\,]_{\pi_F} := \tau^{+}_F f - \,\tau^{-}_F f\,, \label{jump}
\end{equation}

\subsection{Weighted Sobolev spaces\label{ss:wss}}
For any $k \in \N$ and $w \in L^1_{\loc}(\R^k,[0,+\infty))$, we consider the weighted Sobolev-type spaces of integer order $n \in \{0,1,2,...\}$ defined as
\begin{equation} \begin{aligned}
& H^n_{w}(\R^k) := \big\{f \in  \DD'(\R^k)\;\big|\; \|f\|_{H^n_{w}(\R^k)}^2 < +\infty \big\}\;, \\
& \|f\|_{H^n_{w}(\R^k)}^2 := \sum_{|\al| \leqs n}\int_{\R^k}\!d\x\;w(\x)\,\big|(\de^\al f)(\x)\big|^2\;; \label{normHkw}
\end{aligned}
\end{equation}
in particular, for $n = 0$ we set
\begin{equation*}\begin{aligned}
& H^0_{w}(\R^k) \equiv L^2_{w}(\R^k) := \big\{f \in \DD'(\R^k)\;\big|\; \|f\|_{L^2_{w}(\R^k)}^2 < +\infty \big\}\;,  \\ 
&\|f\|_{L^2_{w}(\R^k)}^2 := \int_{\R^k}\!d\x\;w(\x)\,|f(\x)|^2\,. \label{normL2w}
\end{aligned} \end{equation*}
For any $\te \in (0,1)$, we define the analogous fractional order spaces by complex interpolation putting
\begin{equation}
H^{n+\te}_{w}(\R^k) := [H^n_{w}(\R^k),H^{n+1}_{w}(\R^k)]_{\te}\;. \label{HkwInt}
\end{equation}

In passing, let us remark that the non-weighted fractional Sobolev spaces introduced in subsection \ref{subSob} could be equivalently characterized by complex interpolation, via a relation analogous to \eqref{HkwInt}. To be more precise, there holds $H^{n+\te}(\R^k) = [H^n(\R^k),H^{n+1}(\R^k)]_{\te}$; of course, this identity must be understood in the sense that the usual topology on $H^{n+\te}(\R^k)$ descending from the inner product \eqref{innerfra} is equivalent to (though, different from) the natural interpolation topology on $[H^n(\R^k),H^{n+1}(\R^k)]_{\te}$.

Now, let us assume that $1/w \in L^1_{\loc}(\R^k,[0,+\infty))$ is an admissible weight as well; then, for any $r \geqs 0$ we introduce the weighted space of negative order $-r$ in terms of the standard $L^2$-duality setting
\begin{equation}
H^{-r}_{w}(\R^k) := \big(H^r_{1/w}(\R^k)\big)'\;. \label{HrwDual}
\end{equation}
Taking into account the latter position and \cite[p. 98, Cor. 4.5.2]{BeLo}, it can be readily inferred that a relation analogous to Eq. \eqref{HkwInt} does indeed hold true for any $n \in \Z$ and all $\te \in (0,1)$. 
\vspace{0.15cm}

Our analysis involves, in particular, the weights $w_{\su}(x^1) := (1+|x^1|^2)^{\su}$ on $\R$, $w_{\sp}(\xp) := (1+|\xp|^2)^{\sp}$ on $\R^{2}$ and their tensor product $w_{\su,\sp}(x^1\!,\xp) := w_{\su}(x^1)\,w_{\sp}(\xp)$ on $\R^3$, for suitable $\su,\sp \in \R$. For any $r \in \R$, we indicate the corresponding weighted spaces with the short-hand notations
\begin{equation*}
H^r_{\su}(\R) \equiv H^r_{w_{\su}\!}(\R)\,, \qquad H^r_{\sp}(\R^{2}) \equiv H^r_{w_{\sp}}(\R^{2})\,, \qquad H^r_{\su,\sp}(\R^3) \equiv H^r_{w_{\su,\sp}}(\R^3) \;; \label{HrwTen}
\end{equation*}
moreover, noting the elementary identity $1/w_{\su} = w_{-\su}$ and the related analogues for $w_{\sp}$ and $w_{\su,\sp}$, from Eq. \eqref{HrwDual} we infer the duality relations
\begin{equation}\begin{aligned}
H^{-r}_{-\su}(\R) = &\, \big(H^{r}_{\su}(\R)\big)'\;, \qquad
H^{-r}_{-\sp}(\R^2) = \big(H^{r}_{\sp}(\R^2)\big)'\;, \\ 
& H^{-r}_{-\su,-\sp}(\R^3) = \big(H^{r}_{\su,\sp}(\R^3)\big)'\,.  \label{HrsDual}
\end{aligned}
\end{equation}

Before proceeding, let us mention that all the weighted spaces considered above are indeed Hilbert spaces, endowed with the corresponding natural inner products. For any given pair $\HH_1,\HH_2$ of these spaces, we regard the tensor product $\HH_1 \otimes \HH_2$ as a Hilbert space itself, equipped with the usual inner product defined on factorized elements and extended to the whole space by linearity.
\vspace{0.1cm}

In the forthcoming Lemma \ref{lemHW} we collect a number of results providing a more explicit characterization of the weighted spaces described previously; these results will be employed systematically in the derivation of the subsequent developments, in particular in Section \ref{s:LAP} for the proof of the LAP.

\begin{lemma}\label{lemHW}
The following statements $i)-iv)$ hold true. \\
$i)$ Assume that $\su,\sp \geqs 0$; then, for all $r \in \R$ the following embeddings define continuous maps:
\begin{equation*}
H^r_{\su}(\R) \cemb H^r(\R)\;, \qquad H^r_{\sp}(\R^2) \cemb H^r(\R^2)\;, \qquad H^r_{\su,\sp}(\R^r) \cemb H^r(\R^3)\;.
\end{equation*}
$ii)$ Let $\su,\sp \in \R$ and consider the multiplication operators 
\begin{equation}\begin{aligned}
\Iu : \DD'(\R) \to \DD'(\R)\,, \;&\; \uu \mapsto w_{\su}^{1/2} \uu\;, \qquad
\Ip : \DD'(\R^2) \to \DD'(\R^2)\,, \;\; \up \mapsto w_{\sp}^{1/2} \up\;, \\
&\Iup : \DD'(\R^3) \to \DD'(\R^3)\,, \;\; f \mapsto w_{\su,\sp}^{1/2}\,f\;. \label{IsDist}
\end{aligned}\end{equation}\noindent
For any $r \in \R$, these operators define by restriction the isomorphism of Banach spaces
\begin{equation*}
\Iu : H^r_{\su}(\R) \to H^r(\R),\;\;\;
\Ip : H^r_{\sp}(\R^2) \to H^r(\R^2),\;\;\;
\Iup : H^r_{\su,\sp}(\R^3) \to H^r(\R^3). \label{IsIso}
\end{equation*}\noindent
$iii)$ Let $\su,\sp \in \R$; then, for all $r,r' \in \R$ with $r \geqs r'$ the following embeddings define continuous maps:
\begin{equation*}
H^r_{\su}(\R) \cemb H^{r'}_{\su}(\R)\;, \qquad H^r_{\sp}(\R^2) \cemb H^{r'}_{\sp}(\R^2)\;, \qquad H^r_{\su,\sp}(\R^3) \cemb H^{r'}_{\su,\sp}(\R^3)\;.
\end{equation*}
$iv)$ Let $\su,\sp \in \R$; then, for all $\ru,\rp \geqs 0$ the following embedding defines a continuous map:
\begin{equation*}
H^{\ru}_{\su}(\R) \otimes H^{\rp}_{\sp}(\R^2) \cemb H^{\min(\ru,\rp)}_{\su,\sp}\!(\R^3)\;.
\end{equation*}
\end{lemma}
\proof In the following we discuss separately the proofs of items i)$-$iv). Concerning items i)$-$iii), as examples we derive the corresponding statements involving the space $H^{\ru}_{\su}(\R)$; the analogous claims regarding $H^{r}_{\sp}(\R^2)$ and $H^{r}_{\su,\sp}(\R^3)$ can be inferred by similar arguments.\\
$i)$ For $\su \geqs 0$ and $r = n \in \{0,1,2,...\}$, it can be easily checked by direct inspection that $\|\uu\|^2_{H^{n}_{-\su}(\R)} \leqs \|\uu\|^2_{H^{n}(\R)} \leqs \|\uu\|^2_{H^{n}_{\su}(\R)}$, which proves that $H^{n}_{\su}(\R) \cemb H^{n}(\R)$ and $H^{n}(\R) \cemb H^{n}_{-\su}(\R)$. Recalling our position \eqref{HkwInt} and its dual analogue, from here we infer by interpolation that $H^{r}_{\su}(\R) \cemb H^{r}(\R)$ and $H^{r}(\R) \cemb H^{r}_{-\su}(\R)$ for all $r \geqs 0$. The first of these relations proves the thesis for $r \geqs 0$. On the other hand, on account of the duality relation in Eq. \eqref{HrsDual}, the second relation allow us to infer that $H^{r}_{\su} = (H^{-r}_{-\su}(\R))' \cemb (H^{-r}(\R))' = H^{r}(\R)$ for $r \leqs 0$, which completes the proof.\\
$ii)$ First of all, let us remark that for any $\su \in \R$ the map $\Iu$ defined in Eq. \eqref{IsDist} is invertible, with inverse given by $\Iu^{-1} = I_{-\su}$.\\
Next, let $n \in \{0,1,2,...\}$. It can be checked by direct inspection that the weight $w_{\su}^{1/2} = w_{\su/2}$ on $\R$ fulfills  the conditions of \cite[p. 263, Def. 6.1]{TriIII} and that $\uu \mapsto \|w_{\su}^{1/2} \uu\|^2_{H^n(\R)}$ determines a norm on $H^n_{\su}(\R)$ which is equivalent to that of Eq. \eqref{normHkw} (compare with \cite[p. 267, Eq. (6.18)]{TriIII}); so, from \cite[p. 267, Th. 6.9]{TriIII} it follows that $H^n_{\su}(\R) = F^n_{2,2}(\R,w_{\su}^{1/2})$, where $F^r_{2,2}(\R,w_{\su}^{1/2})$ are the weighted Triebel-Lizorkin spaces defined according to \cite[p. 264, Def. 6.3]{TriIII}. Due to our position \eqref{HrsDual}, the latter identity and \cite[p. 244, Eq. (7)]{TriTop} yield in addition the dual relation $H^{-n}_{\su}(\R) = F^{-n}_{2,2}(\R,w_{\su}^{1/2})$, again for $n \in \{0,1,2,...\}$. Furthermore, let us recall that $H^r(\R) = F^r_{2,2}(\R,1) \equiv F^r_{2,2}(\R)$ for any $r \in \R$ (see, e.g., \cite[p. 3, Eq. (1.8)]{TriIII}).
On the other hand, \cite[p. 265, Th. 6.5]{TriIII} states that the map $F^r_{2,2}(\R,w_{\su}^{1/2}) \to F^r_{2,2}(\R)$, $\uu \mapsto w_{\su}^{1/2} \uu$ is an isomorphism of Banach spaces for all $r \in \R$. In view of the facts mentioned previously, the latter result allows us to infer that the map $\big(\Iu\!\rist\!H^r_{\su}(\R)\big) : H^r_{\su}(\R) \to H^r(\R)$ is a Banach isomorphism for all $r \in \Z$. Then, the analogous statement for arbitrary $r \in \R$ follows by interpolation from \cite[p. 46, Th. 2.1.6]{Lun}, recalling our definition \eqref{HkwInt} and its dual counterpart.\\
$iii)$ The previously proven item ii) implies that $H^{r}_{\su}(\R)$ is isomorphic to $H^{r}(\R)$ for any $r \in \R$. Then the thesis follows straightforwardly from the standard Sobolev embedding $H^{r}(\R) \cemb H^{r'}(\R)$, holding true for $r \geqs r'$.\\
$iv)$ Again, due to item ii) we know that $H^{\ru}_{\su}(\R)$, $H^{\rp}_{\sp}(\R^2)$ and $H^{\min(\ru,\rp)}_{\su,\sp}(\R^3)$ are respectively isomorphic to $H^{\ru}(\R)$, $H^{\rp}(\R^2)$ and $H^{\min(\ru,\rp)}(\R^3)$. On the other hand, Eq. \eqref{EstFT2} of Lemma \ref{LemFact} allows us to infer that $H^{\ru}(\R) \otimes H^{\rp}(\R^2)$ is continuously embedded into $H^{\min(\ru,\rp)}(\R^3)$ for any $\ru,\rp \geqs 0$. Altogether, the previous arguments yield the thesis.
\endproof

The subsequent Lemmata \ref{lemRfW} and \ref{lemtaFw}, characterize the free resolvent operator $\Rf(z)$ and the trace $\tau_F$ on the plane $\pi_F$ as bounded maps on the weighted spaces under analysis.
\begin{lemma}\label{lemRfW}
Let $\su,\sp \in \R$ and $r \in \R$. Then, for all $z \in \C \backslash [0,+\infty)$ there holds
\begin{equation}
\Rf(z) \in \Bou(H^{r}_{\su,\sp}(\R^3),H^{r + 2}_{\su,\sp}(\R^3))\;. \label{lemRfWEq}
\end{equation}
\end{lemma}
\proof First of all, let us notice that by \cite[p. 170, Lem. 1]{ReSiIV} we have $\Rf(z) \in \Bou(L^2_{\su,\sp}(\R^3))$ for all $\su,\sp \in \R$; from here it can be inferred that $\Rf(z) \in \Bou(L^2_{\su,\sp}(\R^3),H^2_{\su,\sp}(\R^3))$, using arguments similar to those described in the proof of Th. 4.2 in \cite{MPS2}. This yields the thesis for $r = 0$; then, proceeding by induction it can be inferred by similar arguments that Eq. \eqref{lemRfWEq} holds true as well for all $r = 2 n$ with $n \in \{0,1,2,...\}$.

To proceed let us remark that $\Rf(z) = \Rf(\ov{z})^\dag$, where ${}^{\dag}$ indicates the adjoint with respect to the usual $L^2$-duality; thus, recalling the basic relation \eqref{HrsDual} we obtain $\Rf(z) \in \Bou(H^{-2 n -2}_{-\su,-\sp}(\R^3),H^{-2 n}_{-\su,-\sp}(\R^3))$ for all $n \in \{0,1,2,...\}$, which is equivalent to Eq. \eqref{lemRfWEq} for $r = - 2 n$ after the obvious replacement $(\su\,,\sp) \to (-\su\,,-\sp)$.

Finally, notice that our definition \eqref{HkwInt} and its dual counter part yield the interpolation identity 
\[[H^{2 n}_{\su,\sp}(\R^3),H^{2 n + 2}_{\su,\sp}(\R^3)]_{\te} = H^{2 n + 2\te}_{\su,\sp}(\R^3) \quad\; \mbox{for any $\te \in (0,1)$ and all $n \in \Z$}; \] 
then, from \cite[p. 46, Th. 2.1.6]{Lun} we can conclude that the thesis stated in Eq. \eqref{lemRfWEq} holds true for arbitrary $r \in \R$ by interpolation.
\endproof

\begin{lemma}\label{lemtaFw} Let $\su,\sp \in \R$, $r \in (0,1/2]$. Then, the evaluation map $\tau_F : \DD(\R^3) \to \DD(\R^{2})$, $(\tau_F f)(\xp) := f(F(\xp),\xp)$ (see subsection \ref{ss:trace}) can be uniquely extended to a continuous operator
\begin{equation*}
\tau_F \in \Bou(H^{r+1/2}_{\su,\sp}(\R^3),H^{r}_{\sp}(\R^{2}))\;. \label{taFwEq}
\end{equation*}
\end{lemma}
\proof In order to avoid misunderstandings, throughout the present proof we shall temporarily indicate with $\hat{\tau}_F$ the usual trace operator introduced in Remark \ref{r:taubonded}, mapping $H^{r+1/2}(\R^3)$ into $H^r(\R^2)$ for $r \in (0,1/2]$. It can be checked by elementary computations that on $\DD(\R^3)$ we have 
\begin{equation}
\tau_F = \Ip^{-1}\,W_{\su}^{(F)}\, \hat{\tau}_F\,\Iup\;, \label{tauFW}
\end{equation}
where $\Ip,\Iup$ are the Banach isomorphisms introduced in item ii) of Lemma \ref{lemHW} and $W_{\su}^{(F)}$ indicates the multiplication operator defined by $(W_{\su}^{(F)} \up)(\xp) := w_{\su}^{-1/2}(F(\xp))\,\up(\xp)$.

Taking into account the identity \eqref{tauFW} and the fact that $\DD(\R^3)$ is a dense subset of $H^{r+1/2}_{\su,\sp}(\R^3)$, the thesis follows as soon as we can prove that $W_{\su}^{(F)} \in \Bou(H^r(\R^2))$ for $r \in (0,1/2]$. On the other hand, due to our assumption $F \in C^{0,1}_0(\R^2)$, we have $w_{\su}^{-1/2}(F(\cdot))\in C_b^{0,1}(\R^2)$ (i.e. $w_{\su}^{-1/2}(F(\cdot))$ is Lipschitz continuous and uniformly bounded); by \cite[p. 13, Th. 1.9.2]{agranovich}, this suffices to infer that $W_{\su}^{(F)}$ is indeed a bounded multiplier in $H^r(\R^2)$ for all $r \in (0,1)$ (hence, in particular, for $r \in (0,1/2]$), which yields the thesis in view of the previous considerations.
\endproof

\begin{rem}\label{remta0}
The proof of Lemma \ref{lemtaFw} can be easily generalized to cases where $F$ fulfills stronger regularity assumptions. In particular, for $F = 0$ one readily obtains
\begin{equation*}
\tau_0 \in \Bou(H^{r+1/2}_{\su,\sp}(\R^3),H^{r}_{\sp}(\R^{2})) \qquad \mbox{for all\, $r > 0$}\;.
\end{equation*}
\end{rem}

%%%%%%
%SECTION
%%%%%%
\section{Schr\"odinger operators $A_0$ and $A_F$, and their resolvents\label{s:res}}

In this section we give a rigorous definition of the operators $A_0$ and $A_F$ and obtain a Kre\u{\i}n's type formula for their resolvents. We remark that a rigorous definition of the operators $A_0$ and $A_F$ can also be obtained, more directly, starting from the associated quadratic form, see, e.g., \cite{BEL14, BEKS,ER16}. However, since we shall extensively use Kre\u{\i}n's type resolvent formulae, we prefer to use them  also to characterize  the domain and action of the operators $A_0$ and $A_F$. 

We  consider first the case in which the interaction is supported on the surface $\pi_F$, and give several definitions and results.  The corresponding definitions for the plane $\pi_0$ are simply obtained by setting $F=0$; obviously, all the results obtained for generic $F\in C_0^{0,1}(\R)$ remain true. Indeed, in the flat case, by exploiting the factorized structure of $\R^3$, several quantities can be explicitly computed and some results can be improved. We pursue this goal at the end of the section. 

\begin{rem}\label{r:Rfreebounded}It is a well known fact (see also Lemma \ref{lemRfW}) that 
\begin{equation*}
\Rf(z) \in \Bou(H^{r}(\R^3), H^{r+2}(\R^3))\qquad \mbox{for all $r\in\R$}\;,
\end{equation*}
where $\Rf(z)$ is the resolvent of the free Laplacian (see Eq.s \eqref{Lapfree} \eqref{Rfree}).
\end{rem}

\subsection{Resolvent and rigorous definition of $A_F$\label{ss:AF}}
Next, we  introduce several families of operators defined by means of  the trace operator $\tau_F$ and the free resolvent $\Rf(z)$. 

Taking into account Remarks \ref{r:taubonded} and \ref{r:Rfreebounded}, we define the operator
\begin{equation}
\Gc_F(z) := \tau_F \Rf(z) : L^2(\R^3) \to H^{1/2}(\R^{2})\qquad \mbox{with $z \in \C \backslash [0,+\infty)$}\;, \label{Gcdef}
\end{equation}
which admits a unique continuous extension 
\begin{equation*}
\Gc_F(z)  \in \Bou(H^{r-3/2}(\R^3), H^{r}(\R^{2}))\qquad \mbox{for $r\in(0,1/2]$}\;.% \label{GcExt}
\end{equation*}

The corresponding adjoint (meant in the sense of the Sobolev duality $H^{-r}(\R^k) = (H^r(\R^k))'$) with conjugate spectral parameter is the \emph{single layer operator}
\begin{equation*}
G_F(z) := (\tau_F \Rf(\ov{z}))^\dag : H^{-1/2}(\R^2)\to L^2(\R^3)\qquad \mbox{with $z \in \C \backslash [0,+\infty)$}\;, \label{Gdef}
\end{equation*}
which admits a unique continuous extension
\begin{equation}
G_F(z) \in \Bou(H^{-r}(\R^{2}), H^{3/2-r}(\R^3)) \qquad \mbox{for $r\in(0,1/2]$} \;. \label{GExt}
\end{equation}

Next, taking into account Eq. \eqref{GExt}, we consider the trace of the single layer operator
\begin{equation*}
\Gt_F(z) := \tau_F G_F(z)\,,
%\label{Gtdef}
\end{equation*}
so that  
\begin{equation}\label{3.4a}
\Gt_F(z) \in \Bou(H^{-1/2}(\R^{2}),H^{1/2}(\R^{2})).
\end{equation}
The Sobolev indices in the latter claim are fixed by the restriction $r\in(0,1/2]$ in  Eq. \eqref{GExt}, and by the fact that it must be $1-r \in  (0,1/2]$ as well. Let us notice that with respect to the $H^{-1/2}(\R^{2})$-$H^{1/2}(\R^{2})$ duality induced by the $L^{2}(\R^{2})$ scalar product, one has 
\begin{equation}\label{adj}
\Gt_F(z)^{*}=(\tau_{F}R_{\fr}(z)\tau_{F}^{*})^{*}=\Gt_F(\bar z)\,.
\end{equation}
\par
\begin{rem} Notice that the result in Eq. \eqref{3.4a} can be improved whenever $F$ is smooth: in this case $\Gt_F(z)\in \Bou(H^{r-1/2}(\R^{2}),H^{r+1/2}(\R^{2}))$ for any $r\in\R$ (see, e.g., \cite[Prop. 13]{Rab}).
\end{rem}

For any given $\al > 0$, we define the operator
\begin{equation} \label{3.6a}
\Ga_F(z) := \big(1 + \al\,\Gt_F(z)\big)\;.
\end{equation}

The following lemma guarantees the invertibility of the operator $\Ga_F(z)$ for $z$ far away from the real positive axis. Indeed, in view of \cite[Th. 2.19]{CFPinv}, one has that $\Ga^{-1}_F(z) \in \Bou(H^{r}(\R^{2}))$ for all $|r|<1/2$ and  $z \in \C\backslash [0,+\infty)$, see Theorem \ref{thm:RF} and Remark \ref{r:RF} below. 

\begin{lemma}\label{l:lipschitz} Let $d_z := \inf_{\lambda\in[0,+\infty)}|\lambda-z|$.  Then   there exists $z_0 \in \C\backslash[0,+\infty)$ such that
\begin{equation}\label{pigs2}
\Ga^{-1}_F(z) \in \Bou(H^{r}(\R^{2})) \qquad \mbox{for all $|r|<1/2$ and  $z \in \C$ such that  $d_z > d_{z_0}$}\;.
\end{equation}
\end{lemma}
\proof We start with the proof of the following statement: for any  $\eps\in (0,1/2)$,  and  $z \in \C \backslash [0,+\infty)$, on has 
\begin{equation}\label{pigs1}
\|\Gt_F(z)\|_{\Bou(H^{-1/2+\eps}(\R^{2}),H^{1/2 -\eps}(\R^{2}))} \lec d_z^{-\eps} \quad \left(d_z := \inf_{\lambda\in[0,+\infty)}|\lambda-z|\right) . 
\end{equation}

Recall that $\tau_F \in \Bou(H^{r+1/2}(\R^3),H^r(\R^{2}))$ for all $r\in(0,1/2]$, which suffices to infer $\tau_F^\dag \in \Bou(H^{-r}(\R^{2}),H^{-r-1/2}(\R^3))$ as well. From the identity $\Gt_F(z) = \tau_F \Rf(z) \tau_F^\dag$ we obtain the bound  
\[\begin{aligned}
& \|\Gt_F(z)\|_{\Bou(H^{-r+\eps}(\R^{2}),H^{1-r-\eps}(\R^{2}))}  \\ 
\leq & \,\|\tau_F\|_{\Bou(H^{3/2 - r -\eps}(\R^{3}),H^{1 - r -\eps}(\R^{2}))} \times  \\ 
&\times \|\Rf(z)\|_{\Bou(H^{-r -1/2+\eps}(\R^{3}),H^{3/2-r-\eps}(\R^{3}))} \|\tau_F^\dag\|_{\Bou(H^{-r+\eps}(\R^{2}),H^{-r-1/2+\eps}(\R^{3}))}\;.
\end{aligned}\]
Where we used $1-r-\eps \in (0,1/2]$ and $r-\eps \in(0,1/2]$. Then, the thesis follows as soon as we can show the norm bound
\begin{equation}
\|\Rf(z)\|_{\Bou(H^{-r -1/2+\eps}(\R^{3}),H^{3/2-r-\eps}(\R^{3}))}  \lec {1 \over d_z^\eps}\;. \label{BouRes}
\end{equation}
Let us fix arbitrarily $f \in H^{-r-1/2+\eps}(\R^3)$ and notice that, using the equivalent Sobolev norm \eqref{Equinorm} (see also Eq. \eqref{EquinormFou}), we have
\[\begin{aligned}
\|\Rf(z) f\|_{H^{3/2-r-\eps}(\R^3)}^2  & \lec \|(1-\lap)^{(3/2-r-\eps)/2}\Rf(z) f\|^2_{L^2(\R^3)} \\ 
& =  \int_{\R^3} d\k\;{(1 + |\k|^2)^{3/2-r-\eps} \over \big||\k|^2-z\big|^2}\; |(\Fou f)(\k)|^2\;. 
\end{aligned}
\]
Then, using the elementary estimates $||\k|^2-z| \geqs d_z$ and 
\[
\begin{aligned}
&\sup_{\k \in \R^3} \l({(1 + |\k|^2)^{3/2-r-\eps} (1 + |\k|^2)^{r+1/2-\eps}\over \big||\k|^2-z\big|^2}\r)  \\ 
\leqs &  \sup_{\k \in \R^3} \l({(1 + |\k|^2)^{2-2\eps} \over \big||\k|^2-z\big|^{2-2\eps} d_z^{2\eps}}\r) \lec {1 \over d_z^{2\eps}} 
\end{aligned}
\]
and making reference to the relations contained in Eq.s \eqref{Equinorm}-\eqref{EquinormIne}, we obtain
\[
\begin{aligned}
\|\Rf(z) f\|_{H^{3/2-r-\eps} (\R^3)}^2  \lec & {1 \over d_z^{2\eps}} \int_{\R^3} d\k\;(1 + |\k|^2)^{-r-1/2+\eps}\; |(\Fou f)(\k)|^2  \\ 
\lec & {1 \over d_z^{2\eps}}\; \|f\|_{H^{-r-1/2+\eps}(\R^3)}^2 \;; 
\end{aligned}
\]
in view of the arbitrariness of $f \in H^{-r-1/2+\eps}(\R^3)$, the above bound proves the statement \eqref{BouRes}, whence \eqref{pigs1}.

To prove claim \eqref{pigs2}, we note first that,  by  continuous Sobolev embedding, 
\[
\|\Gt_F(z)\|_{\Bou(H^{r}(\R^{2}))}  \lec\|\Gt_F(z)\|_{\Bou(H^{-1/2+\eps}(\R^{2}),H^{1/2 -\eps}(\R^{2}))}  \lec  {1 \over d_z^\eps }
\]
for any $r \in(-1/2,1/2)$, and $\eps\in(0,1/2)$ such that $-1/2+\eps \leq r \leq 1/2 -\eps$. Hence, $\|\Gt_F(z)\|_{\Bou(H^{r}(\R^{2}))}   \lec  {1 /d_z^\eps }$ for all $r \in (-1/2,1/2) $ and $\eps \in (0, 1/2-|r|)$. As a consequence,   for $r\in(-1/2,1/2)$,  $\Ga_F(z)\in  \Bou(H^{r}(\R^{2})) $ has bounded inverse for any $z$ such that $d_{z}>\alpha^{1/\epsilon}$.
\endproof

We note that  
\begin{equation*}\label{adjG}
(u,\Ga_F(z)\,v)_{H^{-1/2}(\R^{2}),H^{1/2}(\R^{2})} = \overline{(v,\Ga_F(\bar z)\, u)_{H^{-1/2}(\R^{2}),H^{1/2}(\R^{2})}}\,, 
\end{equation*}
with $u,v\in H^{-1/2}(\R^{2})$, and where $(\cdot,\cdot)_{X',X}$ denotes the duality product between $X$ and its dual $X'$, which follows immediately from the identity $\Gt_F(z)^{*} = \Gt_F(\bar z)$. Moreover,  for all $z,w\in\C\backslash [0,+\infty)$, on has 
\begin{equation} \label{diffGa}
\alpha^{-1}\Ga_F(z) -\alpha^{-1} \Ga_F(w) = (z-w)\,\Gc_F(w)\,G_F(z) \;.
\end{equation}
To prove the latter identity,  note that  
\[
\begin{aligned}
 \Gc_F(z) -  \Gc_F(w) =\, & \tau_F \big(\Rf(z)-\Rf(w)\big) \\
 =\, & (z-w)\,\tau_F \Rf(z)\Rf(w) =  (z-w)\,\Gc_F(z)\Rf(w)\,.
\end{aligned}
\]
By taking the adjoint (in $\bar z$ and $\bar w$) it follows that  $ G_F(z) -  G_F(w) = (z-w)\Rf(w)G_F(z)$. Hence, 
\begin{equation*}%\label{diffM}
\Gt_F(z) - \Gt_F(w) =  \tau_F\big(G_F(z) -  G_F(w)\big) = (z-w)\,\Gc_F(w)\,G_F(z) \,,
\end{equation*}
from which Identity \eqref{diffGa} readily follows. 

By Lemma \ref{l:lipschitz},  \cite[Th. 2.1]{Pos01}  and   \cite[Th. 2.19]{CFPinv}, we have the following

\begin{thm}\label{thm:RF} There holds
\begin{equation} \label{3.13a}
\Ga^{-1}_F(z) \in \Bou(H^{r}(\R^{2})) \qquad \mbox{for all $|r|<1/2$ and  $z \in \C\backslash [0,+\infty)$}\;,
\end{equation}
and  the bounded linear operator
\begin{equation}
R_F(z) := \Rf(z) - \al\,G_F(z)\,\Ga^{-1}_F(z)\,\Gc_F(z) \qquad \mbox{with $z\in \C\backslash [0,+\infty)$}
\label{RG}
\end{equation}
is the resolvent of the self-adjoint operator $A_F$ which coincides with $A_\fr$  on $H^2(\R^3 \backslash \pi_F)$ and which is defined by 
\[\Dom(A_F) := \big\{f\in L^2(\R^3) \; \big|\; f = f_z  - \al\,G_F(z)\,\Ga^{-1}_F(z)\,\tau_F f_z\,,\; f_z \in H^2(\R^3)\big\}\]
\[ (A_F - z) f := (A_\fr -z) f_z\;.\]
\end{thm}

\begin{rem}\label{r:RF} Few comments on the proof of Theorem \ref{thm:RF} are in order. By Lemma \ref{l:lipschitz} and  \cite[Th. 2.1]{Pos01} one has that $R_F(z)$, given in Eq. \eqref{RG}, is the resolvent of the operator $A_F$ for $d_z$ large enough. Hence, \cite[Th. 2.19]{CFPinv}, together with Eq. \eqref{diffGa} and the fact that $\Gamma_F(z)^* = \Gamma_F(\bar z)$ (see Eq. \eqref{adj} and \eqref{3.6a}) implies that $\Ga_F(z)$ is invertible for all $z\in \rho(\Afree)\cap \rho(A_F)$, where $\rho(\Afree)$ and $\rho(A_F)$ denote the resolvent set of $\Afree$ and $A_F$ respectively.  That $ \rho(\Afree)\cap \rho(A_F) \subseteq \C\backslash [0,+\infty)$ is a consequence of Remark \ref{r:spectrum} below. Indeed, we shall prove that $ \sigma (A_F) = \sigma(\Afree)  =[0,+\infty)$, hence  $\rho(\Afree)\cap \rho(A_F) =  \C\backslash [0,+\infty)$, see Remark \ref{remWd22}. 
\end{rem}

With $f$ given as in $\Dom(A_F)$ one has 
\[
\tau_F f = \tau_F f_z - \al \, \Gt_F(z) \Ga^{-1}_F(z)\tau_F f_z =  \al \,\Ga^{-1}_F(z) \tau_F f_z\;.
\]

Moreover, indicating with $\bn$ the unit vector normal to the surface $\pi_F$ pointing to the right and with $\de_\bn$ the derivative in the direction normal to $\pi_F$, using the notations introduced in Eq. \eqref{jump} one has 
\[
[\de_\bn f]_{\pi_F} = -\, \al\, [\de_\bn G_F(z)\,\Ga^{-1}_F(z)\,\tau_F f_z]_{\pi_F}\;. 
\]
Since $ [\de_\bn G_F(z)\,u ]_{\pi_F} = -\, u$ (see, e.g., \cite{McL}), the operator $A_F$ can be also characterized as 
\begin{equation*}\label{251}
\Dom(\AA_F) = \big\{f \in H^2(\R^3 \backslash \pi_F)\;\big|\; f\!\rist_{\pi_F^+}\, = f\!\rist_{\pi_F^-}\, = f\!\rist_{\pi_F}\,,\; [\de_\bn f]_{\pi_F} = \al\,f\!\rist_{\pi_F}\!\big\}\,,
\end{equation*}
\begin{equation*}\label{252}
A_F f = -\Delta f \qquad \forall\,\x \in \R^3\backslash \pi_F \quad\! \textrm{and}\quad\! f\in \Dom(\AA_F) \;.
\end{equation*}

The operator $A_F$ corresponds to  the singular perturbation of the free Laplacian formally written as in Eq. \eqref{AF}, and given by  a delta-type potential of strength $\alpha$ supported on the surface $\pi_F$. 

\begin{rem}\label{r:spectrum}
By the definition of $\Gt_F(z)$ one has that, for all $\lambda>0$
\[(u, \Gt_F(-\lambda)\, u )_{L^2(\R^2)} = (\tau_F^* u,  \Rf(-\lambda)\,\tau_F^* u)_{H^{-1}(\R^3),H^{1}(\R^3)}\geqs 0, \]
because  $ \Rf(-\lambda)$ is a positive definite operator.  Hence,
\[
\|u\|_{L^2(\R^2)} \|(1+\alpha \Gt_F(-\lambda))u\|_{L^2(\R^2)} \geqs \big (u, (1+\alpha \Gt_F(-\lambda)) u \big)_{L^2(\R^2)}  \geqs \|u\|^2_{L^2(\R^2)},
\]
which in turn implies that the inverse $\Gamma_F^{-1}(-\lambda) = (1+\alpha \Gt_F(-\lambda))^{-1}$ is a well-defined and bounded operator in $L^2(\R^2)$ for all $\lambda>0$. This argument, together with the fact that for $\lambda>0$ the operators $\Rf(-\lambda)$, $G_F(-\lambda)$, and $\Gc_F(-\lambda)$ are bounded in $L^2(\R^k)$ $(k = 2,3)$, allows us to infer that $R_F(-\lambda) \in \Bou(L^2(\R^3))$ and  $\sigma(A_F) \subseteq [0,+\infty)$.
\end{rem}

\subsection{Resolvent and rigorous definition of $A_0$\label{ss:A0}}  
Obviously, all the results stated in the previous section hold true for the case $F=0$ as well. In particular, by Theorem \ref{thm:RF}, the operator 
\begin{equation*}
\Ga_0(z)= 1 + \al\,\Gt_0(z)
\end{equation*}
is certainly invertible in $H^r(\R^2)$ for all $|r|\leqs 1/2$ and $z\in\C\backslash[0,+\infty)$ (indeed it admits a bounded inverse for all  $r\in\R$, see Lemma \ref{LemGt0} below).  Hence,  the operator $R_0(z)$ is well defined and it is the resolvent of the self-adjoint operator 
\begin{equation*}
\Dom(\AA_0) = \big\{f \in H^2(\R^3 \backslash \pi_0)\;\big|\; f\!\rist_{\pi_0^+}\, = f\!\rist_{\pi_0^-}\, = f\!\rist_{\pi_0}\,,\; [\de_\bn f]_{\pi_0} = \al\,f\!\rist_{\pi_0}\!\big\}\,,
\end{equation*}
where we used the notations of Eq. \eqref{jump}, $\de_\bn$ denotes  the derivative in the direction normal to the plane $\pi_0$, $(\de_\bn f)\!\rist_{\pi_0^\pm}\, \equiv (\partial_{x^1} f)\!\rist_{x^1=0^\pm}$;
\[
A_0 f = -\Delta f \qquad \forall\,\x \in \R^3\backslash \pi_0 \quad\! \textrm{and}\quad\! f\in \Dom(\AA_0) \;.
\]
\begin{rem}\label{r:R0}
Explicitly $R_0(z)$ reads 
\begin{equation}\label{R0}
R_0(z) := \Rf(z) - \al\,G_0(z)\,\Ga^{-1}_0(z)\,\Gc_0(z)\,,
\end{equation}
and it is a bounded operator for all $z\in \C\backslash [0,+\infty)$.
\end{rem}

Before proceeding further let us point out some special properties of the operators $\Gc_0(z)$, $G_0(z)$,  and $\Gt_0(z)$. For all $r > 0$, since  $\tau_0 : H^{r+1/2}(\R^3) \to H^{r}(\R^{2})$, we have
\begin{equation*}%\label{release}
\Gc_0(z)  \in \Bou(H^{r-3/2}(\R^3), H^{r}(\R^{2}))\;,\qquad 
G_0(z) \in \Bou(H^{-r}(\R^{2}), H^{3/2-r}(\R^3))\;.
\end{equation*}

Moreover let us recall a fact which was proven in \cite{CFP}; namely that  there holds true the identity
\begin{equation*}
\Gt_0(z) = {i \over 2}\,(z+\lap_\parallel)^{-1/2} \,\in\, \Bou(L^2(\R^{2})) \qquad \mbox{for any $z \in \C \backslash [0,+\infty)$} \;.
\end{equation*}
In view of the above result, it appears that $\Gt_0(z)$ sends $H^r(\R^{2})$ continuously into $H^{r+1}(\R^{2})$ for any $r \in \R$; in turn, this allows to infer a stronger version of the statement in Eq. \eqref{3.13a} in the particular case $F = 0$:
\begin{rem}\label{LemGt0} For any $z \in \C \backslash [0,+\infty)$ and any $r \in \R$, $\Ga^{-1}_0(z)\in \Bou(H^{r}(\R^{2}))$. To see that this is indeed the case, let $u \in H^r(\R^{2})$; then, using the equivalent Fourier norm on $H^r(\R^2)$, one obtains
$$ \|\Ga^{-1}_0(z)\,u\|_{H^r(\R^{2})}^2 \lec \int_{\R^{2}}d\kp\;(1+|\kp|^2)^r \l|{\sqrt{|\kp|^2-z} \over {\al \over 2} + \sqrt{|\kp|^2-z}}\r|^2 |(\Fou \,u)(\kp)|^2\;. $$
Note that for all $\al > 0$ and all $z \in \C \backslash [0,+\infty)$, the map
$$ (0,+\infty) \ni k \mapsto \l|{\sqrt{k^2-z} \over {\al \over 2} + \sqrt{k^2-z}}\r|^2 $$
is bounded. This and the previous bound suffice to infer the claim.
\end{rem}

\begin{rem} \label{r:factor}
The operator $\AA_0$ can be equivalently represented as
\begin{equation}
\AA_0 = \Hu_0 \!\otimes \uno_{\parallel} + \uno_{1} \otimes \Hp_0\;, \label{AA0fact}
\end{equation}
where $\uno_{1}$, $\uno_{\parallel}$ indicate, respectively, the identity operators on $L^2(\R)$, $L^2(\R^2)$, and
\begin{eqnarray}
&\dd{\Hu_0 \!:= -\,\de_{x^1 x^1} : \Dom \Hu_0 \subset L^2(\R) \to L^2(\R) \,,} &  \label{A01_1}\\ 
&\dd{\Dom \Hu_0\! := \{u \!\in\! H^1(\R)\!\cap\! H^2(\R \backslash \{0\})\,|\, u'(0^{+}) - u'(0^{-}) = \al\,u(0)\}\,,\quad\;} & \label{A01_2}
\end{eqnarray}
\[ 
\Hp_0 \!:= -\,\lapp : H^2(\R^2) \to L^2(\R^2) \;.
\]
In other words, $\Hu_0$ denotes the self-adjoint operator on $L^2(\R)$ corresponding to the formal expression ``$-\,\de_{x^1 x^1} + \al\,\delta_{\{x^1 = 0\}}$'' (see \cite[Th. 3.1.1]{AlbBook}) and $\Hp_0$ is the free Laplacian on the plane $\pi_0$. We recall that $\sigma(\Hp_0) = \sigma_{ac}(\Hp_0) = [0,+\infty)$ and that, for $\alpha>0$, one has $\sigma(\Hu_0) = \sigma_{ac}(\Hu_0) = [0,+\infty)$ (see \cite[Th. 3.1.4]{AlbBook}).  Hence, 
$$\sigma(A_0) =\sigma_{ac}(A_0) =  [0,+\infty)\,.$$
\end{rem}

\subsection{A formula for the difference of the resolvents}
We conclude this section  by pointing  out the following basic identity:
\begin{equation}\label{ResDiff}
\begin{aligned}
& R_F(z) - R_0(z)  
=  -\,\al\Big[G_0(z)\,\Ga^{-1}_0(z)\,\big(\Gc_F(z) - \Gc_0(z)\big) \\ 
 +&   
\big(G_F(z) - G_0(z)\big)\,\Ga^{-1}_F(z)\,\Gc_F(z)
 +  G_0(z)\,\big(\Ga^{-1}_F(z)- \Ga^{-1}_0(z)\big)\,\Gc_F(z)\Big]\,,
 \end{aligned}
\end{equation}
which can be easily derived from Eq.s \eqref{RG} and \eqref{R0} by addition and subtraction of identical terms. 

\section{The limiting absorption principle\label{s:LAP}}

In the present section we derive Limiting Absorption Principles (LAPs) for the resolvent operators $R_0(z)$ and $R_F(z)$; more precisely, we show that, for any $\lam \in (0,+\infty)$, the limits $\eps \downarrow 0$ of $R_0(\lam \pm i \eps)$ and $R_F(\lam\pm i \eps)$ determine bounded operators on suitable functional spaces (namely, on weighted Sobolev spaces, see Section \ref{ss:wss}).
The results obtained here will be employed in the forthcoming Sections \ref{s:wo} and \ref{s:sm}, concerning respectively the existence and completeness of the wave operators and the scattering matrix associated to the couple $(A_F,A_0)$.

Our approach mainly consists of the following two steps: first, we derive LAP for $R_0(z)$, starting from simpler lower dimensional operators and employing a result of Ben-Artzi and Devinatz \cite{BAD83} (see also \cite{BAD87}) about sums of tensor products; then, we determine an analogous result for $R_F(z)$, using some abstract perturbation techniques of Renger \cite{Ren}.

\subsection{The limiting absorption principle for $R_0(z)$}

Let us consider the unperturbed resolvent $R_0(z)$ and the corresponding self-adjoint operator $\AA_0$. 

We recall  that the operator $\AA_0$ is factorized as $\AA_0 =  \Hu_0 \!\otimes \uno_{\parallel} + \uno_{1} \otimes \Hp_0$, see Remark \ref{r:factor}.  For any $z \in \C \backslash [0,+\infty)$ we consider the resolvent operators
\[
\Ru_0(z) := (\Hu_0-z)^{-1} : L^2(\R) \to \Dom \Hu_0 
\]
and 
\[
\Rp_0(z) := (\Hp_0-z)^{-1} : L^2(\R^2) \to H^2(\R^2)\;. 
\]

The result of Ben-Artzi and Devinatz we refer to is Th. 3.8 in \cite{BAD83} which, in our setting,  grants the validity of LAP for the  operator $A_0$ in the factorized form \eqref{AA0fact} mentioned above. The cited theorem of \cite{BAD83} is here employed by setting $H_1 := \Hp$, $H_2 := \Hu$,  $\Lambda := (-\infty,0]$, and  $U = (0,+\infty)$. The corresponding hypotheses on $H_1$ and $H_2$ appear to be fulfilled, respectively, in consequence of the results reported in \cite[Example 2.2]{BAD83} and of Theorem \ref{lemLAP1} below. 

\begin{thm}\label{lemLAP1}
Assume that $\te \in (0,1/2)$ and let $\su > 1/2$. Then, for any $\lam \in (0,+\infty)$, the limits
\begin{equation*}
R^{(1),\pm}_{0}(\lam) := \lim_{\eps \downarrow 0} \Ru_0(\lam \pm i \eps)
\end{equation*}
exist in $\Bou(L^2_{\su}(\R),H^{1+\te}_{-\su}(\R))$ and the convergence is uniform in any compact subset $\KK \subset (0,+\infty)$.
\end{thm}

The proof of the latter  theorem is based on a series of explicit estimates and it is rather lengthy; for this reason we postpone it to   \ref{app:A}. 

Using the previous Theorem \ref{lemLAP1} and some known results of Agmon \cite{AgmonLAP}, and of Ben-Artzi and Devinatz \cite{BAD83} we can infer the following Proposition.
\begin{prop}[LAP for $A_0$]\label{PropLAP0}
Assume that $\te \in (0,1/2)$ and let $\siu,\sip > 1/2\,$. Then, for any $\lam \in (0,+\infty)$, the limits
\begin{equation}
R^{\pm}_{0}(\lam) := \lim_{\eps \downarrow 0} R_0(\lam \pm i \eps) \label{R0lpm}
\end{equation}
exist in $\Bou(L^2_{\siu,\sip}(\R^3),H^{1+\te}_{-\siu,-\sip}(\R^3))$ and the convergence is uniform in any compact subset $\KK \subset (0,+\infty)$; in particular,
\begin{equation*}
R^{\pm}_{0}(\lam) \in \Bou(L^2_{\siu,\sip}(\R^3),L^2_{-\siu,-\sip}(\R^3))\;. \label{LAP0eq}
\end{equation*}
\end{prop}
\proof Let us consider the representation \eqref{AA0fact} of $\AA_0$ as a sum of tensor products involving the reduced operators $\Hu_0$ and $\Hp_0$. Recall that Theorem \ref{lemLAP1} gives LAP for the resolvent operator $\Ru_0(z)\,$. Moreover, let us mention the following analogous result for the resolvent $\Rp_0(z)$ (see \cite[Sec. 4]{AgmonLAP}; see also \cite[Prop. 5.1]{BAD83}): for all $\sip > 1/2$ and for any $\lam \in (0,+\infty)$, the limits $R^{(\parallel),\pm}_{0}(\lam) := \lim_{\eps \downarrow 0} \Rp_0(\lam \pm i \eps)$ exist in $\Bou(L^2_{\sip}(\R^2),H^{2}_{-\sip}(\R^2))$. Then, noting the elementary identity $L^2_{\siu}(\R) \otimes L^2_{\sip}(\R^2) = L^2_{\siu,\sip}(\R^3)$ and that the tensor product of Hilbert spaces $H^{1+\te}_{-\siu}(\R) \otimes H^2_{-\sip}(\R)$ is continuously embedded into $H^{1+\te}_{-\siu,-\sip}(\R)$ on account of item iv) of Lemma \ref{lemHW}, the thesis follows straightforwardly from \cite[Th. 3.8]{BAD83} (recall the remarks at the beginning of this section).
As we remarked in the introduction to this section, the cited theorem of \cite{BAD83} is here employed setting $H_1 := \Hp$, $H_2 := \Hu$,  $\Lambda := (-\infty,0]$, and  $U = (0,+\infty)$; the corresponding hypotheses on $H_1$ and $H_2$ are  fulfilled, respectively, in consequence of the results  in \cite[Example 2.2]{BAD83} and in Theorem \ref{lemLAP1}.
\endproof

\subsection{The limiting absorption principle for $R_F(z)$}

As previously anticipated, we now proceed to derive LAP for the resolvent $R_F(z)$; to this purpose, we start from the analogous result for $R_0(z)$ determined in the previous subsection and use an abstract perturbation method of Renger \cite{Ren}. More precisely, we want to use Th. 7  of \cite{Ren} (see also Prop. 10 of the same paper). To this aim we must check that the operators $A_0$ and $A_F$ satisfy Hypotheses $1$, $8$, and $9$ of \cite{Ren}; adapted to our setting those read: 
\begin{itemize}
\item Hypothesis 1 of \cite{Ren}. The operators $A_0$ and $A_F$ are self-adjoint and semibounded (which is certainly true). There exists a constant $c_R\in \R$ such that for all $z\in\C $ with $\Re z<c_R$, there holds  $R_0(z) \in \Bou(L^2_{\siu,\sip}(\R^3))$ and $R_F(z) \in \Bou(L^2_{\siu,\sip}(\R^3))$. 
\item Hypothesis 8 of \cite{Ren}.  For all $\lambda \in (0,+\infty)$, $R^{\pm}_{0}(\lam) := \lim_{\eps \downarrow 0} R_0(\lam \pm i \eps)$ exist and are continuous in $\Bou(L^2_{\siu,\sip}(\R^3),L^2_{-\siu,-\sip}(\R^3))$. Moreover, for each compact subset $\KK$ of $(0,+\infty)$  there exists a constant $c_K>0$ such that for all $\lambda\in\KK$  and all $f \in L^2_{2\siu,2\sip}(\R^3)$ with $R^{+}_{0}(\lam) f = R^{-}_{0}(\lam) f$, there holds $\|R_0^{\pm}(\lam)\,f \|_{L^2(\R^3)} \lec \|f\|_{L^2_{2\siu,2\sip}(\R^3)}$. 
\item Hypothesis 9 of \cite{Ren}. There exists a constant $c_E\in\R$ such that for all $\mu<c_E$ and for some $\gamma>0$, there holds 
\[R_F(\mu) - R_0(\mu) \in \Com(L^2(\R^3), L^2_{2\siu + \ga,2\sip + \ga}(\R^3))\;.\]
\end{itemize}
Here $\siu$ and $\sip$ are some suitable indices for the weights in the Sobolev spaces. 

The forthcoming Proposition \ref{PropH1RF} (see also Remark \ref{remR0}) proves  Hypothesis 1; Propositions  \ref{PropLAP0} and  \ref{PropR0pm} prove Hypothesis 8, and Proposition   \ref{PropResComW}  proves Hypothesis 9. They are   later employed,  together with Th. 7 of \cite{Ren}, in the proof of Theorem \ref{ThmLAPF} (were the allowed values of the indices $\siu$ and $\sip$ are explicitly  given).

\begin{lemma}\label{lemGF}Let $\su,\sp \in \R$, $r \in (0,1/2]$. Then, for all $z \in \C \backslash [0,+\infty)$, there hold:
\begin{align}
& \Gc_F(z) \in \Bou(H^{r-3/2}_{\su,\sp}(\R^3),H^{r}_{\sp}(\R^2))\;; \label{lemGFEq1}\\
& G_F(z) \in \Bou(H^{-r}_{\sp}(\R^2),H^{3/2-r}_{\su,\sp}(\R^3))\;. \label{lemGFEq2}
\end{align}
\end{lemma}
\proof Recall that $\Gc_F(z) = \tau_F \Rf(z)$  by definition (see Eq. \eqref{Gcdef}); then, the statement in Eq. \eqref{lemGFEq1} follows readily from Lemmata \ref{lemRfW} and \ref{lemtaFw}. The analogous claim in Eq. \eqref{lemGFEq2} can be derived by simple duality arguments recalling Eq. \eqref{HrwDual} and the basic relation $G_F(z) = \Gc_F(\ov{z})^{\dag}$.
\endproof

The following lemma adapts the result of Lemma \ref{l:lipschitz} (in particular of Eq. \eqref{pigs2}) to weighted Sobolev spaces. 
\begin{lemma}\label{lemGaIF}Let $\sp \in \R$ and $|r|<1/2$. Then there exists $z_0\in \C\backslash[0,+\infty)$ such that  
\begin{equation}
\Ga_F^{-1}(z) \in \Bou(H^r_{\sp}(\R^2))\qquad \mbox{for all  $z \in \C$ such that  $d_z > d_{z_0}$}\;. \label{lemGaFEq}
\end{equation}
\end{lemma}
\proof 
Let us first recall that for  $|r| < 1/2$ and $z$ away from the positive real axis ($d_z>d_{z_0}$) the map $\Ga_F(z) : H^r(\R^2) \to H^r(\R^2)$ is bounded, invertible and can be expressed as $\Ga_F(z) = 1 + \al\,\Gt_F(z)$ in terms of the operator $\Gt_F(z) = \tau_F\,G_F(z)$ (see  Lemma \ref{l:lipschitz}).\\
Keeping in mind the definition of $\Gt_F(z)$, by Lemmata \ref{lemtaFw} and \ref{lemGF} we get $\Gt_F(z) \in \Bou(H^{-1/2}_{\sp}(\R^2),H^{1/2}_{\sp}(\R^2))$ for all $\sp \in \R$; so, in particular, by item iii) of Lemma \ref{lemHW} (see also Remark \ref{r:embedding}) we have $\Gt_F(z) \in \Bou(H^{r}_{\sp}(\R^2))$ for all $|r|\leq 1/2$. This implies that $\Ga_F(z) \in \Bou(H^{r}_{\sp}(\R^2))$, under the same assumptions on $r$ and $\sp$. \\ 
Next, let $\sp \geqs 0$ and  $|r|< 1/2$  and notice that $H^r_{\sp}(\R^2) \subset H^r(\R^2)\,$ by item i) of Lemma \ref{lemHW}; then, the previous considerations and the bounded inverse theorem yield the thesis \eqref{lemGaFEq}. \\
On the other hand, for $\sp < 0$ the analogous statement can be derived by the usual duality arguments noting that $\Ga_F(z) = \Ga_F(\ov{z})^{\dag}$.
\endproof

\begin{prop}\label{PropH1RF} [Check of Hypothesis 1 of \cite{Ren}]
Assume that $\eta,\eta'\!>0$ and let $\su,\sp \in \R$. 
Then there exists $z_0\in \C\backslash[0,+\infty)$ such that  $R_F(z) \in \Bou\big(H^{\eta-3/2}_{\su,\sp}(\R^3),H^{1/2+\eta}_{\su,\sp}(\R^3) \cap H^{3/2-\eta'}_{\su,\sp}(\R^3)\big)$ for all  $z \in \C$ such that  $d_z > d_{z_0}$. In particular there holds 
\begin{equation*}
R_F(z) \in \Bou(L^2_{\su,\sp}(\R^3))  \qquad \mbox{for all  $z \in \C$ such that  $d_z > d_{z_0}$}\;. \label{H1RFEqL2}
\end{equation*}
\end{prop}
\proof Consider the Kre\u{\i}n's type relation \eqref{RG} for $R_F(z)$; then, the thesis follows from a straightforward application of Lemmata \ref{lemHW}, \ref{lemRfW}, \ref{lemGF} and \ref{lemGaIF}.
\endproof

\begin{rem}\label{remR0}
In view of the facts pointed out in Remark \ref{remta0}, it appears that the results stated in Lemmata \ref{lemGF} and \ref{lemGaIF} can be easily generalized under stronger hypotheses on $F$. In particular, for $F = 0$ we have
\begin{equation*}
\Gc_0(z) \in \Bou(H^{r-3/2}_{\su,\sp}(\R^3),H^{r}_{\sp}(\R^2)) \quad \mbox{and} \quad
G_0(z) \in \Bou(H^{-r}_{\sp}(\R^2),H^{3/2-r}_{\su,\sp}(\R^3)) \label{lemG0Eq}
\end{equation*}
for all $r > 0$,
\begin{equation*}
\Ga_0^{-1}(z) \in \Bou(H^r_{\sp}(\R^2)) \qquad \mbox{for all $r \in \R$ and $z\in\C\backslash [0,+\infty)$}\;. \label{lemGa0Eq}
\end{equation*}
Of course, Proposition \ref{PropH1RF} continues to hold true if $F = 0$.
\end{rem}

Let us now proceed to characterize the resolvent difference $R_F(z)-R_0(z)$ as a compact operator between suitable, weighted Sobolev spaces. Our main result in this direction is stated in Proposition \ref{PropResComW}, whose proof relies on the forthcoming Lemmata \ref{lemtaComW} and \ref{GcCompW}.

\begin{lemma}\label{lemtaComW}Let $\su,\sp\in \R$. Then, for all $r \in (0,1/2)$, $\eps \in (0,r]$ and all $\tp \in \R$ there holds
\begin{equation*}
\tau_F - \tau_0 \in \Com(H^{r+1/2}_{\su,\sp}(\R^3), H^{r-\eps}_{\tp}(\R^{2})) \;. \label{trdiffW}
\end{equation*}
\end{lemma}
\proof First of all, let us remark that by Lemma \ref{lemtaFw} (see also the related Remark \ref{remta0}) we have $\tau_F - \tau_0 \in \Bou(H^{r+1/2}_{\su,\sp}(\R^3),H^{r}_{\sp}(\R^{2}))$ for any $\su,\sp \in \R$ and for all $r \in (0,1/2)$.

It can be easily checked that the range of $\tau_F - \tau_0$ fulfills
\[
\ran(\tau_F-\tau_0) \subseteq \{u \in H^{r}_{\sp}(\R^{2})\,|\,\supp u \subset \supp F \} = \{u \in H^{r}_{\tp}(\R^{2})\,|\,\supp u \subset \supp F \}\]
for all $r \in (0,1/2)$ and all $\tp \in \R$. The r.h.s. of the above equation is a closed subset of $H^{r }(\R^{2})$ which, for any bounded, open subset $B_F\!\subset\! \R^{2}$ such that $\supp F\! \subset\! B_F$, can be isomorphically identified (as a Hilbert space) with $\{u\!\in\! H^{r}_0(B_F)\,|\,\supp u\!\subset\!\supp F\} \subset H^{r}_0(B_F)$. 

Since $B_F$ is bounded, it follows that $H^{r}_0(B_F)$ can be compactly embedded into $H^{r-\eps}_0(B_F)$ (see \cite[p. 99, Th. 16.1]{LioMag}). Since  $r-\eps \in [0,1/2)$,   one can extend elements of $H^{r- \eps}_0(B_F)$  to elements of $H^{r-\eps}_{\tp}(\R^{2})$ which vanish a.e. outside $B_F$; as well known this procedure gives a continuous map sending $H^{r-\eps}_0(B_F)$ into $H^{r-\eps}_{\tp}(\R^{2})$ (see \cite[p. 60, Th. 11.4]{LioMag}, the weight $w_{\tp}$ is inessential since elements are extended by setting them to zero outside $B_F$). The above considerations show that the map $\tau_F - \tau_0 : H^{r+1/2}_{\su,\sp}(\R^3) \to H^{r-\eps}_{\tp}(\R^{2})$ is compact, since it can be obtained by composition of a compact operator with continuous ones.
\endproof

\begin{lemma}\label{GcCompW} Let $\su,\sp\,,\tp \in \R$. Then, for all $r \in (0,1/2)$, $\eps \in(0,r]$ and $z \in \C \backslash [0,+\infty)$ there hold:
\begin{align}
& \Gc_F(z) - \Gc_0(z) \in \Com(H^{r-3/2}_{\su,\sp}(\R^3), H^{r-\eps}_{\tp}(\R^{2})) \;, \label{GcCompWEq0} \\
& G_F(z) - G_0(z) \in \Com(H^{\eps-r}_{\tp}(\R^{2}),H^{3/2-r}_{\su,\sp}(\R^3))\;. \label{GcCompWEq}
\end{align}
\end{lemma}
\proof First of all, let us recall that $\Gc_F(z) - \Gc_0(z) = (\tau_F - \tau_0) \Rf(z)$ by definition (see Eq. \eqref{Gcdef}); then, the claim in Eq. \eqref{GcCompWEq0} follows easily from Lemmata \ref{lemRfW} and \ref{lemtaComW}. Taking this into account, we are able to derive the analogous statement in Eq. \eqref{GcCompWEq} by the usual duality arguments.
\endproof

\begin{prop}\label{PropResComW}[Check of Hypothesis 9 of \cite{Ren}]
Let $\su,\sp,\tu,\tp \in \R$. Then there exists $z_0\in\C\backslash[0,+\infty)$ such that $R_F(z) - R_0(z) \in \Com(L^2_{\su,\sp}(\R^3), L^2_{\tu,\tp}(\R^3))$ for all $z\in\C$ such that $d_z>d_{z_0}$. In particular, there holds 
\begin{equation*}
R_F(z) - R_0(z) \in \Com(L^2(\R^3), L^2_{\tu,\tp}(\R^3)) \quad \mbox{for all  $z \in \C$ such that  $d_z > d_{z_0}$}.
\end{equation*}
\end{prop}
\proof 
The main argument employed in the proof is the fact that compact operators are a two-sided ideal of bounded operators.

We proceed to prove separately the compactness of each of the three addenda appearing on the r.h.s. of Eq. \eqref{ResDiff}. 
By Lemmata \ref{lemtaComW} and \ref{GcCompW}, for all $r\in(0,1/2)$ we infer 
\begin{align}
& \tau_F - \tau_0 \in \Com(H^{r+1/2}_{\su,\sp}(\R^3), L^2_\tp(\R^{2}))\; ; \label{overdrive0}\\
& \Gc_F(z) - \Gc_0(z) \in \Com(H^{r-3/2}_{\su,\sp}(\R^3), L^2_\tp(\R^{2}))\; ; \label{overdrive1}\\
& G_F(z) - G_0(z) \in \Com(L^2_\tp(\R^{2}),H^{3/2-r}_{\su,\sp}(\R^3))\; \label{overdrive2} 
\end{align}
for all $\su,\sp,\tp \in\R$.  Next we analyze the three addenda on the r.h.s. of Eq. \eqref{ResDiff} one by one. 

\textsl{First term: There holds}
\begin{equation*}
G_0(z)\,\Ga^{-1}_0(z)\,\big(\Gc_F(z) - \Gc_0(z)\big) \in \Com(L^2_{\su,\sp}(\R^3), L^2_{\tu,\tp}(\R^3))\;.
\end{equation*} 
The claim follows by noticing that by Eq. \eqref{overdrive1} one has $\Gc_F(z) - \Gc_0(z) \in \Com(L^2_{\su,\sp}(\R^3),L^2_\tp(\R^2))$, in addition to $\Ga_0^{-1}(z) \in \Bou(L^2_\tp(\R^2))$ and $G_0(z) \in \Bou(L^2_{\tp}(\R^2),$ $L^2_{\tu,\tp}(\R^3))$, see  Remark \ref{remR0}. \\ 

\textsl{Second term: There holds}
\begin{equation*}
\big(G_F(z) - G_0(z)\big)\,\Ga^{-1}_F(z)\,\Gc_F(z) \in \Com(L^2_{\su,\sp}(\R^3), L^2_{\tu,\tp}(\R^3))\;.
\end{equation*}
The claim follows by noticing that by Eq. \eqref{overdrive2} one has $G_F(z) - G_0(z) \in \Com(L^2_\sp(\R^2),L^2_{\tu,\tp}(\R^3))$,  in addition to $\Gc_F(z) \in \Bou(L^2(\R^3)_{\su,\sp},$ $ L^2_\sp(\R^2))$ and  $\Ga_F^{-1}(z) \in \Bou(L^2_\sp(\R^2))$, see Lemmata  \ref{lemGF} and \ref{lemGaIF}. \\ 

\textsl{Third term: There holds}
\begin{equation*}
G_0(z)\,\big(\Ga^{-1}_F(z)- \Ga^{-1}_0(z)\big)\,\Gc_F(z) \in \Com(L^2_{\su,\sp}(\R^3), L^2_{\tu,\tp}(\R^3))\;.
\end{equation*}
Let us recall that $\Gt_F(z) = \tau_F G_F(z)$ (and similarly for $F=0 $). We   point out the following basic identity
(we are using essentially the identity $(A+B)^{-1} - A^{-1} = - A^{-1} B (A+B)^{-1}$). 
\begin{align*}
\Ga_F^{-1}(z) - \Ga_0^{-1}(z) = &  \Big(1 + \al \Gt_0(z) + \al \big(\Gt_F(z)-\Gt_0(z)\big)\Big)^{\!-1} - \Big(1 + \al\Gt_0(z)\Big)^{\!-1}  \nonumber \\ 
 = & -\,\al\;\Ga_0^{-1}(z)\, \big(\tau_F G_F(z) -  \tau_0 G_0(z)\big)\, \Ga_F^{-1}(z)\;. %\label{waiting}
\end{align*}
Since 
\begin{equation*}%\label{freak0}
\tau_F G_F(z) -  \tau_0 G_0(z) = (\tau_F - \tau_0) G_F(z) + \tau_0 \big(G_F(z) - G_0(z)\big)\,, 
\end{equation*}
we infer
\begin{equation*}%\label{lose}
\Ga_F^{-1}(z) - \Ga_0^{-1}(z) = -\al\,\Ga_0^{-1}(z)\, \Big((\tau_F - \tau_0) G_F(z) + \tau_0 \big(G_F(z) - G_0(z)\big)\Big)\, \Ga_F^{-1}(z)\,.
\end{equation*}
To proceed, let us recall that: \\ 
$i)$ $\Gc_F(z)\in  \Bou(L^2_{\su,\sp}(\R^3),L^2_\sp(\R^2))$ and  $G_0(z)\in  \Bou(L^2_\tp(\R^2),$ $L^2_{\tu,\tp}(\R^3))$. \\
$ii)$ $\Ga_F^{-1}(z)\in \Bou(L^2_\sp(\R^2))$ and $\Ga_0^{-1}(z) \in \Bou(L^2_\tp(\R^2))$. \\
$iii)$ $G_F(z)\in  \Bou(L^2_\sp(\R^2),H^1_{\su,\sp}(\R^3))$ and $\tau_0 \in  \Bou(H^1_{\tu,\tp}(\R^3),L^2_\tp(\R^2))$.\\
In view of these facts, the thesis follows from  $\tau_F - \tau_0 \in \Com(H^{1}_{\su,\sp}(\R^{3}),L^2_{\tp}(\R^{2}))$ and $G_F(z) - G_0(z) \in \Com(L^2_\sp(\R^{2}),H^{1}_{\tu,\tp}(\R^{3}))$, which are in turn a consequence of Eq.s \eqref{overdrive0} \eqref{overdrive2}.
\endproof

\begin{rem}The latter proposition implies in particular that $R_F(z) - R_0(z) $ is a compact operator in $L^2(\R^3)$ for $z$ far away from the real axis, hence for any $z \in \rho(A_F)\cap \rho(A_0)$. From which we infer $\sigma_{ess}(A_F) = \sigma_{ess}(A_0) = [0,+\infty)$. 
\end{rem}

Finally we check the validity of Hypothesis 8 of \cite{Ren}.  Let us consider the following set of functions $\{\ff_{\k}\}$, indexed by the labels $\k = (k_1,\kp) \in \R\! \times\! \R^2$ (see \cite[p. 85, Eq. (3.4.1)]{AlbBook}):
\begin{equation}
\ff_{\k} := \fu_{k_1} \otimes \fp_{\kp}\;, \label{geneig}
\end{equation}
with 
\begin{equation*}
\fu_{k_1}(x^1) := {e^{i k_1 x^1} \over \sqrt{2\pi}} - {i\,\al \over 2|k_1|+i\,\al}\,{e^{i\, |k_1|\, |x^1|} \over \sqrt{2\pi}}\;, 
\end{equation*}
and 
\begin{equation*}
 \fp_{\kp}(\xp) := {e^{i \kp \cdot \xp} \over 2\pi}\;. 
\end{equation*}

These form a complete set of generalized eigenfunctions for $\AA_0$ with respect to a suitable rigging of $L^2(\R^3)$ (see \cite{AlbBook} and \cite[Ch. VI, Sec. 21]{Naim}). More precisely, it can be checked by direct inspection that $\ff_{\k} \in L^2_{-\su,-\sp}(\R^3)$ for any $\k \in \R^3$ and for all $\su > 1/2$, $\sp > 1$; furthermore, denoting with $\la \;|\;\ra$ the $(L^2_{-\su,-\sp\,},L^2_{\su,\sp})$-duality pairing induced by the $L^2$-inner product, we have $\la \ff_{\k} | (\AA_0 - |\k|^2) f\ra = 0$ for all $f \in \Dom \AA_0 \cap L^2_{\su,\sp}(\R^3)$ with $\AA_0 f \in L^2_{\su,\sp}(\R^3)$.

Keeping in mind the facts mentioned above, in the forthcoming Lemma \ref{lemFk} we proceed to point out some notable features of the functions $\R^3\! \ni \k \mapsto \la \ff_{\k} | f\ra$, for $f \in L^2_{\su,\sp}(\R^3)$. We will later employ these features in the proof of Proposition \ref{PropR0pm}.

\begin{lemma}\label{lemFk} Assume that $\su > 1/2$, $\sp > 1$ and let $f \in L^2_{\su,\sp}(\R^3)$. Then, the map $\R^3\! \ni \k \mapsto \la \ff_{\k} | f\ra$ enjoys the following properties: \\
$i)$ For all $\eta \in (0,1)$ such that $\eta \leqs \min(\su - 1/2,\sp - 1)$, the map $\k \mapsto \la \ff_{\k} | f\ra$ belongs to $C^{0,\eta}(\ov{\R^3})$; moreover,
\begin{equation}
\|\k \mapsto \la \ff_{\k} | f\ra\|_{C^{0,\eta}(\ov{\R^3})} \,\lec\, \|f\|_{L^2_{\su,\sp}(\R^3)} \;. \label{C0ffk}
\end{equation}
$ii)$ Let $\KK \subset (0,+\infty)$ be any compact subset and consider the operators $R^{\pm}_0(\lam)$ of Eq. \eqref{R0lpm} for $\lam \in \KK$. If $R^{+}_{0}(\lam) f = R^{-}_{0}(\lam) f$ for all $\lam \in \KK$, then 
\begin{equation*}
\la \ff_{\k} |f \ra = 0 \quad \mbox{for all $\k \in \R^3$\! such that\, $|\k|^2\! \in \KK$}\;. \label{fk0}
\end{equation*}
\end{lemma}
\proof We discuss separately the proofs of items i) and ii). \\
$i)$ The thesis follows by obvious density arguments as soon as we can infer the relation stated in Eq. \eqref{C0ffk} for any factorized function of the form $f = \uu \otimes \up$, with $\uu \!\in\! L^2_{\su}(\R)$ and $\up\!\in\!L^2_{\sp}(\R^2)$. To this purpose, let us first notice that the Fourier transforms on $\R$ and $\R^2$ induce, respectively, the isomorphisms of Hilbert spaces $\Fou_1 : H^{\su}(\R) \to L^2_{\su}(\R)$ and $\Foup : H^{\sp}(\R^2) \to L^2_{\sp}(\R)$ for any $\su,\sp \in \R$. 
Besides, indicating with $\Te$ the Heaviside step function and using the explicit expression \eqref{geneig} for $\ff_{\k}$, by direct computations we get 
\[  \la \ff_{\k}| \uu \otimes \up\ra =   \bigg[(\Fou_{1}^{-1} \uu)(k_1)   + {i\,\al \over 2|k_1| - i\,\al}\Big(\big(\Fou_{1}^{-1} (\Te \uu)\big)(|k_1|) + \big(\Fou_{1}^{-1} ((1-\Te) \uu)\big)(-|k_1|) \Big) \bigg] (\Foup^{-1}\up)(\kp)\;. 
\]\noindent
Since $\Te \in L^\infty(\R)$, in view of the previously mentioned facts we have $\Fou_{1}^{-1} \uu$, $\Fou_{1}^{-1} (\Te \uu), \Fou_{1}^{-1}((1-\Te) \uu) \in H^{\su}(\R)$ and $\Foup^{-1} \up \in H^{\sp}(\R^2)$; so, by standard Sobolev embeddings, $\Fou_{1}^{-1} \uu, \Fou_{1}^{-1} (\Te \uu), \Fou_{1}^{-1}((1-\Te) \uu) \in C^{0,\etu}(\ov{\R})$ for all $\etu \in (0,1)$ such that $\etu < \su - 1/2$ and $\Foup^{-1} \up \in C^{0,\etp}(\ov{\R^2})$ for all $\etp \in (0,1)$ such that $\etp < \sp - 1$. Let us also remark that the absolute value $k_1 \mapsto |k_1|$ is uniformly Lipschitz-continuous on $\R$ (in fact, $|\,|k_1|-|h_1|\,| \leqs |k_1 - h_1|$ for all $k_1,h_1 \in \R$). Summing up, for any $\eta \in (0,1)$ fulfilling $\eta \leqs \min(\etu,\etp)$ by elementary computations we obtain
\begin{align}
& \sup_{\k \in \R^3} \big|\la \ff_{\k}| \uu\!\otimes\!\up\ra\big| \lec \|\uu\! \otimes \up \|_{L^2_{\su,\sp}(\R^3)}\;, \label{exfoot1} \\
& \big|\la \ff_{\k}| \uu\!\otimes\!\up\ra - \la \ff_{\h}| \uu\!\otimes\!\up\ra\big| \lec \|\uu\! \otimes\up\|_{L^2_{\su,\sp}(\R^3)}\,|\k-\h|^{\eta}\;. \label{exfoot2}
\end{align}
Let us give a few more details about the derivation of the latter estimates. On the one hand, we have
\[\begin{aligned}
&\sup_{\k \in \R^3} \big|\la \ff_{\k}| \uu\!\otimes\!\up\ra\big|  \\
\lec\, & \Big[\|\Fou_{1}^{-1} \uu\|_{C^{0}(\ov{\R})} + \|\Fou_{1}^{-1} (\Te \uu)\|_{C^{0}(\ov{\R})}  \\ 
         & \qquad + \|\Fou_{1}^{-1} ((1-\Te) \uu)\|_{C^{0}(\ov{\R})} \Big] \|\Foup^{-1}\up\|_{C^{0}(\ov{\R^2})} \\ 
\lec\, & \Big[\|\Fou_{1}^{-1} \uu\|_{H^{\su}(\R)} + \|\Fou_{1}^{-1} (\Te \uu)\|_{H^{\su}(\R)}  \\ 
         &\qquad + \|\Fou_{1}^{-1} ((1-\Te) \uu)\|_{H^{\su}(\R)} \Big] \|\Foup^{-1}\up\|_{H^{\sp}(\R^2)} \,\lec\, \|\uu\! \otimes \up \|_{L^2_{\su,\sp}(\R^3)}\;.
\end{aligned}\]
Taking this into account, for $|\k - \h| \geqs 1$ we readily get $|\la \ff_{\k}| \uu \otimes \up\ra - \la \ff_{\h}| \uu \otimes \up\ra| \lec 2\,\sup_{\k \in \R^3} |\la \ff_{\k}| \uu\!\otimes\!\up\ra|$ $\lec \|\uu\! \otimes \up \|_{L^2_{\su,\sp}(\R^3)}\,|\k-\h|^{\eta}$. On the other hand, for $|\k-\h| < 1$ we have
\[\begin{aligned}
& \big|\la \ff_{\k}| \uu\!\otimes\!\up\ra \!-\! \la \ff_{\h}| \uu\!\otimes\!\up\ra\big| \\
\lec\, & \Big[\|\Fou_{1}^{-1}\! \uu\|_{C^{0,\etu}\!(\ov{\R})} \!+\! \|\Fou_{1}^{-1}\! (\Te \uu)\|_{C^{0,\etu}\!(\ov{\R})} \! \\ 
      & + \|\Fou_{1}^{-1}\!((1\!-\!\Te) \uu)\|_{C^{0,\etu}\!(\ov{\R})} \Big] \|\Foup^{-1}\!\up\|_{C^{0,\etp}\!(\ov{\R^2})} \Big(|k_1 \!-\! h_1|^{\etu}\! + |\kp \!-\! \hp|^{\etp} \Big)  \\
\lec\, & \Big[\|\Fou_{1}^{-1}\! \uu\|_{H^{\su}(\R)} \!+\! \|\Fou_{1}^{-1}\! (\Te \uu)\|_{H^{\su}(\R)} \\ 
        & +\|\Fou_{1}^{-1}\!((1\!-\!\Te) \uu)\|_{H^{\su}(\R)} \Big] \|\Foup^{-1}\up\|_{H^{\sp}(\R^2)}\,|\k-\h|^{\min(\etu,\etp)} \\ 
\lec\, & \|\uu\! \otimes\up\|_{L^2_{\su,\sp}(\R^3)}\,|\k-\h|^{\eta}\;.
\end{aligned}\]
Estimates \eqref{exfoot1} and \eqref{exfoot2} allow us to infer the relation \eqref{C0ffk} for $f = \uu \otimes \up$, which yields the thesis.\\
$ii)$ First of all, for all $z \in \C \backslash [0,+\infty)$, from \cite[p. 121, Cor. 2]{Naim} we infer that
\[
(R_0(z)\,f)(\x) = \int_{\R^3} d\k\; {\la \ff_{\k}|f\ra \over |\k|^2 - z}\;\ff_{\k}(\x) \qquad (\x \in \R^3)\;.
\]
On account of Proposition \ref{PropLAP0}, for any $f \in L^2_{\su,\sp}(\R^3)$ with $\su > 1/2$, $\sp > 1$ and for all $\lam \in \KK$, the condition $R^{+}_{0}(\lam) f = R^{-}_{0}(\lam) f$ can be rephrased as follows using the above integral kernel identity:
\[
\begin{aligned}
0 = & \lim_{\eps \downarrow 0}\, \la \big[ R_0(\lam + i \eps) - R_0(\lam - i \eps)\big] f \,|\, f \ra \\ 
 = & \lim_{\eps \downarrow 0}\int_{\R^3}\!\! d\k\; {2i \eps \over (|\k|^2 - \lam)^2 + \eps^2}\; |\la \ff_{\k} | f \ra|^2\;. 
\end{aligned}
\]
Recalling that the function $\k \mapsto \la \ff_{\k} | f \ra$ is, in particular, continuous (see item $i)$ of this Lemma) and denoting  with $d\sigma_{r}$ the spherical measure on the $2$-sphere $S^2_{r} := \{\k \in \R^3\,|\; |\k|^2 = r \}$ induced by the usual Lebesgue measure on $\R^3$, the above relation 
can be rewritten as 
\[
0 =  \lim_{\eps \downarrow 0} \int_{0}^{+\infty}\!\! dr\; {2i \eps \over (r - \lam)^2 + \eps^2} \int_{S^2_r}\!\! d\sigma_{r}(\k)\; |\la \ff_{\k} | f \ra|^2\;. 
\]
and yields  $0 = \!\int_{S^2_{\lam}} d\sigma_{\lam}(\k)\;|\la \ff_{\k} | f \ra|^2$. This suffices to infer that $|\la \ff_{\k} | f \ra|^2 = 0$ for almost every $\k \in S^2_{\lam}$, which in view of the continuity of $\k \mapsto \la \ff_{\k} | f \ra$ implies the thesis.
\endproof

\begin{prop}\label{PropR0pm}[Check of Hypothesis 8 of \cite{Ren}] Assume that $\su > 1$, $\sp > 3/2$ and let $\KK \subset (0,+\infty)$ be any compact subset. Then, for all $\lam \in \KK$ and for all $f \in L^2_{\su,\sp}(\R^3)$ such that $R^{+}_{0}(\lam) f = R^{-}_{0}(\lam) f$, there holds
\begin{equation*}
\|R_0^{\pm}(\lam) f \|_{L^2(\R^3)} \lec \|f\|_{L^2_{\su,\sp}(\R^3)}\;. \label{R0pmEq}
\end{equation*}
\end{prop}
\proof Let us fix arbitrarily $\eps_0 > 0$ and $\KK \subset (0,+\infty)$. Of course, on account of the definition \eqref{R0lpm} of the operators $R_0^{\pm}(\lam)$, the thesis follows as soon as we are able to derive the following uniform bound \eqref{R0Proof} for all $\eps \in (0,\eps_0)$, $\lam \in \KK$ and for some constant $c > 0$ depending only on $\eps_0$ and $\KK$:
\begin{equation}
\|R_0(\lam \pm i \eps) f \|_{L^2(\R^3)} \,\leqs\, c\; \|f\|_{L^2_{\su,\sp}(\R^3)}\;. \label{R0Proof}
\end{equation}
To this purpose, let us first remark that \cite[p. 111, Th. 2]{Naim} yields 
\[\|R_0(\lam \pm i \eps) f\|^2_{L^2(\R^3)} = \int_{\R^3} d\k\,  \frac{|\la \ff_{\k}| f \ra|^2}{\big|\,|\k|^2 - (\lam \pm i\eps)\big|^2}.\] 
Let us introduce the notation  $\R^3 \backslash \KK$ to denote the set of all $\k\in\R^3$ such that $|\k|^2\notin\KK$. Due to the assumption $R^{+}_{0}(\lam) f = R^{-}_{0}(\lam) f$, by item ii) of Lemma \ref{lemFk} we have $\la \ff_{\k}| f \ra = 0$ for all $\k \in \R^3$ with $|\k|^2 \!\in\! \KK$. Therefore, we obtain
\begin{equation}
\|R_0(\lam \pm i \eps) f\|^2_{L^2(\R^3)}\, = \int_{\R^3 \backslash \KK}\!\! d\k\;{|\la \ff_{\k}| f \ra|^2 \over (|\k|^2\! - \lam)^2 + \eps^2}\;. \label{R0eig}
\end{equation}
Assume that $\lam \in \KK$ is an interior point; in this case, we have $\big||\k|^2\! - \lam\big| \geqs \delta$ for some fixed $\delta > 0$ and for all $\k \in \R^3 \backslash \KK$. Thus, taking into account that \cite[p. 121, Cor. 2]{Naim} gives $\|f\|^2_{L^2(\R^3)} = \int_{\R^3} d\k\,|\la \ff_{\k}| f \ra|^2$, from Eq. \eqref{R0eig} we infer
\[
\|R_0(\lam \pm i \eps)\, f\|^2_{L^2(\R^3)}\, \leqs\, {1 \over \delta^2} \int_{\R^3 \backslash \KK}\!\! d\k\;|\la \ff_{\k}| f \ra|^2 \lec\, \|f\|^2_{L^2(\R^3)}\, \lec\, \|f\|^2_{L^2_{\su,\sp}(\R^3)}\;,
\]
where the last inequality obviously holds true for all $\su,\sp \geqs 0$.\\
Next, let us consider the case where $\lam \in \KK$ is a boundary point; in this situation, it is convenient to re-express Eq. \eqref{R0eig} as
\[
\begin{aligned}
\|R_0(\lam \pm i \eps) f\|^2_{L^2(\R^3)}\, = &  \int_{(\R^3 \backslash \KK)\,\cap\, \{||\k|^2 - \lam| \geqs 1\}}\!\!\! d\k\;{|\la \ff_{\k}| f \ra|^2 \over (|\k|^2\! - \lam)^2 + \eps^2}\;  \\ 
 & + \int_{(\R^3 \backslash \KK)\, \cap\, \{||\k|^2 - \lam| < 1\}}\!\!\! d\k\;{|\la \ff_{\k}| f \ra|^2 \over (|\k|^2\! - \lam)^2 + \eps^2}\;. 
\end{aligned}
\]
On the one hand, by computations similar to those described before we easily get 
\[\int_{(\R^3 \backslash \KK) \,\cap\, \{||\k|^2 - \lam| \geqs 1\}} d\k\,  \frac{ |\la \ff_{\k}| f \ra|^2}{(|\k|^2\! - \lam)^2 + \eps^2} \leqs \|f\|^2_{L^2_{\su,\sp}(\R^3)}\,.\]
 On the other hand, let us consider the unit vector $\hat{\k} := \k/|\k|$ and notice that item ii) of Lemma \ref{lemFk} implies $\la \ff_{\sqrt{\lam}\,\hat{\k}}|f \ra = 0$ for all $\lam \in \KK$; so, for all $\su > 1/2$, $\sp > 1$ and $\eta \in (0,1)$ with $\eta \leqs \min(\su-1/2,\sp - 1)$, by item i) of the same Lemma (see, in particular, Eq. \eqref{C0ffk}) we get $|\la \ff_{\k}|f \ra| = |\la \ff_{\k}|f \ra - \la \ff_{\sqrt{\lam}\,\hat{\k}}|f \ra| \lec \big||\k| \!-\! \sqrt{\lam}\big|^{\eta} \|f\|_{L^2_{\su,\sp}(\R^3)}$\,. Taking this into account, we infer that 
\[
 \begin{aligned}
& \int_{(\R^3 \backslash \KK) \cap \{||\k|^2 - \lam| < 1\}}\! d\k\; {|\la \ff_{\k}| f \ra|^2 \over \big((|\k|^2\! - \lam)^2 + \eps^2\big)} \\ 
\lec\, & \left(\int_{(\R^3 \backslash \KK)\, \cap \,\{||\k|^2 - \lam| < 1\}}\! d\k\; {\big|\,|\k| - \sqrt{\lam}\,\big|^{2\eta} \over (|\k|^2\! - \lam)^2}\right) \|f\|^2_{L^2_{\su,\sp}(\R^3)}\;.
\end{aligned}
\]
Under the stronger hypotheses $\su > 1$, $\sp > 3/2$ and $1/2 < \eta \leqs \min(\su-1/2,\sp-1)$, the latter relation yields
\[\begin{aligned}
& \int_{(\R^3 \backslash \KK) \cap \{||\k|^2 - \lam| < 1\}}\!\!\! d\k\;{|\la \ff_{\k}| f \ra|^2 \over (|\k|^2\! - \lam)^2 + \eps^2}  \\ 
\lec & \l(\int_{\{|r^2-\lam| < 1\}}\! dr\;{1 \over |r-\sqrt{\lam}|^{2-2\eta}}\r) \|f\|^2_{L^2_{\su,\sp}(\R^3)} \lec \|f\|^2_{L^2_{\su,\sp}(\R^3)} \;. 
\end{aligned}
\]
Summing up, the arguments described above prove Eq. \eqref{R0Proof} for all $\lam \in \KK$ (either in the interior or on boundary), thus implying the thesis.
\endproof

We are now ready to state the main result of this subsection.

\begin{thm}[LAP for $A_F$]\label{ThmLAPF}
Let $\siu > 1/2$, $\sip > 3/4$, and denote by $\sigma_p^+(A_F):= (0,+\infty)\cap \sigma_p(A_F)$ the (possibly empty) set of embedded eigenvalues of $A_F$. 

Then  $\sigma_p^+(A_F)$ is a discrete set and for any $\lam \in (0,+\infty)\backslash \sigma_p^+(A_F)$, the limits
\begin{equation*}
R^{\pm}_{F}(\lam) := \lim_{\eps \downarrow 0} R_F(\lam \pm i \eps) \label{RFlpm}
\end{equation*}
exist in $\Bou(L^2_{\siu,\sip}(\R^3),L^2_{-\siu,-\sip}(\R^3))$ and the convergence is uniform in any compact subset $\KK \subset (0,+\infty)\backslash\sigma_p^+(A_F)$.
\end{thm}
\proof As already anticipated, the results derived previously in this work allow us to infer the thesis by a straightforward application of \cite[Th. 7]{Ren}; in the following we give more details about this claim.

Firstly, Proposition \ref{PropH1RF} shows that $R_0(z),R_F(z) \in \Bou(L^2_{\siu,\sip}(\R^3))$ for any $\siu\,,\sip \in \R$; this and the previously discussed properties of the corresponding operators $\AA_0,\AA_F$ prove Hypothesis 1 of \cite{Ren}. \\
Secondly, for all $\siu,\sip > 1/2$, Proposition \ref{PropLAP0} grants the existence of the limits $R^{\pm}_{0}(\lam) := \lim_{\eps \downarrow 0} R_0(\lam \pm i \eps)$ in $\Bou(L^2_{\siu,\sip}(\R^3),L^2_{-\siu,-\sip}(\R^3))$; furthermore, under the stronger assumptions $\siu > 1/2$, $\sip > 3/4$, by Proposition \ref{PropR0pm} we have $\|R^{\pm}_0(\lam) f\|_{L^2(\R^3)} \lec \|f\|_{L^2_{2\siu,2\sip}(\R^3)}$ for all $u\!\in\!L^2_{2\siu,2\sip}(\R^3)$ such that $R^{+}_0(\lam)f \!=\! R^{-}_0(\lam)f$. The above remarks prove that Hypothesis 8 of \cite{Ren} holds true.

Thirdly, Proposition \ref{PropResComW} implies, in particular,
\[R_F(z) - R_0(z) \in  \Com(L^2(\R^3), L^2_{2\siu + \ga,2\sip + \ga}(\R^3))\]
for any $\siu,\sip \in \R$ and $\ga > 0$, thus yielding Hypothesis 9 of \cite{Ren}.

Taking into account \cite[Prop. 10]{Ren}, the above arguments suffice to infer that all the hypotheses of \cite[Th. 7]{Ren} are fulfilled by the pair $\AA_0,\AA_F$.
\endproof

\section{An alternative formula for the resolvents difference\label{s:altres}}
The aim of this section is to re-write the resolvent difference $R_F(z) -  R_0(z)$ in a way that is convenient to get existence and completeness of the wave operators for the scattering couple $(A_{F}, A_{0})$  and to get a representation of the corresponding scattering matrix. In particular we show that $R_F(z) -  R_0(z)$ depends only on the restriction of the traces $\tau_{F}$ and $\tau_{0}$ to $\mbox{\rm supp}(F)$.\par
Let us recall the resolvent formulae that we obtained for the two self-adjoint operators $A_{0}$ and $A_{F}$, respectively (see Eq.s \eqref{R0} \eqref{RG}):
\begin{align}
R_0(z) := \Rf(z) - \al\,G_0(z)\,\Gamma_0^{-1}(z)\,\breve G_0(z)\;, \label{R0_A} \\
R_F(z) := \Rf(z) - \al\,G_F(z)\,\Gamma_F^{-1}(z)\,\breve G_F(z) \,, \label{RF_A}
\end{align}
where
$$
G_0(z)\in\Bou(L^{2}(\R^{2}),L^{2}(\R^{3}))\,,\qquad G_F(z)\in\Bou(L^{2}(\R^{2}),L^{2}(\R^{3}))\,,
$$
$$
\breve G_0(z)\in\Bou(L^{2}(\R^{3}),L^{2}(\R^{2}))\,,\qquad \breve G_F(z)\in\Bou(L^{2}(\R^{3}),L^{2}(\R^{2}))\,,
$$
\begin{equation}\label{marble0_A}
\Ga_0(z) := \big(1 + \al\,\tau_0 G_0(z)\big)\in\Bou(L^{2}(\R^{2}))\,, \qquad \Ga_F(z) :=\big( 1 + \al\,\tau_F G_F(z)\big)\in\Bou(L^{2}(\R^{2})) \,,
\end{equation}
$$
\Ga_0^{-1}(z) \in\Bou(L^{2}(\R^{2}))\,, \qquad \Ga_F^{-1}(z)\in\Bou(L^{2}(\R^{2})) \,.
$$\noindent
Here and below $z$ is any point in $\C\backslash[0,+\infty)$. %(see Eq.  \eqref{GL2}).
The resolvent formula in the following Theorem resembles the one used in \cite{EK05}:
\begin{thm}\label{TS} The resolvent difference $R_{F}(z)-R_{0}(z)$ depends only on the restriction of the traces $\tau_{F}$ and $\tau_{0}$ to $\Sigma:=\mbox{\rm supp}(F)$: 
\begin{align}\label{RS_A}
R_{F}(z)
=R_{0}(z)+ g_{\Sigma}(z)\Lambda_{F,\Sigma}(z)\breve g_{\Sigma}(z)\,, \qquad 
z\in \C\backslash[0,+\infty)\,
\end{align}
where 
$$
\breve g_{\Sigma}(z):L^{2}(\R^{3})\to L^{2}(\Sigma)\oplus L^{2}(\Sigma)\,,\quad \breve g_{\Sigma}(z):=\tau_{\Sigma}R_{0}(z)\,,
$$
$$
\tau_{\Sigma}\!:\!H^{r+1/2}(\R^{3})\!\to \!L^{2}(\Sigma)\oplus L^{2}(\Sigma)\,,\!\quad\tau_{\Sigma} f\!:=(\tau_{F}f |\Sigma)\oplus(\tau_{0}f|\Sigma) \quad \mbox{for $r>0 $}\,,
$$
$$
g_{\Sigma}(z):L^{2}(\Sigma)\oplus L^{2}(\Sigma)\to L^{2}(\R^{3})
\,,\quad g_{\Sigma}(z):=(g_{\Sigma}(\bar z))^{*}\,,
$$
$$
\Lambda_{F,\Sigma}(z):L^{2}(\Sigma)\oplus L^{2}(\Sigma)\to L^{2}(\Sigma)\oplus L^{2}(\Sigma)\,,\quad 
$$
$$\Lambda_{F,\Sigma}(z):= - \,\alpha\, \big(1+\alpha J\tau_{\Sigma}g_{\Sigma}(z)\big)^{-1}J\,,\quad
J =  \begin{pmatrix}
 1 & 0 \\
 0 & -1  
 \end{pmatrix} \,.$$
\end{thm}
\begin{proof}
Combining \eqref{R0_A} and \eqref{RF_A}, one gets (here and below we use the matrix block operator notation) 
\begin{equation}\label{marble-1_A}
 R_F(z) -  R_0(z)  = -\alpha\, \big( G_F(z),\, G_0(z)\big)\!
 \begin{pmatrix}
 \Gamma_F^{-1}(z) & 0 \\ 
 0 & -\Gamma_0^{-1}(z)
 \end{pmatrix} \!
 \begin{pmatrix}
 \breve G_F(z) \\ \breve G_0(z)
 \end{pmatrix} .
\end{equation}
Defining
\begin{equation*}
\breve g_0(z) := \tau_0 R_0(z) \,,\qquad \breve g_F(z) := \tau_F R_0(z) 
\end{equation*}
and 
\begin{equation*}
 g_0(z) :=\breve  g_0(\bar z)^*\,, \qquad  g_F(z) := \breve g_F(\bar z)^* \,,
\end{equation*}
one has
\begin{equation*}
\breve g_0(z) = \Gamma_0^{-1}(z)\,  \breve G_0(z)\,, \qquad g_0(z) = G_0(z)\, \Gamma_0^{-1}(z)
\end{equation*}
and so 
\begin{equation}\label{marble1_A}
\breve G_0(z) = \Gamma_0(z)\, \breve g_0(z)\,, \qquad G_0(z) = g_0(z)\, \Gamma_0(z).
\end{equation}
From the latter identity, together with the definition of $\Gamma_0(z)$ given above, one infers that 
$$
\gamma_0(z) := \big(1-\alpha \, \tau_0 g_0(z)\big)\in\Bou(L^{2}(\R^{2}))
$$
has a bounded inverse:
\begin{equation}\label{marble1.1_A}
 \gamma_0^{-1}(z)=\Gamma_0(z)\,.
\end{equation}
Hence, by Eq. \eqref{R0_A}, 
\begin{equation*}
 \Rf(z) = R_0(z)+  \al\,g_0(z)\,\gamma_0^{-1}(z)\,\breve g_0(z) .
\end{equation*}
By applying $\tau_F$ to the latter equation it follows that 
\begin{equation}\label{marble2_A}
\breve G_F(z) = \tau_F  \Rf(z) = \breve g_F(z)+  \al\, \tau_F g_0(z)\,\gamma_0^{-1}(z)\,\breve g_0(z) 
\end{equation}
and, by taking the adjoint (in $\bar z$), 
\begin{equation}\label{marble3_A}
G_F(z) = g_F(z)+  \al\,  g_0(z)  \,\gamma_0^{-1}(z) \, \tau_0 g_F(z),
\end{equation}
where we used $ (\tau_F g_0(\bar z))^* = (\breve g_F(\bar z) \tau_0^*)^* = \tau_0 g_F(z)$. 
\par
By Eq.s \eqref{marble1_A}, \eqref{marble2_A} and \eqref{marble3_A} (together with Eq. \eqref{marble1.1_A}) it follows that 
\begin{equation*} 
\begin{aligned}
\begin{pmatrix}
 \breve G_F(z) \\ \breve G_0(z)
 \end{pmatrix} = & 
 \begin{pmatrix}
\breve g_F(z)+  \al\, \tau_F g_0(z)\,\gamma_0^{-1}(z)\,\breve g_0(z) \\ \gamma_0^{-1}(z) \,\breve g_0(z)
 \end{pmatrix} \\ 
 = &   
 \begin{pmatrix}
1 &  \al\, \tau_F g_0(z)\,\gamma_0^{-1}(z) \\ 0 &  \gamma_0^{-1}(z)
 \end{pmatrix}
\begin{pmatrix}
\breve g_F(z) \\ \breve g_0(z)
 \end{pmatrix}.
 \end{aligned}
\end{equation*}
Similarly, for $ (G_F(z),\, G_0(z))$ one has that 
\begin{align*}
 (G_F(z),\, G_0(z)) = & \,\big(g_F(z)+  \al\,  g_0(z)  \,\gamma_0^{-1}(z) \, \tau_0 g_F(z), g_0(z) \gamma_0^{-1}(z) \big) \\ 
 = & \,\big(g_F(z), g_0(z) \big) 
 \begin{pmatrix}
  1&   0 \\ 
   \al \,\gamma_0^{-1}(z) \, \tau_0 g_F(z) &  \gamma_0^{-1}(z)
 \end{pmatrix}.
\end{align*}
By Eq.s \eqref{marble0_A} and \eqref{marble3_A}, it follows that 
\begin{equation*}
\Gamma_F(z) =   1+ \alpha \, \tau_F g_F(z) + \alpha^2 \, \tau_F  g_0(z)  \,\gamma_0^{-1}(z) \, \tau_0 g_F(z).
\end{equation*}
Now, let us set 
\begin{equation*}
{\Ma_1}(z)= 1 + \alpha \, \tau_F g_F(z)  \,,\quad \quad {\Ma_2}(z) = \al\, \tau_F g_0(z) \,,
\end{equation*}
\begin{equation*}
{\Ma_3}(z) =    \al \, \tau_0 g_F(z) \,,\quad{\Ma_4}(z)= - \gamma_0 (z)\,. 
\end{equation*}
With this notation, one has 
\begin{equation*}
\Gamma_F(z)=  {\Ma_1}(z) - {\Ma_2}(z) {\Ma_4}^{-1}(z){\Ma_3}(z)\,, \qquad  \Gamma_0(z) = - {\Ma_4}^{-1}(z)
\end{equation*}
and (here we temporarily omit the dependence on $z$) 
\begin{equation*}\begin{aligned}
& \begin{pmatrix}
  1&   0 \\ 
   \al \,\gamma_0^{-1} \, \tau_0 g_F &  \gamma_0^{-1}
 \end{pmatrix}
\begin{pmatrix}
 \Gamma_F^{-1} & 0 \\ 
 0 & -\Gamma_0^{-1}
 \end{pmatrix}
 \begin{pmatrix}
1 &  \al\, \tau_F g_0\,\gamma_0^{-1} \\ 0 &  \gamma_0^{-1}
 \end{pmatrix}  \\
 = &
 \begin{pmatrix}
  1&   0 \\ 
 -{\Ma_4}^{-1} {\Ma_3} & - {\Ma_4}^{-1}
 \end{pmatrix} 
\begin{pmatrix}
\big(   {\Ma_1} - {\Ma_2} {\Ma_4}^{-1}{\Ma_3} \big)^{-1} & 0 \\ 
 0 & {\Ma_4}
 \end{pmatrix}
 \begin{pmatrix}
1 & - {\Ma_2} {\Ma_4}^{-1}  \\ 0 &-{\Ma_4}^{-1}
 \end{pmatrix}\\ 
 = &
 \begin{pmatrix}
\big(   {\Ma_1} - {\Ma_2} {\Ma_4}^{-1}{\Ma_3} \big)^{-1}  & - \big(   {\Ma_1} - {\Ma_2} {\Ma_4}^{-1}{\Ma_3} \big)^{-1}   {\Ma_2} {\Ma_4}^{-1}   \\ 
 -{\Ma_4}^{-1} {\Ma_3}\big(   {\Ma_1} - {\Ma_2} {\Ma_4}^{-1}{\Ma_3} \big)^{-1}   &{\Ma_4}^{-1} {\Ma_3}\big(   {\Ma_1} - {\Ma_2} {\Ma_4}^{-1}{\Ma_3} \big)^{-1}   {\Ma_2} {\Ma_4}^{-1} +  {\Ma_4}^{-1}
 \end{pmatrix} \\ 
  = &
 \begin{pmatrix}
 {\Ma_1}&   {\Ma_2}  \\
 {\Ma_3} & {\Ma_4} 
 \end{pmatrix}^{-1} 
  =
 \begin{pmatrix}
 1 + \alpha \, \tau_F g_F & \al\, \tau_F g_0 \\
 \al \, \tau_0 g_F  & -( 1 - \alpha \, \tau_0 g_0)
 \end{pmatrix}^{-1}\,,
 \end{aligned}
 \end{equation*}\noindent
where  we used the inversion  formula for block matrices.  Hence we can set 
\begin{equation*}
\gamma(z) :=\frac{J}{\alpha}
  +
  \begin{pmatrix}
  \tau_F g_F(z) &  \tau_F g_0(z) \\
  \tau_0 g_F(z)  &  \tau_0 g_0(z)
 \end{pmatrix} \,;\qquad J =  \begin{pmatrix}
 1 & 0 \\
 0 & -1  
 \end{pmatrix} \,,
 \end{equation*}
 and we have that 
 $$
 \gamma(z)\in\Bou\big(L^{2}(\R^{2})\oplus L^{2}(\R^{2})\big)
 $$
 has a bounded inverse
\begin{equation*}\gamma^{-1}(z)   = \alpha
  \begin{pmatrix}
  1&   0 \\ 
   \al \,\gamma_0(z)^{-1} \, \tau_0 g_F(z) &  \gamma_0(z)^{-1}
 \end{pmatrix}
\begin{pmatrix}
 \Gamma_F(z)^{-1} & 0 \\ 
 0 & -\Gamma_0(z)^{-1}
 \end{pmatrix}
 \begin{pmatrix}
1 &  \al\, \tau_F g_0(z)\,\gamma_0(z)^{-1} \\ 0 &  \gamma_0(z)^{-1}
 \end{pmatrix}.
 \end{equation*}\noindent
Going back to the formula for the resolvent difference, see Eq. \eqref{marble-1_A}, we set 
 \begin{equation*}
g(z) :=  \big(g_F(z),\, g_0(z)\big) \,,\qquad  
  \breve g(z) := \begin{pmatrix}
 \breve g_F(z) \\ \breve g_0(z)
 \end{pmatrix}
 \end{equation*}
and we have that
\begin{equation*}
 R_F(z) -  R_0(z) = -\, g(z) \, \gamma^{-1}(z) \,  \breve g(z)\;. 
\end{equation*}
We remark that all the quantities at the r.h.s. of the latter equation are written in terms of the operator $R_0(z)$ and of the traces $\tau_0$ and $\tau_F$. 
\par
Since $\gamma(z)$ is a continuous bijection from $L^{2}(\R^{2})\oplus L^{2}(\R^{2})$ onto itself,
from now on we work in the following setting:
$$
\tau:=\tau_{F}\oplus\tau_{0}\in\Bou\big(H^{s}(\R^{3}), L^{2}(\R^{2})\oplus L^{2}(\R^{2})\big)\,,\qquad s>1/2\,,
$$
$$\breve g(z):=\tau R_{0}(z)\in \Bou\big(L^{2}(\R^{3}), L^{2}(\R^{2})\oplus L^{2}(\R^{2})\big)\,,$$ 
$$ g(z):=\breve g(\bar z)^{*}\in \Bou\big(L^{2}(\R^{2})\oplus L^{2}(\R^{2}), L^{2}(\R^{3})\big)\,,$$ 
$$
\tau g(z)\in\Bou\big(L^{2}(\R^{2})\oplus L^{2}(\R^{2})\big)\,.
$$
We freely use the identifications 
\[\begin{aligned}
L^{2}(\R^{2})\oplus L^{2}(\R^{2})\equiv\, & L^{2}(\Sigma)\oplus L^{2}(\Sigma^{c})\oplus 
L^{2}(\Sigma)\oplus L^{2}(\Sigma^{c}) \\ 
\equiv\, & L^{2}(\Sigma)\oplus L^{2}(\Sigma)\oplus 
L^{2}(\Sigma^{c})\oplus L^{2}(\Sigma^{c})\,.
\end{aligned}
\]
Let us now introduce, in the component 
$L^{2}(\Sigma^{c})\oplus L^{2}(\Sigma^{c})$, the orthogonal projections 
$$
P_{0}(\phi\oplus\psi):=\frac{\phi+\psi}2\oplus \frac{\phi+\psi}2\,,\qquad Q_{0}:=\uno-P_{0}\,.
$$
This induces the further decomposition 
\begin{equation}\label{decomp_A}
L^{2}(\R^{2})\oplus L^{2}(\R^{2})\equiv  L^{2}(\Sigma)\oplus L^{2}(\Sigma)\oplus 
\text{Range}(P_{0})\oplus \text{Range}(Q_{0})\,.
\end{equation}
We then define the orthogonal projectors in $L^{2}(\R^{2})\oplus L^{2}(\R^{2})$
$$
P=\uno\oplus P_{0}\,,\qquad Q=\uno\oplus Q_{0}\,.
$$
From now on a vector  $\phi_{F}\oplus\phi_{0}\in L^{2}(\R^{2})\oplus L^{2}(\R^{2})$ will be 
decomposed as $$\phi_{F}\oplus\phi_{0}\equiv\phi_{\circ}\oplus\phi_{+}\oplus\phi_{-}\,,$$ where 
$$\phi_{\circ}=\phi_{F}^{\Sigma}\oplus\phi_{0}^{\Sigma}\,,$$
$$
\phi_{+}:=P_{0}(\phi_{F}^{\Sigma^{c}}\oplus
\phi_{0}^{\Sigma^{c}})\,,\quad\phi_{-}:=Q_{0}(\phi_{F}^{\Sigma^{c}}\oplus
\phi_{0}^{\Sigma^{c}})\,,$$
$$
\phi_{F/0}^{\Sigma/\Sigma^{c}}:=\phi_{F/0}|\Sigma/\Sigma^{c}\,,
$$
and we used the notation $\phi|\Sigma$ (resp. $\phi|\Sigma^{c}$) to denote the restriction of the function $\phi$ to the set $\Sigma$ (resp. $\Sigma^c$). 

We also introduce the decompositions
$$
\tau\equiv \tau_{\Sigma}\oplus\tau_{\Sigma^{c}}\equiv P\tau\oplus Q_{0}\tau_{\Sigma^{c}}\equiv\tau_{\Sigma}\oplus P_{0}\tau_{\Sigma^{c}}\oplus Q_{0}\tau_{\Sigma^{c}}
$$
where $\tau_{\Sigma/\Sigma^{c}} f:=\tau f|\Sigma/\Sigma^{c}$. This induces the further decompositions 
$$\breve g(z)\equiv\breve g_{\Sigma}(z)\oplus\breve g_{\Sigma^{c}}(z)
\equiv P\breve g(z)\oplus Q_{0}\breve g_{\Sigma^{c}}(z)
\equiv 
\breve g_{\Sigma}(z)\oplus P_{0}\breve g_{\Sigma^{c}}(z)\oplus Q_{0}\breve g_{\Sigma^{c}}(z)
\,,
$$
where $\breve g_{\Sigma/\Sigma^{c}}(z) f:=(\tau R_{0}(z)f)|{\Sigma/\Sigma^{c}}$.\par
Since $\tau_{F}f|\Sigma^{c}=\tau_{0}f|\Sigma^{c}$, one has $Q_{0}\tau_{\Sigma^{c}}=Q_{0}\breve g_{\Sigma^{c}}(z)=0$
and so 
$$\text{Range}(\tau)\subseteq \text {Range}(P)\,,\qquad 
\text{Range}(\breve g(z))\subseteq \text {Range}(P)\,.
$$
Thus $P\breve g(z)=\breve g(z)$; by duality $g(z)=g(z)P$ and so $g(z)(\zero\oplus\zero\oplus Q_{0})=0$. Moreover
$$
\tau g(z)=P\tau g(z)P\,.
$$
Equivalently, using block operator matrix notation with respect to the decomposition \eqref{decomp_A} and setting $g_{\Sigma/\Sigma^{c}}(z):=\breve g_{\Sigma/\Sigma^{c}}(\bar z)^{*}$, one has
$$
\tau g(z)=\begin{pmatrix}\tau_{\Sigma}g_{\Sigma}(z)&\tau_{\Sigma}g_{\Sigma^{c}}(z)&\zero\\
\tau_{\Sigma^{c}}g_{\Sigma}(z)&\tau_{\Sigma^{c}}g_{\Sigma^{c}}(z)&\zero\\
\zero&\zero&\zero
\end{pmatrix}\,.
$$
Then, since
$$
J:\text{Range}(P)\to \text{Range}(Q)\,,\qquad J:\text{Range}(Q)\to \text{Range}(P)\,,
$$ 
one gets 
\begin{equation}\label{gamma_A}
\gamma(z)=\begin{pmatrix}
{J}/{\alpha}+\tau_{\Sigma}g_{\Sigma}(z)&\tau_{\Sigma}g_{\Sigma^{c}}(z)&\zero\\
\tau_{\Sigma^{c}}g_{\Sigma}(z)&\tau_{\Sigma^{c}}g_{\Sigma^{c}}(z)&{J}/\alpha\\
\zero&{J}/\alpha&\zero
\end{pmatrix}\,.
\end{equation}
Since $\gamma(z)$ is surjective, for any given $\psi_{\circ}\in L^{2}(\Sigma)\oplus L^{2}(\Sigma)$ there exists $\phi_{\circ}\oplus\phi_{+}\oplus\phi_{-}\equiv\phi_{F}\oplus\phi_{0}\in L^{2}(\R^{2})\oplus L^{2}(\R^{2})$ such that 
$$
\psi_{\circ}\oplus0 \oplus 0=\gamma(z)(\phi_{\circ}\oplus\phi_{+}\oplus\phi_{-})\,.
$$
By Eq. \eqref{gamma_A}, since $J$ is injective, one gets $\phi_{+}=0$ and $\psi_{\circ}=({J}/{\alpha}+\tau_{\Sigma}g_{\Sigma}(z))\phi_{\circ}$. Therefore ${J}/{\alpha}+\tau_{\Sigma}g_{\Sigma}(z)$ is surjective. Since ${J}/{\alpha}+\tau_{\Sigma}g_{\Sigma}(z)=
({J}/{\alpha}+\tau_{\Sigma}g_{\Sigma}(\bar z))^{*}$, 
$\mbox{ker}({J}/{\alpha}+\tau_{\Sigma}g_{\Sigma}(z))=\ran({J}/{\alpha}+\tau_{\Sigma}g_{\Sigma}(\bar z))^{\perp}=\{0\}$. Hence  ${J}/{\alpha}+\tau_{\Sigma}g_{\Sigma}(z)$ is a continuous bijection and so, by the inverse mapping theorem, has a bounded inverse. We introduce the notation 
\[
\begin{aligned}
\Lambda_{F,\Sigma}(z):= & - \big({J}/{\alpha}+\tau_{\Sigma}g_{\Sigma}(z)\big)^{-1}  \\ 
= & -\alpha\,\big(1+\alpha J\tau_{\Sigma}g_{\Sigma}(z)\big)^{-1}{J}\in\Bou(L^{2}(\Sigma)\oplus L^{2}(\Sigma))\,.
\end{aligned}
\]
Therefore, using again the inversion formula for block operator matrices, 
\begin{equation*}
\gamma(z)^{-1} 
= \begin{pmatrix}
- \Lambda_{F,\Sigma}(z)&\zero&\alpha \Lambda_{F,\Sigma}(z)\tau_{\Sigma}g_{\Sigma^{c}}(z)J\\
\zero&\zero&\alpha J\\
\alpha J\tau_{\Sigma^{c}}g_{\Sigma}(z) \Lambda_{F,\Sigma}(z) &\alpha J&-\alpha^{2}J\left(\tau_{\Sigma^{c}}g_{\Sigma}(z) \Lambda_{F,\Sigma}(z)\tau_{\Sigma}g_{\Sigma^{c}}(z)+ \tau_{\Sigma^{c}}g_{\Sigma^{c}}(z)\right){J}
\end{pmatrix}.
\end{equation*}\noindent
In conclusion,
\[
g(z)\,\gamma(z)^{-1}\,\breve g(z)=  - g_{\Sigma}(z)\,\Lambda_{F,\Sigma}(z)\,\breve g_{\Sigma}(z)
\]
and the proof is concluded.
\end{proof}

\section{Existence and asymptotic completeness of the wave operators\label{s:wo}}
The aim of this section is to prove the existence and completeness of the wave operators relative to the couple $(A_F,A_0)$. We make use of the resolvent formula given in Theorem \ref{TS} together with \cite[Th. 2.8]{MP} and LAP.
 Recall that $ \si(A_0) =  \si_{ac}(A_0) =[0,+\infty)$.
\begin{thm}\label{CorWd22} Let $A_F,A_0$ be defined as above.  Then,
the wave operators
\begin{equation*}\label{waveop1}
\dd{W_{\pm}(A_F,A_0) := \slim_{t \to \pm \infty} e^{it A_F}e^{-it A_0} } \,,
\end{equation*}
\begin{equation*}
\dd{W_{\pm}(A_0,A_F) := \slim_{t \to \pm \infty} e^{it A_0}e^{-it A_F} P_{ac}(A_F)} ,
\end{equation*}
where $P_{ac}(A_F)$ denotes the orthogonal projector on the absolutely continuous subspace relative to $A_F$,  exist and are asymptotically complete, i.e., are complete and $\sigma_{sc}(A_{F})=\emptyset$.
\end{thm}
\begin{proof} Taking into account the resolvent formula \eqref{RS_A}, according to \cite[Th. 2.8 and Rem. 2.1]{MP}, in order to get existence and completeness of the wave operators one needs to show that 
\begin{equation}\label{MP1}
\sup_{(\lambda,\epsilon)\in I\times(0,1)}\sqrt{\epsilon}\,\|g_{\Sigma}(\lambda\pm i\epsilon)\|_{\Bou(L^{2}(\Sigma)\oplus L^{2}(\Sigma),L^{2}(\R^{3}))}<+\infty
\end{equation}
and
\begin{equation}
\label{MP2}
\sup_{(\lambda,\epsilon)\in I\times(0,1)}\|\Lambda_{F,\Sigma}(\lambda\pm i\epsilon)\|_{\Bou(L^{2}(\Sigma)\oplus L^{2}(\Sigma))}<+\infty\,,
\end{equation}
for any open and bounded interval $I$ such that $\overline I\subset \R\backslash(\{0\}\cup \sigma^{+}_{p}(A_{F}))$. Since $z\mapsto g_{\Sigma}(z)$ and $z\mapsto \Lambda_{F,\Sigma}(z)$ are continuous on $\rho(A_{0})=\C\backslash[0,+\infty)$, it suffices to prove \eqref{MP1} and \eqref{MP2} whenever 
$\overline I\subset (0,+\infty)\backslash\sigma^{+}_{p}(A_{F})$.\par
Let $\mu\in\rho(A_{0})\cap\R=(-\infty,0)$; by the inequality (see \cite[Eq. (3.16)]{MP})
 \begin{align*}
&\epsilon\,\|g_{\Sigma}(\lambda\pm i\epsilon)\|_{\Bou(L^2(\Sigma)\oplus L^2(\Sigma),L^{2}(\R^{3}))}^{2}\\
\leqs &\; \frac12\,\big(|\mu-\lambda|+\epsilon\big)\|g_{\Sigma}(\mu)\|_{\Bou(L^2(\Sigma)\oplus L^2(\Sigma),L_{\sigma_{1},\sigma_{\parallel}}^{2}(\R^{3}))}\times\\
&\times\Big(
\|\breve g_{\Sigma}(\lambda+i\epsilon)\|_{\Bou(L^{2}_{\sigma_{1},\sigma_{\parallel}}(\R^{3}),{L^2(\Sigma)\oplus L^2(\Sigma))}} +
\|\breve g_{\Sigma}(\lambda-i\epsilon)\|_{\Bou(L^{2}_{\sigma_{1},\sigma_{\parallel}}(\R^{3}),{L^2(\Sigma)\oplus L^2(\Sigma))}}\Big)\\
\leqs &\, \frac12\,\big(|\mu-\lambda|+\epsilon\big)\|g_{\Sigma}(\mu)\|_{\Bou(L^2(\Sigma)\oplus L^2(\Sigma),L_{\sigma_{1},\sigma_{\parallel}}^{2}(\R^{3}))}\|\tau_{\Sigma}\|_{\Bou(H^{1}_{-\sigma_{1},-\sigma_{\parallel}}(\R^{3}),{L^2(\Sigma)\oplus L^2(\Sigma))}}\times\\
&\times\Big(
\|R_{0}(\lambda+i\epsilon)\|_{\Bou(L^{2}_{\sigma_{1},\sigma_{\parallel}}(\R^{3}),H^{1}_{-\sigma_{1},-\sigma_{\parallel}}(\R^{3}))}  +
\|R_{0}(\lambda-i\epsilon)\|_{\Bou(L^{2}_{\sigma_{1},\sigma_{\parallel}}(\R^{3}),H^{1}_{-\sigma_{1},-\sigma_{\parallel}}(\R^{3}))}\Big)\,,
\end{align*}\noindent
inequality \eqref{MP1} is consequence of LAP for $A_{0}$ (see Proposition \ref{PropLAP0}).\par
By the relation (see \cite[Eq. (4.2) and Lem. 4.2]{MP}, see also \cite[Eq. (2.9)]{CFPinv})
\begin{equation*}
\Lambda_{F,\Sigma}(\lambda\pm i\epsilon)\; = \Lambda_{F,\Sigma}(\mu)\big[1+(\lambda-\mu\pm i\epsilon)\breve g_{\Sigma}(\mu)\big(1-(\lambda-\mu\pm i\epsilon)R_{F}(\lambda\pm i\epsilon)\big)g_{\Sigma}(\mu)\Lambda_{F,\Sigma}(\mu)\big]\,,
\end{equation*}
\noindent
one gets the inequality (here we set $\nu:=|\lambda-\mu|+1$)
\begin{align*}
&\|\Lambda_{F,\Sigma}(\lambda\pm i\epsilon)\|_{\Bou(L^2(\Sigma)\oplus L^2(\Sigma))} \\ 
\leqs & \|\Lambda_{F,\Sigma}(\mu)\|_{\Bou(L^2(\Sigma)\oplus L^2(\Sigma))}+\nu 
\|\Lambda_{F,\Sigma}(\mu)\|^{2}_{\Bou(L^2(\Sigma)\oplus L^2(\Sigma))}
\|\breve g_{\Sigma}(\mu)\|^{2}_{\Bou(L^{2}(\R^{3}),L^{2}(\Sigma)\oplus L^{2}(\Sigma))}\\
&+\nu^{2}\|\Lambda_{F,\Sigma}(\mu)\|^{2}_{\Bou(L^2(\Sigma)\oplus L^2(\Sigma))}
\|g_{\Sigma}(\mu)\|^{2}_{\Bou(L^{2}(\Sigma)\oplus L^{2}(\Sigma),L_{-\sigma_{1},-\sigma_{\parallel}}^{2}(\R^{3}))}\|R_{F}(\lambda\pm i\epsilon)\|_{\Bou(L_{\sigma_{1},\sigma_{\parallel}}^{2}(\R^{3}),L_{-\sigma_{1},-\sigma_{\parallel}}^{2}(\R^{3}))}\,.
\end{align*}
\noindent
Hence inequality \eqref{MP1} is consequence of LAP for $A_{F}$ (see Proposition \ref{ThmLAPF}).\par
Finally, since  LAP implies absence of the singular continuous spectrum (see e.g.  \cite[Cor. 4.7]{MPS2}), $\sigma_{sc}(A_F)$ is empty.
\end{proof}
\begin{rem}\label{remWd22} As immediate consequence of Theorem \ref{CorWd22}, one gets 
\begin{equation*}
\si_{ac}(A_F) = \si_{ac}(A_0)=[0,+\infty)\,.
\end{equation*} 
We conjecture that there are no embedded eigenvalues for $A_F$. Hence, if this is the case,  $\sigma(A_F) = \sigma_{ac}(A_F)$ and $P_{ac}(A_F)$ is the identity. 
\end{rem}
\section{The scattering matrix\label{s:sm}}
In this section  our aim is (following the same strategy as in \cite[Sec. 4]{MP}) to use the Kato-Birman invariance principle to recover the scattering matrix $S_{F}(\lambda)$ for the scattering couple $(A_{F},A_{0})$ from the one for  $(R_{F}(\mu),R_{0}(\mu))$, $\mu\in\rho(A_{0})\cap\rho(A_{F})$. Let us denote by 
$$S_{F}:=W_{+}^{*}W_{-}:L^{2}(\R^{3})\to L^{2}(\R^{3})$$ the scattering operator corresponding to the wave operators $W_{\pm}\equiv W_{\pm}(A_{F},A_{0})$.\par
Let  
$$
{\mathcal F}_0:L^{2}(\R^{3})\to L^{2}((0,+\infty);L^{2}({\mathbb S}^{2}))\,,
$$
be the unitary which diagonalizes $A_{0}$; below we will provide an explicit representation for ${\mathcal F}_0$ in terms of generalized eigenfunctions introduced in Eq. \eqref{geneig}. 
Then, define the scattering matrix
$$
S_{F}(\lambda):L^{2}({\mathbb S}^{2})\to L^{2}({\mathbb S}^{2})\,,\qquad \lambda\in(0,+\infty)\backslash \sigma_p^+(A_F)
$$
by the relation 
$$
[{\mathcal F}_0 S_F {\mathcal F}_0^{*}u](\lambda)=S_{F}(\lambda)u(\lambda)\,,\qquad
u:(0,+\infty)\backslash \sigma_p^+(A_F)\to L^{2}({\mathbb S}^{2})\,.
$$
Let us now denote by $W_{\pm}^{\mu}$ the wave operators  for the scattering couple $(R_{0}(\mu),R_{F}(\mu))$; below we show that they exist and are complete. We denote by $S_{F}^{\mu}$ the corresponding scattering operator $S_{F}^{\mu}:=(W_{+}^{\mu})^{*}W_{-}^{\mu}$\,. As above, we then define the corresponding scattering matrix by
$$
S_{F}^{\mu}(\lambda):L^{2}({\mathbb S}^{2})\to L^{2}({\mathbb S}^{2})\,,
\qquad
[{\mathcal F}_0^{\mu}S_{F}^{\mu}({\mathcal F}_0^{\mu})^{*}u](\lambda)=S_{F}^{\mu}(\lambda)u(\lambda)\,,
$$
where $[{\mathcal F}_0^{\mu} f](\lambda):=\frac1\lambda\,[{\mathcal F}_0 f](\mu+1/\lambda)$ is the unitary which diagonalizes $R_{0}(\mu)$.\par
The forthcoming theorem provides a more explicit representation of $S_F^{\mu}(\lambda)$.

\begin{thm}\label{thmR} The wave operator $W^{\mu}_{\pm}$ exist, are complete and the corresponding scattering matrix has the representation 
$$
S_F^{\mu}(\lambda)=\uno-2\pi iL_F^{\mu}(\lambda)\Lambda_{F,\Sigma}(\mu)\big[\uno-\breve g_{\Sigma}(\mu)\big(R_{F}(\mu)-(\lambda+i0)\big)^{-1}g_{\Sigma}(\mu)\Lambda_{F,\Sigma}(\mu)\big] L_F^{\mu}(\lambda)^{*},
$$\noindent
for all $\lambda$ such that $  \mu+\frac1\lambda\in(0,+\infty) \backslash \sigma_p^+(A_F)$, and  where
$$
L_F^{\mu}(\lambda):L^{2}(\Sigma)\oplus L^{2}(\Sigma)\to L^{2}({\mathbb S}^{2})\,,
$$
$$
L_F^{\mu}(\lambda)(\phi_{F}^{\Sigma}\oplus \phi_{0}^{\Sigma}):=\frac1\lambda\,\big[{\mathcal F}_0g_{\Sigma}(\mu)(\phi_{F}^{\Sigma}\oplus \phi_{0}^{\Sigma})\big]({\mu+1/\lambda})\,.
$$
\end{thm}
\proof
At first, let us notice the relations
\begin{equation}\label{RR_A}\begin{aligned}
&\big(R_{0}(\mu)-z\big)^{-1} = - \,\frac1z \left(\frac1{z}\,R_{0}(\mu+1/{z}) + \uno\right)\,, \\ 
&\big(R_F(\mu)-z\big)^{-1} = - \,\frac1z \left(\frac1{z}\,R_{F}(\mu+1/z) + \uno\right)\,.
\end{aligned}\end{equation}
Therefore, by LAP for $A_{0}$ and $A_{F}$ (see Theorems \ref{lemLAP1} and \ref{ThmLAPF}) the limits
\begin{equation*}
\big(R_0(\mu)-(\lambda\pm i0)\big)^{-1}\!:=\lim_{\epsilon\downarrow 0}\big(R_0(\mu)-(\lambda\pm i\epsilon)\big)^{-1}\,,
\end{equation*}
with $\mu+\frac1\lambda\in(0,+\infty)$, exists in $\Bou(L^{2}_{\sigma_{1},\sigma_{\parallel}}(\R^{3}),H^{s}_{-\sigma_{1},-\sigma_{\parallel}}(\R^{3}))$ for $\sigma_{1},\sigma_{\parallel}>1/2$ and $s\in (1,3/2)$ (see Proposition \ref{PropLAP0}),
and the limits 
\begin{equation*}
\big(R_F(\mu)-(\lambda\pm i0)\big)^{-1} :=
\lim_{\epsilon\downarrow 0}\big(R_F(\mu)-(\lambda\pm i\epsilon)\big)^{-1}\,,
\end{equation*}
with $ \mu+\frac1\lambda\in(0,+\infty) \backslash \sigma_p^+(A_F)$, exist in $\Bou(L^{2}_{\sigma_{1},\sigma_{\parallel}}(\R^{3}),L^{2}_{-\sigma_{1},-\sigma_{\parallel}}(\R^{3}))$ for $\sigma_{1} > 1/2$ and $\sigma_{\parallel}>3/4$. Moreover, by Proposition \ref{PropH1RF}, Remark \ref{remR0} and Lemma \ref{lemtaFw} one gets
$$\breve g_{\Sigma}(\mu)\in\Bou\big(L_{-\sigma_{1},-\sigma_{\parallel}}^{2}(\R^{3}),{L^2(\Sigma)\oplus L^2}(\Sigma)\big)\,,$$
$$ g_{\Sigma}(\mu) = \breve g_{\Sigma}(\mu)^{*}\in\Bou\big({L^2(\Sigma)\oplus L^2(\Sigma)},L_{\sigma_{1},\sigma_{\parallel}}^{2}(\R^{3})\big)\, $$ for all $\sigma_{1},\sigma_{\parallel}\in\R$. 
Therefore the following limits exist:
\begin{equation}\label{L1_A}
\lim_{\epsilon\downarrow 0}\breve g_{\Sigma}(\mu)\big(R_0(\mu)-(\lambda\pm i\epsilon)\big)^{-1}\,;
\end{equation}
\begin{equation*}
\lim_{\epsilon\downarrow 0}\breve g_{\Sigma}(\mu)\big(R_F(\mu)-(\lambda\pm i\epsilon)\big)^{-1}\,;
\end{equation*}
\begin{equation}\label{L3_A}
\lim_{\epsilon\downarrow 0}\breve g_{\Sigma}(\mu)\big(R_F(\mu)-(\lambda\pm i\epsilon)\big)^{-1}g_{\Sigma}(\mu)\,.
\end{equation}
Let us now show that $\breve g_{\Sigma}(\mu)$ is weakly-$R^0(\mu)$ smooth, i.e. by \cite[p. 154, Lem. 2]{Y}, 
\begin{equation}\label{in1.1_A}
\sup_{0<\epsilon<1}\epsilon\,\|\breve g_{\Sigma}(\mu) (R_{0}(\mu)-(\lambda\pm i\epsilon))^{-1}\|_{\Bou({L^{2}(\Sigma),{L^2(\Sigma)\oplus L^2(\Sigma))}}}^{2}\leqs c_{\lambda}<+\infty\,,\quad\text{for a.e. $\lambda$}\,.
\end{equation}\noindent
By \eqref{RR_A}, this is consequence of
\begin{equation*}%\label{in2.2_A}
\sup_{0<\epsilon<1}\epsilon\,\|\breve g_{\Sigma}(\mu) R_{0}(\mu+1/\lambda\pm i\epsilon)\|_{\Bou({L^{2}(\Sigma),{L^2(\Sigma)\oplus L^2(\Sigma))}}}^{2}\leqs C_{\lambda}<+\infty\,,\quad\text{for a.e. $\lambda$}\,.
\end{equation*}
To prove the latter claim let us note that, by 
$$\breve g_{\Sigma}(\mu) R_{0}(z)=\tau_{\Sigma} R_0(\mu)R_{0}(z)=\tau_{\Sigma} R_0(z)R_{0}(\mu)=\breve g_{\Sigma}(z) R_{0}(\mu)=(R_{0}(\mu)g_{\Sigma}(\bar z) )^{*}\,, $$\noindent
it is consequence of Eq. \eqref{MP1}, which has been shown to hold in the proof of Theorem \ref{CorWd22}. Therefore, by the factorization of the difference $R_{F}(\mu)-R_{0}(\mu)$ provided in Theorem \ref{TS}, by the existence of the limits \eqref{L1_A}-\eqref{L3_A} and by the bound \eqref{in1.1_A}, the hypotheses in \cite[p. 178, Th. 4']{Y} are satisfied, which suffices to infer the thesis.
\endproof

By\footnote{Here we correct a minor mistake in the computation of the scattering matrix, occurring in the published version of this paper (see J. Math. Anal. Appl. {\bf 473}(1) (2019), pp. 215--257).  The sign before $i0$ in Eq.  \eqref{IP_A_2} was wrong and the adjoint in Eq. \eqref{IP_A} was missing. This  affected the statement of Corollary \ref{coroll}, specifically, the sign before $2\pi i \alpha$ in Eq. \eqref{SM_A} was wrong.  Regrettably the formula for $S_F$ in the Corrigendum J. Math. Anal. Appl. {\bf 482}(1) (2020), 123554, still contains a misprint: $M_{F,\Sigma}(\lambda^{-})$ instead of $M_{F,\Sigma}(\lambda^{+})$.} \cite[Lem. 4.2]{MP} (be aware that there the resolvent of an operator $A$ is defined as $(-A+z)^{-1}$) one has the identity 
$$
\Lambda_{F,\Sigma}(\mu) \big[\uno-\breve g_{\Sigma}(\mu)\big(R_{F}(\mu)-z\big)^{-1}g_{\Sigma}(\mu)\Lambda_{F,\Sigma}(\mu)\big]
=\Lambda_{F,\Sigma}(\mu+1/z)\,.
$$ 
This gives the existence of the limits
\begin{align}
\Lambda_{F,\Sigma}(\lambda^{\pm}):= & \lim_{\epsilon\downarrow 0}\Lambda_{F,\Sigma}(\lambda\pm i\epsilon) \nonumber \\  
= & \lim_{\epsilon\downarrow 0}
\Lambda_{F,\Sigma}(\mu)\big[\uno-\breve g_{\Sigma}(\mu) \big(R_{F}(\mu)-1/(\lambda-\mu \pm i\epsilon)\big)^{-1}g_{\Sigma}(\mu)\Lambda_{F,\Sigma}(\mu)\big] \nonumber \\ 
= & 
\Lambda_{F,\Sigma}(\mu)\big[\uno-\breve g_{\Sigma}(\mu) \big(R_{F}(\mu)-(1/(\lambda-\mu)\mp i0)\big)^{-1}g_{\Sigma}(\mu)\Lambda_{F,\Sigma}(\mu)\big]\,, \label{IP_A_2}
\end{align}
where in the latter equality we used the identity
\[\lim_{\epsilon \downarrow 0}\frac{1}{\lambda-\mu \pm i\epsilon} = \lim_{\epsilon \downarrow 0} \left(\frac1{\lambda-\mu} \mp\frac{ i\epsilon}{(\lambda-\mu)^2} + O(\epsilon^2) \right)= \frac1{\lambda-\mu} \mp i0 \,.\]
Taking into account the identities
$$
-\Lambda_{F,\Sigma}(\lambda\pm i\epsilon)\left(\frac{J}\alpha+\tau_{\Sigma}g_{\Sigma}(\lambda\pm i\epsilon)\right)=\uno=-\left(\frac{J}\alpha+\tau_{\Sigma}g_{\Sigma}(\lambda\pm i\epsilon)\right)\Lambda_{F,\Sigma}(\lambda\pm i\epsilon)\,,
$$
\[
\left(\frac{J}\alpha+\tau_{\Sigma}g_{\Sigma}(\lambda\pm i\epsilon)\right) \alpha \left(\uno+\alpha J\tau_{\Sigma}g_{\Sigma}(\lambda\pm i\epsilon)\right)^{-1}  J
= \uno 
=\alpha \big(\uno+\alpha J\tau_{\Sigma}g_{\Sigma}(\lambda\pm i\epsilon)\big)^{-1}J\left(\frac{J}\alpha+\tau_{\Sigma}g_{\Sigma}(\lambda\pm i\epsilon)\right) 
\]
and considering the limit $\epsilon\downarrow 0$, also provides the existence of the inverse 
 $(\uno+\alpha J\tau_{\Sigma}g_{\Sigma}(\lambda\pm i0))^{-1}$
and the identity 
$$
\Lambda_{F,\Sigma}(\lambda^{\pm})=-\alpha\left(\uno+\alpha J \Gt_{F,\Sigma}(\lambda^{\pm})\right)^{-1}J\,,
$$
where we set
\begin{equation}\label{Gt_00}
\Gt_{F,\Sigma}(\lambda^{\pm}):=\tau_{\Sigma}g_{\Sigma}(\lambda\pm i0)=\lim_{\epsilon\downarrow 0}\tau_{\Sigma}g_{\Sigma}(\lambda\pm i\epsilon)\,.
\end{equation}
Since, by the invariance principle (taking into account the fact that $\frac{d}{d\lambda} \frac{1}{\lambda-\mu}<0$), one has the relation (see \cite[Ch. 2, Sec. 6, Eq. (14)]{Y})
\begin{equation}\label{IP_A}
S_{F}(\lambda)=\big(S_{F}^{\mu}\big(1/(\lambda-\mu)\big)\big)^*\,,
\end{equation}
Theorem \ref{thmR} has the following 
\begin{cor}\label{coroll}
For all $  \lambda\in(0,+\infty) \backslash \sigma_p^+(A_F)$ the scattering matrix for the scattering couple $(A_{F}, A_{0})$ is given by  
\begin{equation}\label{SM_A}
S_{F}(\lambda)=\uno-2\pi i \alpha L_{F}(\lambda)(\uno+\alpha J M_{F,\Sigma}(\lambda^{+}))^{-1}JL_{F}(\lambda)^{*}\,,
\end{equation}

where $L_{F}(\lambda):L^{2}(\Sigma)\oplus L^{2}(\Sigma)\to L^{2}({\mathbb S}^{2})$ is given by
\begin{align*}
 \big(L_{F}(\lambda)(\phi_{F}^{\Sigma}\oplus \phi_{0}^{\Sigma})\big)({\boldsymbol\xi})  
= & \frac{\lambda^{{1}/{4}}}{2^{1/2}}\,\frac1{(2\pi)^{3/2}}\int_{\Sigma}d{\xp}\left(e^{-i\sqrt{\lambda}{\xi}_{1}F({\xp})} - \frac{\alpha\, e^{-i\sqrt{\lambda}\,|{\xi}_{1}|\,|F({\xp})|}}{\alpha+2i\sqrt{\lambda}\,|{\xi}_{1}|}\right)e^{-i\sqrt{\lambda}\,{\boldsymbol\xi}_{\parallel}\cdot {\xp}}\,\phi^{\Sigma}_{F}({\xp})\\
& + \,
\frac{\lambda^{{1}/{4}}}{(4\pi)^{1/2}}\,\left(1-\frac{\alpha}{\alpha+2i\sqrt{\lambda}\,|{\xi}_{1}|}\right)\,\widehat{\widetilde{\phi^{\Sigma}_{0}}}(\sqrt\lambda\,{\boldsymbol\xi}_{\parallel})\,.
\end{align*}
Here $\widehat{\ }$ denotes the Fourier transform in $L^{2}(\R^{2})$, $\widetilde{\phi^{\Sigma}_{0}}$ is the extension by zero of $\phi^{\Sigma}_{0}$ to the whole $\R^{2}$ and ${\boldsymbol\xi}\equiv({\xi}_{1},{\boldsymbol\xi}_{\parallel})$, $|{\xi}_{1}|^{2}+\|{\boldsymbol\xi}_{\parallel}\|^{2}=1$.
\end{cor}
\begin{proof} By  Theorem \ref{thmR} and by Eq. \eqref{IP_A_2}, one gets 
\[
S_{F}^{\mu}\big(1/(\lambda-\mu)\big)=\uno-2\pi iL_F^{\mu}(1/(\lambda-\mu))\Lambda_{F,\Sigma}(\lambda^-) L_F^{\mu}(1/(\lambda-\mu))^{*}.
\]
Hence, 
\begin{equation}\label{SM_A_00}
S_{F}^{\mu}\big(1/(\lambda-\mu)\big)=
\uno+2\pi i \alpha L_{F}(\lambda)(\uno+\alpha J M_{F,\Sigma}(\lambda^{-}))^{-1}JL_{F}(\lambda)^{*}\,,
\end{equation}
with 
$$
L_{F}(\lambda)(\phi_{F}^{\Sigma}\oplus \phi_{0}^{\Sigma}):=(\lambda-\mu)\big({\mathcal F}_0g_{\Sigma}(\mu)(\phi_{F}^{\Sigma}\oplus \phi_{0}^{\Sigma})\big)(\lambda)\,.
$$
By the definition  \eqref{Gt_00} it follows that 
$$  
M_{F,\Sigma}(\lambda^{\pm}) = \lim_{\varepsilon  \downarrow 0}\tau_{\Sigma}(\tau_{\Sigma}R_{0}(\lambda\mp i\varepsilon  ))^{*}\,;
$$
hence, the relation $(M_{F,\Sigma}(\lambda^{\pm}))^{*}=M_{F,\Sigma}(\lambda^{\mp})$ holds true. Taking the adjoint in Eq. \eqref{SM_A_00} and  recalling that $J^{2}=\uno$, $J=J^{*}$  one obtains Eq. \eqref{SM_A}. Then by 
$$
(\lambda-\mu)\big({\mathcal F}_0R_{0}(\mu) f\big)(\lambda)=({\mathcal F}_0 f)(\lambda)
$$
and by 
$$
\big(g_{\Sigma}(\mu)(\phi_{F}^{\Sigma}\oplus \phi_{0}^{\Sigma})\big)(\lambda)=
R_{0}(\mu)(\tau^{*}_{F,\Sigma}\phi_{F}^{\Sigma}+\tau^{*}_{0,\Sigma}\phi_{0}^{\Sigma})\,,
$$
where $\tau_{F/0,\Sigma}u:=\tau_{F/0}u|\Sigma$, one gets
$$
L_{F}(\lambda)(\phi_{F}^{\Sigma}\oplus \phi_{0}^{\Sigma})={\mathcal F}_0(\tau^{*}_{F,\Sigma}\phi_{F}^{\Sigma}+\tau^{*}_{0,\Sigma}\phi_{0}^{\Sigma})(\lambda)\,.
$$
By the structure of the generalized eigenfunctions of $A_{0}$ (see Eq. \eqref{geneig}), one gets 
$$
{\mathcal F}_0:L^{2}(\R^{3})\to L^{2}\big((0,+\infty);L^{2}({\mathbb S}^{2})\big)\,,
$$
$$ 
\big(({\mathcal F}_0 f)(\lambda)\big)({\boldsymbol\xi}):=
\frac{\lambda^{{1}/{4}}}{2^{1/2}}\,\int_{\R^{3}}\overline {\varphi}_{\sqrt\lambda\, \boldsymbol\xi }(\bold x) \,  f(\bold x)\,d\bold x\,, 
$$
where ${\varphi}_{\boldsymbol k }(\bold x)$ was defined in Eq. \eqref{geneig}, $\bold x\equiv(x^{1},{\xp})$, ${\boldsymbol\xi}\equiv({\xi}_{1},{\boldsymbol\xi}_{\parallel})$, $|{\xi}_{1}|^{2}+|{\boldsymbol\xi}_{\parallel}|^{2}=1$ and $\int_{\R^{3}}$ is to be understood as the $L^{2}$-limit of $\int_{|\x|\leqs R}$ for $R \!\nearrow +\infty$.
Therefore
\begin{align*}
 \big(L_{F}(\lambda)(\phi_{F}^{\Sigma}\oplus \phi_{0}^{\Sigma})\big)({\boldsymbol\xi}) 
= & \frac{\lambda^{{1}/{4}}}{2^{1/2}}
\int_{\Sigma}d{\xp}\left(\overline {\varphi}_{\sqrt\lambda\, \boldsymbol\xi }(F({\xp}),{\xp})\phi^{\Sigma}_{F}({\xp})+\overline {\varphi}_{\sqrt\lambda\, \boldsymbol\xi }(0,{\xp})\phi^{\Sigma}_{0}({\xp})
\right)\\
=\; &\frac{\lambda^{{1}/{4}}}{2^{1/2}}\,\frac1{(2\pi)^{3/2}}\int_{\Sigma}d{\xp}\left(e^{-i\sqrt{\lambda}{\xi}_{1}F({\xp})}-\frac{\alpha\, e^{-i\sqrt{\lambda}\,|{\xi}_{1}|\,|F({\xp})|}}{\alpha+2i\sqrt{\lambda}\,|{\xi}_{1}|}\right)e^{-i\sqrt{\lambda}\,{\boldsymbol\xi}_{\parallel}\cdot {\xp}}\phi^{\Sigma}_{F}({\xp})\\
& + \,
\frac{\lambda^{{1}/{4}}}{(4\pi)^{1/2}}\,\left(1-\frac{\alpha}{\alpha+2i\sqrt{\lambda}\,|{\xi}_{1}|}\right)\,\widehat{\widetilde{\phi^{\Sigma}_{0}}}(\sqrt\lambda\,{\boldsymbol\xi}_{\parallel})\,.\vspace{-0.5cm}
\end{align*}
\end{proof}

Let us now denote by $S_{0}(\lambda)$, $L_{0}(\lambda)$,  $\Lambda_{0,\Sigma}(\lambda^{+})$, $\Gt_{0,\Sigma}(\lambda^{+})$ the operators $S_{F}(\lambda)$, $L_{F}(\lambda)$,  $\Lambda_{F,\Sigma}(\lambda^{+})$, $\Gt_{F,\Sigma}(\lambda^{+})$ corresponding to the choice $F=0$; obviously $S_{0}(\lambda)=\uno$. By 
\begin{align*}
S_{F}(\lambda)-\uno=S_{F}(\lambda)-S_{0}(\lambda)  = &2\pi i \big(L_{F}(\lambda)-L_{0}(\lambda)\big) \Lambda_{F,\Sigma}(\lambda^{+})L_{F}(\lambda)^{*} \\ 
&  + 2\pi iL_{0}(\lambda)\Lambda_{F,\Sigma}(\lambda^{+})\big(\Gt_{F,\Sigma}(\lambda^{+})-\Gt_{0,\Sigma}(\lambda^{+})\big)\Lambda_{0,\Sigma}(\lambda^{+})L_{F}(\lambda)^{*}\\
&  + 2\pi iL_{0}(\lambda) \Lambda_{0,\Sigma}(\lambda^{+})\big(L_{F}(\lambda)^{*}-L_{0}(\lambda)^{*}\big)
\end{align*}
one gets
\begin{align*}
 \|S_{F}(\lambda)-\uno\|_{\Bou(L^{2}({\mathbb S}^{2}))}  \leqs &
2\pi \big(K_{F,\Sigma}(\lambda)+K_{0,\Sigma}(\lambda)\big)\,\|L_{F}(\lambda)-L_{0}(\lambda)\|_{\Bou({L^2(\Sigma)\oplus L^2(\Sigma)},L^{2}({\mathbb S}^{2}))}\\
&   + 2\pi K_{F,\Sigma}(\lambda)K_{0,\Sigma}(\lambda)\,\|\Gt_{F,\Sigma}(\lambda^{+})-\Gt_{0,\Sigma}(\lambda^{+})\|_{\Bou(L^2(\Sigma)\oplus L^2(\Sigma))}\,,
\end{align*}
where
$$
K_{F/0,\Sigma}(\lambda):=\|L_{F/0}(\lambda))\|_{\Bou({L^2(\Sigma)\oplus L^2(\Sigma)},L^{2}({\mathbb S}^{2}))}\|\Lambda_{F/0,\Sigma}(\lambda^{+})\|_{\Bou(L^2(\Sigma)\oplus L^2(\Sigma))}\,.
$$
The next Lemmata provide estimates on the norms of the differences $L_{F}(\lambda)-L_{0}(\lambda)$ and $\Gt_{F,\Sigma}(\lambda^{+})-\Gt_{0,\Sigma}(\lambda^{+})$ and hence on the norm of $S_{F}(\lambda)-\uno$.
\begin{lemma} For all $\lambda \in (0,+\infty)\backslash \sigma_p^+(A_F)$, there holds
$$
\|L_{F}(\lambda)-L_{0}(\lambda)\|^{2}_{\Bou({L^2(\Sigma)\oplus L^2(\Sigma)},L^{2}({\mathbb S}^{2}))}\leqs 
\frac{2 \sqrt\lambda}{\pi^{2}}\int_{\Sigma}d{\xp}\left(1-\frac{\sin\!\big(\sqrt\lambda\,F({\xp})\big)}{\sqrt\lambda\,F({\xp})}\right).
$$
\end{lemma}
\begin{proof}
Indicating with $d\sigma$ the usual surface element on the sphere ${\mathbb S}^{2}$, we have
\begin{align*}
& \big\|\big(L_{F}(\lambda)-L_{0}(\lambda)\big)(\phi_{F}^{\Sigma}\oplus \phi_{0}^{\Sigma})\big\|^{2}_{L^{2}({\mathbb S}^{2})} = \int_{{\mathbb S}^{2}}d\sigma({\boldsymbol\xi})\big|\big((L_{F}(\lambda)- L_{0}(\lambda))(\phi_{F}^{\Sigma}\oplus \phi_{0}^{\Sigma})\big)({\boldsymbol\xi})\big|^{2} \\
\leqs &\;
\frac{\lambda^{{1}/{2}}}{2(2\pi)^{3}}\int_{{\mathbb S}^{2}}d\sigma({\boldsymbol\xi})\left(\int_{\Sigma}d{\xp}\left(\left|e^{-i\sqrt{\lambda}{\xi}_{1}F({\xp})}-1\right|+
\left|e^{-i\sqrt{\lambda}\,|{\xi}_{1}|\,|F({\xp})|}-1\right|\right)\,|\phi^{\Sigma}_{F}({\xp})|\right)^{2}\\
\leqs &\;
\frac{\lambda^{{1}/{2}}}{2}\,\frac{\|\phi^{\Sigma}_{F}\|^{2}_{L^{2}(\Sigma)}}{(2\pi)^{3}}\int_{\Sigma}d{\xp}\left(\int_{{\mathbb S}^{2}}d\sigma({\boldsymbol\xi})\left(\left|e^{-i\sqrt{\lambda}{\xi}_{1}F({\xp})}-1\right|+
\left|e^{-i\sqrt{\lambda}\,|{\xi}_{1}|\,|F({\xp})|}-1\right|\right)^{2}\right)\\
\leqs &\;
{4\lambda^{{1}/{2}}}\,\frac{\|\phi^{\Sigma}_{F}\|^{2}_{L^{2}(\Sigma)}}{(2\pi)^{3}}\int_{\Sigma}d{\xp}\left(\int_{{\mathbb S}^{2}}d\sigma({\boldsymbol\xi})\Big(1-\cos\!\big(\sqrt\lambda\,{\xi}_{1}F({\xp})\big)\Big)\right)\\
= &\;
{2\lambda^{{1}/{2}}}\,\frac{\|\phi^{\Sigma}_{F}\|^{2}_{L^{2}(\Sigma)}}{\pi^{2}}\int_{\Sigma}d{\xp}\left(1-\frac{\sin \sqrt\lambda\,F({\xp})}{\sqrt\lambda\,F({\xp})}\right),
\end{align*}
\noindent
where we used the basic identity $\smallint_{0}^{\pi}dx\cos(z\cos x)\sin x=\frac2z\sin z$.
\end{proof}
\begin{rem} By the inequality $0\leq 1-\frac{\sin x}{x}\leqs \frac{|x|}{\pi}$ one gets 
$$
\|L_{F}(\lambda)-L_{0}(\lambda)\|^{2}_{\Bou({L^2(\Sigma)\oplus L^2(\Sigma)},L^{2}({\mathbb S}^{2}))}\leqs 
\frac{2 \lambda}{\pi^{3}}\;\|F\|_{L^{1}(\Sigma)}\,.
$$
\end{rem}
\begin{lemma} Let $r=\frac12+\gamma$, $0<\gamma<1$ and $s>1$. Then
$$
\|\tau_{F,\Sigma}-\tau_{0,\Sigma}\|^{2}_{\Bou(H^{r}(I)\otimes H^{s}(\Omega),L^{2}(\Sigma))}\lec \big\| |F|^{2\gamma}\big\|_{L^{1}(\Sigma)}\,.
$$\end{lemma}
\begin{proof} Using the Sobolev inequalities
$$\begin{aligned}
|\uu(x_{1})-\uu(y_{1})| &\lec \|\uu\|_{H^{r}(I)}\,|x_{1}-y_{1}|^{\gamma}\,, \\
\sup_{{\xp}\in\Omega}|\up({\xp})| &\lec\|\up\|_{H^{s}(\Omega)}\,,
\end{aligned}
$$ 
one obtains
$$
\int_{\Sigma}d{\xp}\,|\uu(F({\xp}))-\uu(0)|^{2}\,|\up({\xp})|^{2}\lec\|\uu\|^{2}_{H^{r}(I)}\,\|\up\|^{2}_{H^{s}(\Omega)\|}\int_{\Sigma}d{\xp}\,|F({\xp})|^{2\gamma}\,,
$$
which proves the thesis.
\end{proof}

\begin{lemma}\label{LemCit} There holds
\begin{align*}
&\|\Gt_{F,\Sigma}(\lambda^{+})-\Gt_{0,\Sigma}(\lambda^{+})\|^{2}_{\Bou(L^{2}(\Sigma)\oplus L^{2}(\Sigma),L^{2}(\Sigma)\oplus L^{2}(\Sigma))}\\
\lec &  \left(N^{2}_{s_{1},s_{2},\epsilon}(\lambda^{+})N^{2}_{F,\Sigma, s_{1},s_{2}}+
\left(N^{2}_{s_{1},s_{2},\epsilon}(\lambda^{+})+N^{2}_{s_{1},s_{2},\epsilon}(\lambda^{-})\right) 
N^{2}_{0,\Sigma, s_{1},s_{2}}\right)\times\\
&\times\|\tau_{F,\Sigma}-\tau_{0,\Sigma}\|^{2}_{\Bou(H^{-s_{1}+3/2-\epsilon}(I)\otimes H^{-s_{2}+2}(\Omega),L^{2}(\Sigma))}\,,
\end{align*}
where $\epsilon>0$, $\frac12<s_{1}<1-\epsilon$, $\frac12<s_{2}<\frac32$,
$$
N_{F/0,\Sigma, s_{1},s_{2}}:=\|\tau_{F/0,\Sigma}\|_{\Bou(H^{s_{1}}( I)\otimes H^{s_{2}}(\Omega),L^{2}(\Sigma))}\,,
$$
\[ N_{s_{1},s_{2},\epsilon}(\lambda^{\pm}) :=  \|(r_{I}\otimes r_{\Omega})R_{0}(\lambda^{\pm})
(r^{*}_{I}\otimes r^{*}_{\Omega})\|
_{\Bou(H^{-s_{1}}_{\overline I}(\R)\otimes H^{-s_{2}}_{\overline \Omega}(\R^{2}),
H^{-s_{1}+3/2-\epsilon}(I)\otimes H^{-s_{2}+2}(\Omega))}\,,
\]
$\Omega\subset \R^{2}$ is a open ball containing $\Sigma$, $I\subset\R$ is an open bounded interval which contains {\rm range}$(F)$ and $r_{I}:H_{\sigma_{1}}^{s_{1}}(\R)\to H^{s_{1}}(I)$, $r_{\Omega}:H_{\sigma_{\parallel}}^{s_{2}}(\R^{2})\to H^{s_{2}}(\Omega)$ are the restriction operators.
\end{lemma}
\begin{proof} One has $H^{s}(I)^{\prime}\simeq H^{-s}_{\overline I}(\R)$ and $H^{s}(\Omega)^{\prime}\simeq H^{-s}_{\overline\Omega}(\R^{2})$ for any $s\in \R$, (see, e.g., \cite[Th.s 3.14, 3.29]{McL}); we use such  identifications in the following. Then, the duals of the restriction operators  
$r^{*}_{I}:H_{\overline I}^{-s_{1}}(\R)\to 
H_{-\sigma_{1}}^{-s_{1}}(\R)$ and $r^{*}_{\Omega}:H_{\overline\Omega}^{-s_{2}}(\R^{2})\to  H_{-\sigma_{\parallel}}^{-s_{2}}(\R^{2})$ are given by the extensions by zero. \par
We know that 
$$
R^{(1)}_{0}(\lambda^{\pm})\in\Bou\big(L^{2}_{\sigma_{1}}(\R), H^{3/2-\epsilon}_{-\sigma_{1}}(\R)\big)\,,
\quad \mbox{for $\sigma_{1}> 1/2$, $0<\epsilon< 1/2$}\,,
$$
and
$$
R^{(\parallel)}_{0}(\lambda^{\pm})\in\Bou\big(L^{2}_{\sigma_{\parallel}}(\R^{2}), H^{2}_{-\sigma_{\parallel}}(\R^{2})\big)\,,\quad\mbox{$\sigma_{\parallel}> 1/2$}\,.
$$ 
Thus 
\begin{equation}\label{B1-1}
r_{I}R^{(1)}_{0}(\lambda^{\pm})r^{*}_{I}\in \Bou\big(L_{\overline I}^{2}(\R), H^{3/2-\epsilon}(I)\big)
\end{equation}
and
\begin{equation}\label{B1-2}
r_{\Omega}R^{(\parallel)}_{0}(\lambda^{\pm})r^{*}_{\Omega}\in \Bou\big(L_{\overline \Omega}^{2}(\R^{2}), H^{2}(\Omega)\big)\,.
\end{equation}
Hence, by duality, 
\begin{equation}\label{B2-1}
r_{I}R^{(1)}_{0}(\lambda^{\pm})r^{*}_{I}\in \Bou\big(H^{-3/2+\epsilon}_{\overline I}(\R),L^{2}(I)\big)
\end{equation}
and
\begin{equation}\label{B2-2}
r_{\Omega}R^{(\parallel)}_{0}(\lambda^{\pm})r^{*}_{\Omega}\in \Bou\big(H^{-2}_{\overline \Omega}(\R^{2}),L^{2}(\Omega)\big)\,.
\end{equation}
By Eq.s \eqref{B1-1} - \eqref{B2-2} and by interpolation 
({\footnote{One should keep in mind the previously mentioned isomorphisms 
$H^{s}(I)^{\prime}\simeq H^{-s}_{\overline I}(\R)$, $H^s(\Om)^{\prime} \simeq H^{-s}_{\overline \Om}(\R^2)$ (for $s \in \R$), recall the duality interpolation theorem stated in \cite[Cor. 4.5.2]{BeLo} and notice that the spaces $H^{s}(I)$, $H^{s}(\Om)$ enjoy the interpolation property (see, e.g., \cite{Tart}).}}) one gets
$$
r_{I}R^{(1)}_{0}(\lambda^{\pm})r^{*}_{I}\in \Bou(H^{-s_{1}}_{\overline I}(\R),
H^{-s_{1}+3/2-\epsilon}(I))\,,\quad \mbox{for $0\leqs s_{1}\leqs 3/2-\epsilon$}\,,
$$ 
and
$$
r_{\Omega}R^{(\parallel)}_{0}(\lambda^{\pm})r^{*}_{\Omega}\in \Bou(H^{-s_{2}}_{\overline \Omega}(\R^{2}),
H^{-s_{2}+2}(\Omega))\,,\quad \mbox{for $0\leqs s_{2}\leqs 2$}\,.
$$ 
Therefore, since $A_{0}=A_{0}^{(1)}\otimes\uno+\uno\otimes A_{0}^{(\parallel)}$, by \cite{BAD83} one has 
 $$
(r_{I}\otimes r_{\Omega})R_{0}(\lambda^{\pm})(r^{*}_{I}\otimes r^{*}_{\Omega})\in \Bou(H^{-s_{1}}_{\overline I}(\R)\otimes H^{-s_{2}}_{\overline \Omega}(\R^{2}),
H^{-s_{1}+3/2-\epsilon}(I)\otimes H^{-s_{2}+2}(\Omega))\,.
$$\noindent
Since $$(\tau_{F,\Sigma}-\tau_{0,\Sigma})R_{0}(\lambda^{\pm})\tau^{*}_{F/0,\Sigma}=(\tau_{F,\Sigma}-\tau_{0,\Sigma})(r_{I}\otimes r_{\Omega})R_{0}(\lambda^{\pm})(r^{*}_{I}\otimes r^{*}_{\Omega})\tau^{*}_{F/0,\Sigma}\,,
$$
one has, for any $\frac12<s_{1}<1-\epsilon$ and $\frac12<s_{2}<\frac32$,
\begin{align*}
&\|(\tau_{F,\Sigma}-\tau_{0,\Sigma})R_{0}(\lambda^{\pm})\tau^{*}_{F/0,\Sigma}\|_{\Bou(L^{2}(\Sigma))}\\
\leqs &\; N_{s_{1},s_{2},\epsilon}(\lambda^{\pm})\,N_{F/0,\Sigma, s_{1},s_{2}}\,\|\tau_{F,\Sigma}-\tau_{0,\Sigma}\|_{\Bou(H^{-s_{1}+3/2-\epsilon}(I)\otimes H^{-s_{2}+2}(\Omega),L^{2}(\Sigma))}\,.
\end{align*}
Using the block operator matrix notation one has
\[
 \Gt_{F,\Sigma}(\lambda^{+})-\Gt_{0,\Sigma}(\lambda^{+}) 
=  \left(\begin{matrix}\tau_{F,\Sigma}(\tau_{F,\Sigma}R_{0}(\lambda^{-}))^{*}-\tau_{0,\Sigma}(\tau_{0,\Sigma}R_{0}(\lambda^{-}))^{*}&(\tau_{F,\Sigma}-\tau_{0,\Sigma})(\tau_{0,\Sigma}R_{0}(\lambda^{-}))^{*}\\
\tau_{0,\Sigma}((\tau_{F,\Sigma}-\tau_{0,\Sigma})R_{0}(\lambda^{-}))^{*}&0
\end{matrix}\right)
\]
and so (for brevity here we omitted some norm indexes)
\begin{align*}
&\|\Gt_{F,\Sigma}(\lambda^{+})-\Gt_{0,\Sigma}(\lambda^{+})\|^{2}_{\Bou(L^{2}(\Sigma)\oplus L^{2}(\Sigma))}\\
\lec\,& \| (\tau_{F,\Sigma}-\tau_{0\Sigma})(\tau_{F,\Sigma}R_{0}(\lambda^{-}))^{*}\|^{2}+
\| \tau_{0,\Sigma}((\tau_{F,\Sigma}-\tau_{0\Sigma})R_{0}(\lambda^{-}))^{*}\|^{2} \\ 
& +
\|(\tau_{F,\Sigma}-\tau_{0\Sigma})(\tau_{0,\Sigma}R_{0}(\lambda^{+}))^{*}\|^{2}
\\
\lec\,& \| (\tau_{F,\Sigma}-\tau_{0\Sigma})R_{0}(\lambda^{+})\tau^{*}_{F,\Sigma}\|^{2}+
\|(\tau_{F,\Sigma}-\tau_{0\Sigma})R_{0}(\lambda^{-})\tau^{*}_{0,\Sigma}\|^{2} \\ 
 & + \|(\tau_{F,\Sigma}-\tau_{0\Sigma})R_{0}(\lambda^{+})\tau^{*}_{0,\Sigma}\|^{2}\\
\lec &  \left(N^{2}_{s_{1},s_{2},\epsilon}(\lambda^{+})N^{2}_{F,\Sigma, s_{1},s_{2}}+
\left(N^{2}_{s_{1},s_{2},\epsilon}(\lambda^{+})+N^{2}_{s_{1},s_{2},\epsilon}(\lambda^{-})\right) 
N^{2}_{0,\Sigma, s_{1},s_{2}}\right)\times\\
&\times\|\tau_{F,\Sigma}-\tau_{0,\Sigma}\|^{2}_{\Bou(H^{-s_{1}+3/2-\epsilon}(I)\otimes H^{-s_{2}+2}(\Omega),L^{2}(\Sigma))}\,.
\end{align*}
\end{proof}
Summing up, using the previous Lemma \ref{LemCit} with $\epsilon$ replaced by $\epsilon/4$ ($0<\epsilon<1$) and $s_{1}=1/2+\epsilon/4$, $s_{2}=1- \epsilon/4$, one gets
\begin{thm}\label{teo7} For all $  \lambda\in(0,+\infty) \backslash \sigma_p^+(A_F)$ and $0<\epsilon<1$, the following estimate holds:
\begin{align}\label{S_F} 
&\|S_{F}(\lambda)-\uno\|^{2}_{\Bou(L^{2}({\mathbb S}^{2}))}\leqs\,
c_{F,\Sigma}(\lambda)\,\|F\|_{L^1(\R^{2})}+c_{F,\Sigma,\epsilon}(\lambda)\,\|\,|F|^{1-\epsilon}\|_{L^{1}(\R^{2})}\,,
\end{align}
where 
$$
c_{F,\Sigma}(\lambda):=c\,\big(K_{F,\Sigma}(\lambda)+K_{0,\Sigma}(\lambda)\big)^{2}\,,
$$
$$
c_{F,\Sigma,\epsilon}(\lambda):=c\, K^{2}_{F,\Sigma}(\lambda)K^{2}_{0,\Sigma}(\lambda)\,
\left(N^{2}_{\epsilon}(\lambda^{+})N^{2}_{F,\Sigma, \epsilon}+
\left(N^{2}_{\epsilon}(\lambda^{+})+N^{2}_{\epsilon}(\lambda^{-})\right) 
N^{2}_{0,\Sigma, \epsilon}\right),
$$
$$N_{F/0,\Sigma, \epsilon}:=N_{F/0,\Sigma, \frac12+\frac{\epsilon}4,1-\frac{\epsilon}4}\,,\qquad N_{\epsilon}(\lambda^{\pm}):=N_{\frac12+\frac{\epsilon}4,1-\frac{\epsilon}4,\frac{\epsilon}4}(\lambda^{\pm})\,.
$$
\end{thm}
\begin{rem} Since the constants $c_{F,\Sigma}(\lambda)$ and $c_{F,\Sigma,\epsilon}(\lambda)$ are bounded and away from zero as $F\to 0$, by \eqref{S_F} one gets, for any $0<\gamma<1$,
$$
\|S_{F}(\lambda)-\uno\|^{2}_{\Bou(L^{2}({\mathbb S}^{2}))}={\mathcal O}\!\left(\int_{\R^{2}}d{\xp}|F({\xp})|^{\gamma}\right)\,.
$$ 
\end{rem}

\noindent
\textbf{Acknowledgments.} We thank the anonymous referee for the stimulating remarks.\par
This work was supported by: INdAM, Gruppo Nazio\-nale per la Fisica Matematica; ``Progetto Giovani GNFM 2017 - Dinamica quasi classica per il modello di polarone'' fostered by GNFM-INdAM; INFN, Istituto Nazionale di Fisica Nucleare.

\appendix

\section{LAP for the Laplacian plus a $\delta$-interaction in dimension one\label{app:A}}
In this section we prove Theorem \ref{lemLAP1}. This result implies Limiting Absorption Principle (LAP) for the operator $\Hu_0$,  which is formally written as the Laplacian plus a $\delta$-interaction of strength $\alpha$ in dimension one, and is rigorously   defined in Eq.s \eqref{A01_1} - \eqref{A01_2}. 

\proof[Proof of Theorem  \ref{lemLAP1}] Unless otherwise stated, throughout all the proof we implicitly understand the assumptions $\te \in (0,1/2)$ and $\su > 1/2$\,; moreover, let us arbitrarily fix $\eps_0 > 0$ and a compact subset $\KK \subset (0,+\infty)$. Then, on account of item ii) in Lemma \ref{lemHW}, the thesis is proved as soon as we are able to infer the following uniform bound for all $\eps \in (0,\eps_0)$, $\lam \in \KK$ and for some constant $c > 0$ (depending on $\eps_0$ and $\KK$, but not on $\eps$ and $\lam$):
\begin{equation}
\|I_{-\su} \Ru_0(\lam \!\pm\! i \eps)\,u\|_{H^{1+\te}(\R)}^2 \leqs c\, \|u\|_{L^2_{\su}(\R)}\;. \label{ThR01}
\end{equation}
As an example, in the sequel we proceed to evaluate the expression $\|I_{-\su} \Ru_0$ $( \lam + i \eps)\,u\|_{H^{1+\te}(\R)}^2$; altogether, the forthcoming uniform bounds \eqref{boL2}, \eqref{boH1} and \eqref{boSem} imply the corresponding version of Eq. \eqref{ThR01}. Of course, similar results can be derived also for $\|I_{-\su} \Ru_0(\lam - i \eps)\,u\|_{H^{1+\te}(\R)}^2$, ultimately yielding the thesis.

The starting point of our analysis is the following integral kernel identity, holding true for any given $u \in L^2_{\su}(\R) \subset L^2(\R)$ and for all $z \in \C \backslash [0,+\infty)$ (see \cite[p. 77, Th. 3.1.2]{AlbBook}):
 \begin{equation}
\big(\Ru_0(z)\,u\big)(x) \,= \int_\R dy \l[{i \over 2\sqrt{z}}\;e^{i\sqrt{z}\,|x-y|} - {i\al \over 2\sqrt{z}\,(\al-2i\sqrt{z})}\;e^{i \sqrt{z}\,(|x|+|y|)}\r] u(y) \;. \label{R0exp}
\end{equation}\noindent
Differentiating the above identity, by the dominated convergence theorem we obtain
 \begin{equation}
\big(\Ru_0(z)\,u\big)'(x) = -\,{1 \over 2}\int_\R dy \l[\sgn(x\!-\!y)\;e^{i\sqrt{z}\,|x-y|} - {\al\;\sgn x \over \al-2i\sqrt{z}}\;e^{i \sqrt{z}\,(|x|+|y|)}\r] u(y) \;. \label{R0pexp}
\end{equation}\noindent
Let us also remark that most of our arguments rely on the use of the elementary identity $\sqrt{\lam\!+\!i \eps} = (1/\sqrt{2})\big[ \sqrt{\!\sqrt{\lam^2\!+\!\eps^2}\!+\!\lam}\, +\, i \, \sqrt{\!\sqrt{\lam^2\!+\!\eps^2} \!-\!\lam}\,\big]$, along with some related simple estimates.

Firstly, performing few elementary manipulations, from Eq. \eqref{R0exp} we obtain
\[ \|I_{-\su} \Ru_0 (\lam\!+\!i\eps)\,u\|^2_{L^2(\R)}  = \int_{\R} dx\;w_{-\su}(x) \l|\int_\R dy \!\l[{\,e^{i\sqrt{\lam + i\eps}\,|x-y|} \over \al \!-\! 2 i \sqrt{\lam \!+\! i\eps}} + {i\al \big(e^{i \sqrt{\lam + i\eps}\,|x-y|} - e^{i \sqrt{\lam + i\eps}\,(|x|+|y|)}\big)\! \over 2\sqrt{\lam\!+\!i\eps}\,(\al-2i\sqrt{\lam\!+\!i\eps})}\r]\! u(y)\r|^2 . 
\]\noindent
By means of the Cauchy-Schwarz inequality and of the bounds $\big|e^{i\sqrt{\lam + i\eps}\,|x-y|}$ $/(\al - 2 i \sqrt{\lam \!+\! i\eps})\big| \leqs 1/\sqrt{\al^2 \!+\! 4\lam}\,$, $\big|i\al \big(e^{i \sqrt{\lam + i\eps}\,|x-y|} -  e^{i \sqrt{\lam + i\eps}\,(|x|+|y|)}\big) /(2\sqrt{\lam\!+\!i\eps}$ $(\al\!-\!2i\sqrt{\lam\!+\!i\eps}))\big| \leqs \al/\sqrt{\lam\,(\al^2\! +\! 4\lam)}$, the above identity implies the following relations
({\footnote{The first estimate in Eq. \eqref{R0pexp} gives an indication of the fact that the arguments employed here cannot be extended to include the boundary case $\lam = 0$; this is the underlying reason behind our assumption $\lam \in (0,+\infty)$.}}):
\begin{equation}\|I_{-\su} \Ru_0 (\lam \!+\! i\eps)\,u\|^2_{L^2(\R)}  \leqs  {2\,(\al^2\! + \lam) \over \lam\,(\al^2\! + 4\lam)} 
\l(\int_{\R}dx\;w_{-\su}(x)\r)^{\!\!2} \l(\int_\R dy\;w_{\su}(y)\,|u(y)|^2\r) \lec\, \|u\|_{L_{\su}^2(\R)}^2 \;.\label{boL2}
\end{equation}

Next, let us consider the expression \eqref{R0pexp} for $(\Ru_0(z)\,u)'(x)$; bearing in mind the previous estimate \eqref{boL2} and using again the Cauchy-Schwarz inequality, along with the elementary relations $|e^{i \sqrt{\lam + i\eps}|x-y|}| \!\leqs\! 1$, $$\big|\al\,e^{i \sqrt{\lam + i\eps}\,(|x|+|y|)}/(\al - 2i\sqrt{\lam + i\eps})\big| \leqs \al/\sqrt{\al^2\!+ 4 \lam}$$ and $|(w_{-\su}^{1/2})'(x)| \leqs (\su/2)\,w_{-\su}^{1/2}(x)$, by similar computations we get
\begin{equation}
\begin{aligned}
 & \big\|\big(I_{-\su} \Ru_0\!(\lam \!+\! i\eps) u\big)'\big\|^2_{L^2(\R)}  \\ 
\leqs &\,
{\al^2\! + 2 \lam \over \al^2\! + 4\lam}\! \l[1\! + {\su^2(\al^2\!+\lam) \over 2\lam(\al^2\!+2\lam)}\r]\!
\l(\int_{\R}dx\, w_{-\su}(x)\r)^{\!\!2}\! \l(\int_\R dy\,  w_{\su}(y)\,|u(y)|^2\r) \lec
\|u\|_{L_{\su}^2(\R)}^2 . \label{boH1}
\end{aligned}
\end{equation}\noindent

Let us now pass to the evaluation of the following seminorm, corresponding to the inner product defined as in Eq. \eqref{innerfra}:
\begin{equation}
\big|(I_{-\su} \Ru_0 (\lam \!+\! i\eps)\,u)'\big|^2_{\te}\, := \SIm\! + \SIp \;, \label{sem}
\end{equation}
\[
\SI^{(\lessgtr)} := \int_{\{|x-y| \,\lessgtr\, 1\}}\!\!\! dx\,dy\;{\big|(I_{-\su} \Ru_0 (\lam \!+\! i\eps)\,u)'(x) - (I_{-\su} \Ru_0(\lam \!+\! i\eps)\,u)'(y)\big|^2 \over |x-y|^{1 + 2\te}} \; . 
\]
On the one hand, taking into account Eq. \eqref{boH1} we readily obtain 
\begin{equation}\begin{aligned}
\SIp & \lec \int_{\R} dx\;\big|(I_{-\su} \Ru_0\!(\lam + i\eps)\,u)'(x)\big|^2\! \int_{\{|x-y| \,>\, 1\}}\!\!\!dy\;{1 \over |x-y|^{1 + 2\te}}\\
 & \lec  \|(I_{-\su} \Ru_0 (\lam \!+\! i\eps)\,u)'\|^2_{L^2(\R)} \lec \|u\|_{L_{\su}^2(\R)}^2 \;. \label{sem1}
\end{aligned}
\end{equation}
On the other hand, the derivation of a uniform bound for $\SIm$ is  less straightforward. To attain such a bound let us first point out the following relation, which can be readily inferred by addition and subtraction of identical terms and by elementary symmetry arguments:
\begin{equation}
\SIm \,\lec\, \SJ^{(1)} + \SJ^{(2)} + \SJ^{(3)}\;; \label{sem2a} \vspace{0.1cm}
\end{equation}
\[\begin{aligned}
 \SJ^{(1)} \,:= & \int_{\{|x-y| \,<\, 1\}}\!\!\! dx\,dy \,
w_{\su}(y)\, {\big|w_{-\su}^{1/2}(x)\! - \!w_{-\su}^{1/2}(y)\big|^2\! \over |x-y|^{1 + 2\te}}\; 
\big|(I_{-\su} \Ru_0\!(\lam\!+\!i\eps)\,u)'(x)\big|^2\,   \\
+&  \int_{\{|x-y| \,<\, 1\}}\!\!\! dx\,dy \Bigg(\!w_{\su}(x)\,{\big|(w_{-\su}^{1/2})'(x)\! - \!(w_{-\su}^{1/2})'(y)\big|^2\! \over |x-y|^{1 + 2\te}} \\
+ & \big(w_{\su}(x)\! + \!w_{\su}(y)\big)\, w_{\su}(x)\,|(w_{-\su}^{1/2})'(x)|^2\, {\big|w_{-\su}^{1/2}(x)\! - \!w_{-\su}^{1/2}(y)\big|^2\! \over |x-y|^{1 + 2\te}}\,\Bigg)
\big|(I_{-\su} \Ru_0\!(\lam\!+\!i\eps)\,u)(x)\big|^2 
\end{aligned}\]
\noindent
\[\begin{aligned}
\SJ^{(2)} \,:= & \int_{\{|x-y| \,<\, 1\}}\!\!\! dx\,dy\; w_{\su}(x)\; |(w_{-\su}^{1/2})'(y)|^2 {\big|(I_{-\su} \Ru_0\!(\lam\!+\!i\eps)\,u)(x) - (I_{-\su} \Ru_0\!(\lam\!+\!i\eps)\,u)(y)\big|^2 \over |x-y|^{1 + 2\te}}\;,  \vspace{0.1cm}
\end{aligned}\]\noindent
\[ 
\SJ^{(3)} \,:= \int_{\{|x-y| \,<\, 1\}}\!\!\! dx\,dy\; w_{-\su}(x)\; {\big|(\Ru_0\!(\lam\!+\!i\eps)\,u)'(x) - (\Ru_0\!(\lam\!+\!i\eps)\,u)'(y)\big|^2 \over |x-y|^{1 + 2\te}}\;. 
\]
Concerning the term $\SJ^{(1)}$, it can be checked by direct inspection that the integrals over $y \in (x-1,x+1)$ of the expressions involving the weights $w_{\pm \su}$ (and their derivatives) are uniformly bounded for $x \in \R$.

Let us give a few more details about this statement. First of all, it should be noticed that $w_{\su}(x)\,|(w_{-\su}^{1/2})'(x)|^2$ $= \su^2 x^2/(1\!+\!x^2)^2 \leqs \su^2/4$ for all $x \!\in\! \R$. Furthermore, starting from the elementary identity $w_{\su}(y) - w_{\su}(x) = 2 \su \int_{0}^{y-x}\! dt\, (x+t)\,(1+(x+t)^2)^{\su-1}$, for any given $\su \!\in\! \R$ and for all $x,y \!\in \!\R$ with $|x-y| \!<\! 1$ one gets 
\[
\begin{aligned}
 |w_{\su}(y) - w_{\su}(x)|  \leqs &  2\su \left(\!\sup_{x \in \R,\,y \in (x-1,x+1)}\left|\!\int_0^{y-x}\! dt\, {(x\!+\!t)\,(1\!+\!(x\!+\!t)^2)^{\su-1}\!/(1\!+\!x^2)^{\su}}\right|\right) w_{\su}(x)  \\ 
 \lec &  w_{\su}(x)\,|x-y|;
\end{aligned}
\]
\linebreak for $|x-y| < 1$, the latter relation yields in particular $w_{\su}(y) \leqs w_{\su}(x) + |w_{\su}(y) - w_{\su}(x)| \lec w_{\su}(x)$ and $|w_{-\su}^{1/2}(x) - w_{-\su}^{1/2}(y)| = |w_{-\su/2}(x) - w_{-\su/2}(y)|$ $ \lec w_{-\su}^{1/2}(x)\;|x-y|$\,. In a similar way, it can even be shown that $|(w_{-\su}^{1/2})'(x) - (w_{-\su}^{1/2})'(y)| \lec w_{-\su}^{1/2}(x)\;|x-y|$. Taking into account the previously mentioned facts, it is easy to infer that
\[ 
\sup_{x \in \R} \int_{\{|x-y| \,<\, 1\}}\!\!\! dy\;w_{\su}(y)\, {\big|w_{-\su}^{1/2}(x)\! - \!w_{-\su}^{1/2}(y)\big|^2\! \over |x-y|^{1 + 2\te}} < + \infty\;, 
\]
\[
\sup_{x \in \R} \int_{\{|x-y| \,<\, 1\}}\!\!\! dy\; w_{\su}(x)\,{\big|(w_{-\su}^{1/2})'(x)\! - \!(w_{-\su}^{1/2})'(y)\big|^2\! \over |x-y|^{1 + 2\te}} < +\infty\;,
\]
and
\[
\sup_{x \in \R} \int_{\{|x-y| \,<\, 1\}}\!\!\! dy \big(w_{\su}(x)\! + \!w_{\su}(y)\big)\, w_{\su}(x)\,|(w_{-\su}^{1/2})'(x)|^2\, {\big|w_{-\su}^{1/2}(x)\! - \!w_{-\su}^{1/2}(y)\big|^2\! \over |x-y|^{1 + 2\te}} < +\infty. 
\]\noindent

Then, recalling the bounds derived previously in Eq.s \eqref{boL2} and  \eqref{boH1}, we readily obtain $\SJ^{(1)} \lec \|u\|_{L_{\su}^2(\R)}^2$.\\
As for the expression $\SJ^{(2)}$, computations similar to those described above grant that $w_{\su}(x)\,|(w_{-\su}^{1/2})'(y)|$ is uniformly bounded for all $x,y \in \R$ with $|x-y|< 1$; in addition, by the Cauchy-Schwarz inequality we have
\[\begin{aligned}
 \big|(I_{-\su} \Ru_0(\lam\!+\!i\eps)\,u)(x) - (I_{-\su} \Ru_0(\lam\!+\!i\eps)\,u)(y)\big|  = &  \l|\int_{y}^x\!dt\; (I_{-\su} \Ru_0\!(\lam\!+\!i\eps)\,u)'(t)\r|  \\ 
\leqs & \big\|\big(I_{-\su} \Ru_0(\lam \!+\! i\eps)\,u\big)'\big\|^2_{L^2(\R)} \;|x-y| \;. 
\end{aligned}\]
The above arguments, along with the estimate \eqref{boH1}, allows us to infer that $\SJ^{(2)}\lec \|u\|_{L_{\su}^2(\R)}^2$.\\
In order to derive a uniform bound for $\SJ^{(3)}$, let us indicate with $\Te$ the Heaviside step function; this is such that $\Te(t) = 1$ for $t \geqs 0$, $\Te(t) = 0$ for $t < 0$ and $\Te(t) + \Te(-t) = 1$ almost everywhere on $\R$. Then, using the explicit expression \eqref{R0pexp} for $(\Ru_0 (\lam\!+\!i\eps)\,u)'$, by a carefully devised procedure of addition and subtraction of identical terms and by triangular inequality (plus a number of elementary manipulations) we obtain
\begin{equation}
\SJ^{(3)} \lec J_{1} + J_{2} + J_{3}\;; \label{Im}
\end{equation}
{\small
\[
J_1 := \int_{\{|x-y| < 1\}}\!\!\! dx\,dy\;{w_{-\su}(x)\;\Te(x\!-\!y) \over |x-y|^{1+ 2 \te}} \bigg[\,\l|\int_{y}^{x}\! dt\;e^{i \sqrt{\lam + i\eps}\,(x-t)} u(t)\r|^2\!   +\, \l|\big(e^{i \sqrt{\lam + i\eps}\,(x-y)}-1 \big)\! \int_{-\infty}^{y}\! dt\;e^{i \sqrt{\lam + i\eps}\,(y-t)} u(t)\r|^2\, \bigg] , 
\]
\[
J_2 := \int_{\{|x-y| < 1\}}\!\!\! dx\,dy\;{w_{-\su}(x)\;\Te(x\!-\!y) \over |x-y|^{1+ 2 \te}} \bigg[\,\l|\int_{y}^{x}\! dt \;e^{i \sqrt{\lam + i\eps}\,(t-y)} u(t)\r|^2\!  + \l|\big(e^{i \sqrt{\lam + i\eps}\,(x-y)}-1 \big)\! \int_x^{+\infty}\!\! dt\;e^{i \sqrt{\lam + i\eps}\,(t-x)} u(t)\r|^2\, \bigg] , 
\]
\[
J_3 := \int_{\{|x-y| < 1\}}\!\!\! dx\,dy\;{w_{-\su}(x) \over |x-y|^{1+ 2 \te}}\, \l|{\al \over \al - 2i\sqrt{\lam \!+\! i \eps}}\r|^2\, \l|\,\int_\R dt\;e^{i \sqrt{\lam+ i\eps}\,|t|}\,u(t)\,\r|^2 \l|(\sgn x)\,e^{i \sqrt{\lam+i\eps}\,|x|} - (\sgn y)\,e^{i \sqrt{\lam+i\eps}\,|y|} \r|^2 .
\]
}\noindent
The terms $J_1,J_2$ can be treated similarly; as an example, let us consider $J_1$. Since $|e^{i \sqrt{\lam+ i\eps}\,(x-t)}| \leqs 1$ for $y\!<\!t\!<\!x$, by the Cauchy-Schwarz inequality we have $\big|\!\int_{y}^{x} dt\,e^{i \sqrt{\lam+ i\eps}\,(x-t)} u(t)\big|^2 \!\leqs |x - y|\,\|u\|_{L^2_{\su}(\R)}^2$. Again by the Cauchy-Schwarz inequality, we get 
\[ \begin{aligned}
\bigg|\!\int_{-\infty}^{y}\! dt\,e^{i \sqrt{\lam + i\eps}\,(y-t)} u(t)\bigg|^2\! & \leqs \int_{-\infty}^{y}\! dt\,w_{-\su}(t)\, e^{- \sqrt{2}\,(y-t) \sqrt{\!\sqrt{\lam^2+\eps^2}-\lam}}  \|u\|_{L_{\su}^2(\R)}^2 \\
& \lec \|u\|_{L_{\su}^2(\R)}^2\;;
\end{aligned}\]
 moreover, we have $|e^{i \sqrt{\lam + i\eps}\,(x-y)}\!-\!1|\! \lec |x - y|\,$. The above remarks allow us to infer that
\begin{equation}
J_1 \,\lec\, \l(\int_{\R} dx\; w_{-\su}(x)\r)\l(\int_{\{|x-y| < 1\}}\!\!\! dy\;{1 \over |x\!-\!y|^{2 \te}}\r) \|u\|_{L^2_{\su}(\R)}^2 \,\lec\, \|u\|_{L^2_{\su}(\R)}^2 \,. \label{J1}
\end{equation}
Finally, let us pass to the term $J_3$. Estimates similar to those employed before yield $|\al/(\al \!-\! 2i\sqrt{\lam \!+\! i \eps}\,)| \!\leqs\! \al/\sqrt{\al^2\! + 4\lam}$ and $\big|\!\int_\R dt\,e^{i \sqrt{\lam+ i\eps}\,|t|}u(t)\big|^2\! \leqs\! \big(\!\int_{\R} dt\,w_{-\su}(t)\,e^{- \sqrt{2}\,|t| \sqrt{\!\sqrt{\lam^2+\eps^2}-\lam}}\,\big)\, \|u\|_{L_{\su}^2(\R)}^2\!\lec \|u\|_{L_{\su}^2(\R)}^2\,$. Moreover, indicating again with $\Te$ the Heaviside step function, by simple symmetry considerations (and related changes of integration variables) we have
\[ \begin{aligned}
&\int_{\{|x-y| < 1\}}\!\!\! dx\,dy\;w_{-\su}(x)\,{\big|(\sgn x)\,e^{i \sqrt{\lam+i\eps}\,|x|} - (\sgn y)\,e^{i \sqrt{\lam+i\eps}\,|y|} \big|^2 \over |x-y|^{1+ 2 \te}}\\
& = 2 \int_{\{|x-y| < 1\}}\!\!\!\!\!\!\!\!\! dx\,dy\;w_{-\su}(x)\,\Te(x)\,\Te(y)\,{\big|e^{i\,x\sqrt{\lam+i\eps}} - e^{i\,y\,\sqrt{\lam+i\eps}} \big|^2 \over |x-y|^{1+ 2 \te}}\,\\ 
& \quad +\, 2 \int_{\{x+y < 1\}}\!\!\!\!\!\!\!\!\! dx\,dy\;w_{-\su}(x)\,\Te(x)\,\Te(y)\,{\big|e^{i\,x\,\sqrt{\lam+i\eps}} + e^{ i\,y\,\sqrt{\lam+i\eps}} \big|^2 \over |x+y|^{1+ 2 \te}}\;.
\end{aligned}\]
The first addendum in the second line of the above equation can be easily proved to be finite noting that $|e^{i\,x\sqrt{\lam+i\eps}} - e^{i\,y\,\sqrt{\lam+i\eps}}| \lec |x-y|\,$; on the other hand, the finiteness of the second addendum can be inferred noting that $|e^{i\,x\sqrt{\lam+i\eps}} + e^{i\,y\,\sqrt{\lam+i\eps}}| \leqs 2$ (for all $x,y > 0$) and introducing a system of polar coordinates in the quadrant $x,y > 0$. 

More precisely, let us put $x = \rho \cos\te$, $y = \rho \sin \te$ for $\rho \in (0,+\infty)$, $\te \in (0,\pi/2)$. Then, noting that $\cos\te + \sin\te \geqs 1$ for $\te \in (0,\pi/2)$, we obtain the following for all $\su > 1/2$ and $\te \in (0,1/2)$:
\[\begin{aligned}
\int_{\{x+y < 1\}}\!\!\!\!\! dx\,dy\;{w_{-\su}(x)\,\Te(x)\,\Te(y) \over |x+y|^{1+ 2 \eta}}   = & \int_{0}^{\pi/2}\!\!\!\! d\te\! \int_{\{\rho(\cos\te + \sin\te) < 1\}}\!\!\!\!\!\! d\rho\;{1 \over \big(1\!+\!\rho^2\! + \!\rho^4 \sin^2\!\te \cos^2\!\te\big)^{\su} \rho^{2\eta}\, |\cos\te + \sin\te|^{1+ 2 \eta}}\, \\ 
 \leqs & {\pi \over 2}\! \int_{0}^{1}\!\! d\rho\;{1 \over \big(1\!+\!\rho^2\big)^{\!\su} \rho^{2\eta}} < + \infty\;.
\end{aligned}
\]
The above arguments show that $J_3 \lec \|u\|_{L^2_{\su}(\R)}^2$, which along with Eq.s \eqref{Im} and \eqref{J1} yields $\SJ^{(3)}\lec \|u\|_{L_{\su}^2(\R)}^2$.

In view of the previously described results, Eq.s \eqref{sem}, \eqref{sem1} and \eqref{sem2a} allow us to infer that
\begin{equation}
\big|\big(I_{-\su} \Ru_0 (\lam \!+\! i\eps)\,u\big)'\big|^2_{\te} \lec \|u\|_{L^2_{\su}(\R)}^2\,. \label{boSem}
\end{equation}
Summing up, Eq.s \eqref{boL2}, \eqref{boH1} and \eqref{boSem} (together with the definition of the norm on $H^{1+\eta}(\R)$ descending from Eq. \eqref{innerfra}) imply the claim \eqref{ThR01} stated at the beginning of the present proof; as already mentioned, this suffices to prove the thesis.
\endproof

\end{document}